\documentclass[a4paper,12pt]{article}

\RequirePackage{amsthm,amsmath,amsfonts,amssymb}
\RequirePackage[authoryear]{natbib}
\RequirePackage{xr}
\RequirePackage[OT1]{fontenc}
\RequirePackage{dsfont,bm}
\RequirePackage{algorithm}
\RequirePackage{algpseudocode}
\RequirePackage{graphicx} 
\RequirePackage[section]{placeins}
\RequirePackage{color}
\RequirePackage{enumerate}
\usepackage[shortlabels]{enumitem}
\usepackage{multirow}
\usepackage[english]{babel} 
\usepackage{xcolor}
\usepackage{bbm}
\usepackage{mathtools}
\usepackage{upgreek}
\usepackage[textfont=it]{caption}
\usepackage{xfrac}
\usepackage{verbatim}
\usepackage{url}
\usepackage{mathtools}
\usepackage{lmodern}
\usepackage{breqn}
\usepackage{colortbl}
\usepackage{etoolbox}
\usepackage{anyfontsize}
\usepackage{subcaption}
\usepackage[multiple]{footmisc}
\usepackage{thmtools}
\usepackage{geometry}
\RequirePackage[colorlinks,citecolor=blue,urlcolor=blue,linkcolor=blue]{hyperref}

\newcommand{\Xmp}{(X_k)_{k\in\mathbb{N}_0}}
\newcommand{\set}{\{0,1\}^\ell }

\newcommand{\dint}{\,\mathrm{d}}

\newcommand{\R}{{\mathbb R}}

\newcommand{\N}{{\mathbb N}}

\renewcommand\thmcontinues[1]{Continued}
\makeatletter
\renewcommand\@makefnmark{\hbox{{\rm\textsuperscript \@thefnmark }}}
\def\@fnsymbol#1{\ensuremath{\ifcase#1\or A\or *\or B\or C\or D\or E\or F\or G
    \or \dagger \else\@ctrerr\fi}}
\makeatother

\newtheorem{thm}{Theorem}[section]
\newtheorem{lem}[thm]{Lemma}

\newtheorem{prop}[thm]{Proposition}

\newtheorem{rem}[thm]{Remark}
\newtheorem{example}[thm]{Example}

\newtheorem{defi}[thm]{Definition}

\begin{document}
\pagestyle{plain}
\title{Analyzing cross-talk between superimposed signals: Vector norm dependent hidden Markov models and applications to ion channels}
\hypersetup{pdftitle={VND Markov chains}}

\author{Laura Jula Vanegas\thanks{Institute for Mathematical Stochastics, Georg-August-Universit\"at G\"ottingen, Germany}\phantom{\footnotesize 1\,\,}\thanks{The authors contributed equally.},
  Benjamin Eltzner\thanks{Max Planck Institute for Multidisciplinary Sciences, G\"ottingen, Germany}\phantom{\footnotesize 1\,\,}\footnotemark[2],
  Daniel Rudolf\thanks{Universit\"at Passau, Germany}\phantom{\footnotesize 1\,}\thanks{Felix-Bernstein-Institute for Mathematical Statistics in the Biosciences, G\"ottingen, Germany}\phantom{\footnotesize 1\,\,}\footnotemark[2], \\
  Miroslav Dura\thanks{Cellular Biophysics and Translational Cardiology Section, Heart Research Center G\"ottingen, Department of Cardiology \& Pneumology, University Medical Center G\"ottingen, Germany}\phantom{\footnotesize 1\,}\thanks{DZHK (German Centre for Cardiovascular Research), partner site G\"ottingen, Germany},
  Stephan E. Lehnart\footnotemark[7]\phantom{\footnotesize 1\,}\thanks{DFG Cluster of Excellence ``Multiscale Bioimaging: from molecular Machines to Networks of excitable cells'', University Medical Center G\"ottingen, Germany}, and
  Axel Munk\footnotemark[1]\phantom{\footnotesize 1\,}\footnotemark[5]\phantom{\footnotesize 1\,}\footnotemark[8]\phantom{\footnotesize 1\,\,}\thanks{Corresponding author, e-mail: munk@math.uni-goettingen.de}}

\date{\today}

\maketitle

\begin{abstract}
  We propose and investigate a hidden Markov model (HMM) for the analysis of dependent, aggregated, superimposed two-state signal recordings. A major motivation for this work is that often these signals cannot be observed individually but only their superposition.
  Among others, such models are in high demand for the understanding of cross-talk between ion channels, where each single channel
  cannot be measured separately. As an essential building block, we introduce a parameterized vector norm dependent Markov chain model and characterize it in terms of permutation invariance as well as conditional independence. 
  This building block leads to a hidden Markov chain sum process which can be used for analyzing the dependence structure of superimposed two-state signal observations within an HMM. 
  Notably, the model parameters of the vector norm dependent Markov chain are uniquely determined by the parameters of the sum process and are therefore identifiable. 
  We provide algorithms to estimate the parameters, discuss model selection and apply our methodology to real-world ion channel data from the heart muscle, where we show competitive gating.
\end{abstract}





\section{Introduction}

\subsection{Motivation}
\emph{Hidden Markov models} (HMMs) were introduced in the late 1960s \citep{baum1966,baum1970} and have since been widely adopted, see \citep{cappe2005,westhead2017hidden,zucchini2017hidden} for recent monographs. HMMs can be used to model signals that stem from an underlying, not directly observable, Markov chain and are nowadays well established tools in a variety of disciplines including information science \citep{gales2008}, biology \citep{krogh2001, chen2016} and medicine \citep{manogaran2018}. In particular, HMMs serve as a standard modeling tool in physiology for the analysis of ion channel recordings, see e.g. \citep{ball1992,becker1994,sakmann1995b,gunst2001,venkataramanan2002,khan2005}.

Whereas classical theory is mainly concerned with univariate scenarios, more recently, much progress has been made in the case where the Markov chain is multivariate and exhibits dependencies between its components. Such models usually rely on the assumption that one has access to observations of all single components and have been proven useful when analyzing smartphone sensor data of many sources \citep{vanderkamp2017}, disease interaction in medical research \citep{sherlock2013}, or for classification tasks in computer vision \citep{brand1997} to mention a few applications. An asymptotic analysis of multivariate HMMs can be found, e.g., in \citep{bielecki2013} and computational aspects are discussed in \citep{touloupou2020}.

However, only little methodology is available for the case when the signal cannot be marginally observed but only a superimposed version is available. This appears somewhat surprising, as the modeling, recovery and analysis of superimposed Markovian signals is in high demand for various applications, e.g., ion channel investigations \citep{chung2007}, super-resolution microscopy \citep{staudt2020}, or magnetotelluric data assessment \citep{neukirch2019}. Besides the masking effect from superposition, such analysis is hindered by ``crosstalk'' between these signals, i.e., by its statistical dependency. Nevertheless, in many applications the understanding of this ``crosstalk'' is actually the primary aim of the data analysis.

Therefore, in this paper we develop and characterize a novel Markov chain model allowing crosstalk between single two-state signals. This will be employed in an HMM and we provide statistical methodology for its analysis. While our methodology is applicable in any situation where superpositions of general two-state Markov-systems are observed we focus for illustrative purposes on a challenging ion channel application and show the advantages of our approach in that scenario. 

\subsection{Ion channels ensembles}\label{subsec: intro_ion_channels}
Ion channels are large protein complexes in the cell membranes of living organisms that control the flux of charged ions into and out of the cell. Ion concentrations in cells are crucial for key functionality of cells like signal transmission in nerve cells and contraction of muscle cells \citep{chung2007}. Therefore, understanding the conductance properties of ion channels is a major endeavor in physiology and of great medical importance. Fundamental to this is the patch-clamp technique, which allows to measure ion channel conductivity of single channels. The development of artificial lipid bilayers has facilitated the exclusion of interfering environmental factors, see~\citep{sakmann1995b}. In addition, current investigation of automatized patch-clamp-like techniques can lead to faster data collection and fully automated data analysis is in high demand, see \citep{perkel2010}. As mentioned above, single channel modeling is often done via HMMs, but more recently also non-parametric change point regression methods have been developed as a flexible and computationally efficient alternative, see e.g.~\citep{gnanasambandam2017,pein2018a,Bartsch2019,PEM2021,jula2021multiscale}.

However, isolating experimentally a single ion channel is not always possible or desirable. Measuring conductivity of multiple ion channels simultaneously allows one to simplify experimental design and enables the study of channel interactions. Moreover, having a reliable model for multiple channels can lead to important biological insight as observed in \citep{mirams2011}. While non-parametric change point regression methods for single channel analysis do not provide enough structure to infer properties of single channels from total conductivity of an ensemble of channels we will see that the proposed HMM allows to recover channel dependencies from superpositions by encoding interactions in the transition matrices. The simplest case occurs for independent channels \citep{dabrowski1992}, which, however, is not fulfilled in many applications, see, e.g., \citep{keleshian2000a} and the data analyzed in the present paper.

Ion channels can open and close, a process called \textit{gating}, in order to control the flux of charged ions across the membrane of the cell or intracellular organelles. The electrical current due to migrating ions is measured as a function of time, see Figure~\ref{fig:intro} for a data set. In Section \ref{sec:ion-channels}, we investigate a time series of current measurements on a synthetic lipid bilayer with multiple Ryanodine Receptor type 2 (RyR2) ion channels. Such channels are primarily found in cardiac muscle cells and neurons, since these receptors are important in controlling intracellular Ca${}^{2+}$ release from the endoplasmatic reticulum during cardiac excitation-contraction coupling. Moreover, genetic and proteomic defects in RyR2 lead to abnormally increased resting Ca${}^{2+}$ release, causing cardiac arrhythmia and contractile dysfunction \citep{taur2005,salvage2019}.

In the present experiment, the Ca$^{2+}$ concentration on the cis side is low, namely 150 nM, while the concentration on the trans side is much higher, namely 5 mM and 10 mM, respectively. No external voltage was applied, so the measured currents are purely an effect of ion concentration differences between the two sides of the membrane. Experiments were performed in the Lehnart Lab of the Cellular Biophysics and Translational Cardiology Section in the Heart Research Center G\"ottingen (HRCG). A central question arising from this study is whether RyR2 channels act independently, cooperatively or competitively. The detailed investigation of this is a major motivation of this paper.

For recovering the dependency from superpositions when the channels interact, the state of the art model and method was developed by \cite{chung1996}, which we refer to as CK model, see Subsection~\ref{subsec: CK-model} for details. Unfortunately it relies on a simplifying assumption that may be very restrictive for application purposes and is demonstrated not to be satisfied for the RyR2 data, allowing only clustered dependency, i.e., a specific form to incorporate a higher probability for the channels to be in the same state.

\subsection{Contribution of this paper}

The core of this paper is a novel Markov model called
vector norm dependent (VND) model for ensembles of coupled two-state Markov chains. In the following, states will be identified by 
`$0$' and `$1$'. Within the ion channel context `$0$' refers to a channel being `closed' and `$1$' to a channel being `open'.
The VND model, introduced in Definition~\ref{def:vnd}, contains as a special case the model of uncoupled signals, \cite{dabrowski1992}, where all signals within the superposition always act independently. However, depending on parameter values, the probability for signals to be `$0$' or `$1$' can depend on the previous number of signals being `$1$' in a positively or negatively correlated fashion. 
The model emerges from two easily interpretable properties, namely \emph{permutation invariance},
see Definition~\ref{def:perm_inv}, and \emph{conditional independence}, see Definition~\ref{def:cond_ind}, as shown in Theorem~\ref{thm:vnd-char}. This is particularly significant, since it gives the practitioner a guideline of characteristics that indicate in which scenarios the model is generally applicable.
Furthermore, by construction, the VND model satisfies the so-called lumping property, see Definition~\ref{def:lump} below. This is fundamental since, if it holds, the superimposed signal process, which we just call sum process, is again Markovian and we can use HMM techniques to estimate the corresponding transition matrix from noisy observations simplifying data analysis and interpretation significantly.
The VND model is, on the one hand, sufficiently flexible to describe a wide range of behaviors, such as competitive\footnote{Competitive dependencies increase the probability of signals to be in different states.} or cooperative dependencies, see Definition~\ref{def:coop-comp} below, with assumptions that fit well to the application. On the other hand, it is specific enough to allow for estimation of the parameters from superimposed data. As a consequence, we show in Theorem~\ref{thm:params-unique} that the VND model allows fully reconstructing the transition matrix of the underlying vector Markov chain from the transition matrix of the Markovian sum process.

In  Section~\ref{sec:ion-channels}, the VND model is illustrated in action: We show a competitive dependency, 
unnoticed by the widely used CK model, 
of RyR2 ion channels, which play an important role in cardiac muscle cells. This is supported by a BIC-type model selection in our data applications in Section~\ref{subsec:det_suit_model}, which is investigated in more detail in Section~\ref{subsec: est_num_chan}. There, we consider the question whether the number of channels in the membrane can be reliably determined even if the data shows at most a small fraction of the total number of channels open at the same time. To this end we perform a simulation study using BIC-type criteria and cross validation based model selection. It turns out that the number of channels determined by these criteria typically comes close to the maximum number of active channels visible in the data set, which is often a conservative estimate of the number of channels. In Section~\ref{subsec:robust}, we show that the findings of cooperative or competitive gating are robust to such an underestimation of the number of channels.

A documented R package for simulation and estimation in the VND model can be found at \url{https://github.com/ljvanegas/VND}.

\subsection{Literature review}
Various extensions of HMMs have been suggested, see e.g. \citep{sin1995,mari1997,fine1998,guan2016,siekmann2016, diehn2019}. Most related to our setting are factorial HMMs \citep{ghahramani1997, chen2009a}, that consider several independent unobservable chains. In this sense, each individual 
signal can be seen as an unobservable layer. In contrast, in our setting the chains are coupled and the dependency structure plays a key role. Coupled HMMs, see \citep{brand1997}, deal with dependency by embedding the system in a multidimensional chain as we do, and then applying HMM techniques to it. However, in our setting we cannot observe the state of each individual 
signal at any moment in time, which significantly complicates the situation and is the major motivation for our approach.

There are several works related to this situation, i.e., when the observations depend only on the sum of Markov chains; most of them require independence of channels, see e.g. \cite{yeo1989, fredkin1991, dabrowski1992, klein1997}. The concept of exchangeability of Markov chains in a multidimensional setting, which is equivalent to permutation invariance, was explored by \cite{gottschau1992} and extended to continuous time Markov chains by \cite{ball1997} denoted as \emph{aggregation}. Most similar in spirit to our work is the CK model by \cite{chung1996}, which provides a simple way to model the dependency of channels. However, the dependency structure of such a model is limited to a specific form of cooperative gating, since it is a linear interpolation between the fully coupled case, where all signals are always in the same state, and the independent case. In particular, within the ion channel setting it can only model fully clustered gating, where all channels have a tendency to open and close synchronously.

\section{Theoretical framework} \label{sec:theory}

\subsection{Setup}
Suppose there are $\ell\in\mathbb{N}$ \emph{emitters}, e.g. single ion channels or digital signals, where each generates a $\{0,1\}$-valued discrete-time sequence which we call \emph{signal}. The two values $\{0,1\}$ that each entry of the sequence can attain are called \emph{states}. Each signal at each time-point is absorbed by an aggregation procedure, leading to a superposition of the whole system of signals, which is afterwards noisily recorded by the receiver. For a schematic view of this setting see  Figure~\ref{fig:test}.
For ion channels, here $0\mathrel{\widehat{=}}$ closed, $1\mathrel{\widehat{=}}$ open, this corresponds to the measurement of total current of the superimposed channels. 
\begin{figure}[h!]
  \centering
  \includegraphics[width=\textwidth]{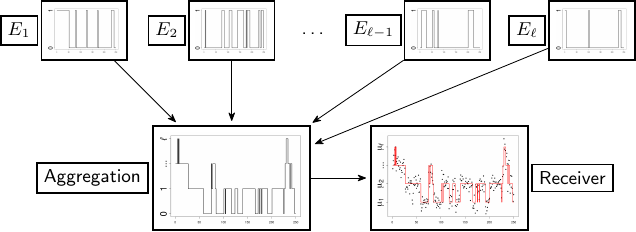}
  \caption{At a fixed time-point each emitter produces a $\{0,1\}$-valued entry of the signal. The experimental setup now leads to an aggregation, where the signals are superposed. Such aggregation is not directly observable but is hidden according to an HMM, such that only the sum of all signals, perturbed by random noise, can be recorded. \label{fig:test}}
\end{figure}
There, a single ion channel takes the role of an emitter and their conductance takes the role of the signals. We refer to \cite{zhang2001blind} and \cite{behr2018multiscale} for further applications in digital communication and cancer genetics.

For a formal description, suppose that in the following all random variables are defined on a common probability space $(\Omega,\mathcal{F},\mathbb{P})$. As an essential building block, for any ${\ell\in\mathbb{N}}$, ${k\in\mathbb{N}}$ and $j\in\{1,\dots,\ell\}$ let $X_{k}^{(j)} \colon \Omega \to \{0,1\}$ be a random variable describing the state of the $j$-th emitter at time point $k$. Considering those random variables simultaneously in $k$ we assume that each of the $\ell$ individual signals from the emitters are modeled by a homogeneous Markov chain $(X_k^{(j)})_{k\in\mathbb{N}_0}$, with $j=1,\ldots,\ell$ on the finite state space $\{0,1\}$. 
Then, the whole system of the $\ell$ signals can be modeled by an $\ell$-dimensional homogeneous Markov chain $\Xmp$ on the finite state space $\{0,1\}^\ell$, where $X_k := (X_k^{(1)},\dots,X_k^{(\ell)})^T$. Note that the transition matrix, say $M\in \R^{2^\ell\times 2^\ell}$, of the multidimensional process contains the full information on the dependence (coupling) between emitters. Therefore, finding suitable parameterizations of the matrix $M$ and its identification is one of the key points of this paper.

Moreover, most relevant is that we observe noisy measurements only on the sum of the system of signals of the individual emitters and not for each emitter separately, i.e., marginally. Therefore, we only have access to the sum process $(S_k)_{k\in\mathbb{N}}$ on the finite state space $[\ell\,] := \{ 0,\dots, \ell \}$ given by
\[
S_k := \sum_{j=1}^\ell X_k^{(j)}.
\]
Note that the sum process starts with $S_1$ to restrict the influence of the initial distribution of $(X_k)_{k\in\mathbb{N}_0}$.
The process $(S_k)_{k\in\mathbb{N}}$ can be seen as counting how many signals are in state ``$1$'', i.e., how many channels are open in the case of ion channels, at each discretized time point $k\in\mathbb{N}$. In general, this process is difficult to characterize, however, under certain conditions, the sum process $(S_k)_{k\in\mathbb{N}}$ is again a homogeneous Markov chain. For this, the lumping property turns out to be a sufficient criterion, see \cite[Chapter~6.3]{kemeny1976}.

\subsection{Lumping property}\label{sec:lump}
We start with providing notation. For $x\in \{0,1\}^\ell$, let $x=(x^{(1)},\dots,x^{(\ell)})^T$ and define the $1$-norm by
$
\Vert x \Vert_1 := \sum_{i=1}^{\ell} \vert x^{(i)} \vert, 
$ 
which denotes the number of non-zero entries of $x$.
Further, for $m\in[\ell\,]$ let
\begin{equation}\label{eq:zm}
\mathcal{Z}_m := \left\{z\in\{0,1\}^\ell \colon 
\Vert z \Vert_1
=m\right\}.
\end{equation}
Now we are able to define the lumping property and provide the aforementioned sufficiency criterion.
\begin{defi}[Lumping property]\label{def:lump}
  We say that $\Xmp$ satisfies the lumping property if for any $k\in\mathbb{N}$, $i,j\in[\ell\,]$ and $x,y\in \mathcal{Z}_i$ 
  it holds that
  \[
  \mathbb{P}(S_{k+1}=j \mid X_k=x) = \mathbb{P}(S_{k+1}=j \mid X_k=y),
  \]
  whenever $\mathbb{P}(X_k=y)\cdot \mathbb{P}(X_k=x)>0$.
\end{defi}
The lumping property implies that the sum process, i.e., the sequence $(S_k)_{k\in\mathbb{N}}$ of random variables, is a homogeneous Markov chain. The following result is proven in Appendix~\ref{suppl_sec: proof_lump}.
\begin{thm}\label{thm:q_dep_on_m}
  Let $M=(m_{x,y})_{x,y\in\{0,1\}^\ell}$ be the transition matrix of the Markov chain $\Xmp$. If $\Xmp$ satisfies the lumping property, then $(S_k)_{k\in\mathbb{N}}$ is a Markov chain with transition matrix $Q=(q_{i,j})_{i,j\in[\ell\,]}$  given by 
  \begin{equation}
  \label{eq:q_dep_on_m}
  q_{i,j} = \sum_{y\in \mathcal{Z}_j} m_{x,y}
  \end{equation}
  for arbitrary $x\in \mathcal{Z}_i$.
\end{thm}

\begin{rem}\label{rem:par_lum}
 If $\Xmp$ satisfies the lumping property, then the number of free parameters of the corresponding transition matrix $M=(m_{x,y})_{x,y\in\{0,1\}^\ell}$ is reduced to $2^\ell(2^\ell-1)-\ell(2^\ell-1-\ell)$. 
 This can be justified as follows: For any $i\in [\ell]$ choose a representative $x_i\in \mathcal{Z}_i$. Then, by \eqref{eq:q_dep_on_m}, we have for any $x\in \mathcal{Z}_i\setminus \{x_i\}$ and for any $j\in \{1, \dots, \ell\}$ 
 that\footnote{We leave $j=0$ out here, since for any Markov chain, whether it satisfies the lumping property or not, rows sum to $1$ and therefore
 	$
 	\sum_{y\in \mathcal{Z}_0} m_{x,y} = 1 - \sum_{j=1}^{\ell} \sum_{y\in \mathcal{Z}_j} m_{x,y} \, .
 	$
 }
 \[
 	\sum_{y\in \mathcal{Z}_j} m_{x,y} = \sum_{y\in \mathcal{Z}_j} m_{x_i,y},
 \]
 which yields
 \[
  m_{x,x_j} = \sum_{y\in\mathcal{Z}_j} m_{x_i,y} - \sum_{y\in\mathcal{Z}_j,\,y\not=x_j} m_{x,y}.
 \]
 Hence $\ell \sum_{i=0}^\ell (\vert \mathcal{Z}_i \vert-1) = \ell (2^\ell-\ell-1)$ entries can be expressed in terms of other entries of the transition matrix. This means that the lumping property reduces the number of free parameters by $\ell(2^\ell - \ell-1)$
 compared to the $2^\ell(2^{\ell}-1)$ free entries of a transition matrix for the general case. However, the number of parameters remains of the order $\mathcal{O}(2^{2\ell})$.
\end{rem}

For illustrative purposes we consider the reduction of the parameters in the case $\ell=2$.
\begin{example}
  Let $\ell=2$ and $M$ be the transition matrix of $\Xmp$. Then $M$ can be parameterized by
  \begin{equation}\label{eq:m_2x2}
  M=\begin{pmatrix}
  m_{(0,0),(0,0)} & m_{(0,0),(1,0)} & m_{(0,0),(0,1)}&m_{(0,0),(1,1)}\\
  m_{(1,0),(0,0)} & m_{(1,0),(1,0)} & m_{(1,0),(0,1)}&m_{(1,0),(1,1)}\\
  m_{(0,1),(0,0)} & m_{(0,1),(1,0)} & m_{(0,1),(0,1)}&m_{(0,1),(1,1)}\\
  m_{(1,1),(0,0)} & m_{(1,1),(1,0)} & m_{(1,1),(0,1)}&m_{(1,1),(1,1)}
  \end{pmatrix},\end{equation}
  with $m_{ x,(1,1) }=1-\sum_{y\neq (1,1)}m_{x,y}$ for all $x\in\{0,1\}^\ell$. According to Theorem \ref{thm:q_dep_on_m}, if we assume the lumping property is satisfied, then we have the following extra conditions
  \begin{align*}
  m_{(1,0),(0,0)} &= m_{(0,1),(0,0)}\\
  m_{(1,0),(1,0)}+m_{(1,0),(0,1)} &= m_{(0,1),(1,0)}+m_{(0,1),(0,1)}\,\
  \end{align*}
  which reduces the number of free entries from $12$ to $10$.
\end{example}

\subsection{Permutation invariance}
\label{sec:perm_inv}

In this section we introduce permutation invariance of a vector Markov chain $\Xmp$ and discuss its relationship to the lumping property and its consequences.

\begin{defi}[Permutation invariance] \label{def:perm_inv}
  We call a vector Markov chain $\Xmp$ \emph{permutation invariant} if for any $k\in\N$, $x,y\in\{0,1\}^\ell$, and any permutation matrix $P\in\{0,1\}^{\ell\times\ell}$ holds
  \[\mathbb{P}(X_{k+1}=y\mid X_k=x)=\mathbb{P}(X_{k+1}=P y\mid X_k= P x).\]
\end{defi}

This condition constrains clustered interactions, in particular, since any interaction must have the same effect on all signals, which only allows for clustered behavior among all signals, not subsets of them. Furthermore, we show that it implies the lumping property. We also discuss the number of free entries or parameters which determine the transition matrix of a Markov chain with this property. Note that permutation invariance means that relabeling of the coordinates of the vector Markov chain $\Xmp$ does not change conditional distributions. The following proposition is proven in Appendix~\ref{suppl_sec: proof_pi_lump}.

\begin{prop} \label{prop:pi_lump}
  If $\Xmp$ is permutation invariant, it satisfies the lumping property.
\end{prop}

The converse of Proposition \ref{prop:pi_lump} is in general not true, thus permutation invariance can be considered as a stronger condition than the lumping property. However, it is more accessible in the sense that it is easier to verify and has a more direct interpretation.
Moreover, permutation invariance leads to a considerable reduction of the number of parameters of the transition matrix of the Markov chain $\Xmp$. For the proof of the next result we refer to Appendix~\ref{suppl_sec: proof_pi_lump}.

\begin{prop} \label{prop:perm-inv-param}
  The transition matrix $M$ of a permutation invariant Markov chain $\Xmp$ is determined by $\ell(\ell+1)(\ell+5)/6$ parameters.
\end{prop}

It is particularly beneficial that the permutation invariance reduces the number of parameters of the transition matrix
from an exponential number in $\ell$ to the order $\mathcal{O}(\ell^3)$. Thus, it has significantly fewer free entries/independent parameters which eases estimation of each parameter significantly. For illustrative purposes we consider $\ell=2$ in the following example.
\begin{example}
  \label{exa_cont}
  Let $\ell=2$ and $M$ be the transition matrix of $\Xmp$. Then $M$ can be parameterized as in \eqref{eq:m_2x2}. If we have permutation invariance we obtain the extra conditions
  \begin{align*}
  m_{(0,0),(1,0)} &= m_{(0,0),(0,1)} &
  m_{(1,0),(0,0)} &= m_{(0,1),(0,0)}\\
  m_{(1,0),(1,0)} &= m_{(0,1),(0,1)} &
  m_{(1,0),(0,1)} &= m_{(0,1), (1,0)}\\
  m_{(1,1),(1,0)} &= m_{(1,1), (0,1)},
  \end{align*}
  which reduce the number of free entries from 12 to 7. Note that $m_{(1,0),(1,1)} = m_{(0,1), (1,1)}$ follows from row normalization.
\end{example}

\begin{rem}\label{rem:par_pi}
  For all $\ell \ge 2$ we have $\ell(\ell+1)(\ell+5)/6 > \ell(\ell+1)$. Therefore, the transition matrix of a permutation invariant Markov chain $\Xmp$ roughly has about $\ell/6$ times more free entries/parameters than the corresponding transition matrix of the resulting sum Markov chain $(S_k)_{k \in \mathbb{N}}$. Having our ion channel interpretation in mind this shows that permutation invariance is not sufficient for the transition probabilities of the individual channel currents to be fully determined by the transition probabilities of the sum current.
\end{rem}

\subsection{Chung-Kennedy model}
\label{subsec: CK-model}

We present the Chung-Kennedy model, abbreviated as CK model, following \cite{chung1996} by introducing two complementary examples of permutation invariant vector Markov chains, where a convex combination of both yields to the CK model. Note that by the permutation invariance those Markov chains also satisfy the lumping property.

\begin{example}[Fully coupled case]
	\label{ex: FC-case}
  For $a,b\in \{0,1\}$ let $\vec{a} := (a,\dots,a)^T \in\{0,1\}^\ell$ and $\vec{b}=(b,\dots,b)^T\in \{0,1\}^\ell$. Define the transition matrix $M^{\rm(FC)}=(m^{\rm(FC)}_{x,y})_{x,y\in\{0,1\}^\ell}$
  by
  \[
  m^{\rm(FC)}_{x,y} := 
  \begin{cases}
  \mathbb{P}(X^{(1)}_{k+1}=b\mid X^{(1)}_{k}=a) &\mbox{for } x = \vec{a}, \, y = \vec{b} \\
  0.5 &\mbox{for } x\not\in \{\vec{0},\vec{1} \},\, y\in \{\vec{0},\vec{1} \}\\
  0 & \mbox{for }y\not\in \{\vec{0},\vec{1} \}
  \end{cases}
  \]
  where $x=(x^{(1)},\dots,x^{(\ell)})^T$ and $y=(y^{(1)},\dots,y^{(\ell)})^T$ with $x,y\in\{0,1\}^\ell$.
  Then, for a Markov chain $\Xmp$ with transition matrix $M^{\rm(FC)}$ one obtains for any $i\in[\ell\,]$, that
  \[
  (X_k^{(i)})_{k\in\mathbb{N}} = (X_k^{(1)})_{k\in\mathbb{N}}.
  \]
  Therefore, except for the initial state, the Markov chain $\Xmp$ is completely determined by $(X_k^{(1)})_{k\in\mathbb{N}}$, where the entries of $X_k$ are just copies of $X_k^{(1)}$ and thus $X_k\in\{\vec{0},\vec{1}\}$ for any $k\in\mathbb{N}$.
  We refer to this scenario as the fully coupled case. Either $X_k\in \mathcal{Z}_0$ or $X_k\in\mathcal{Z}_\ell$ for any $k\in\mathbb{N}$, and therefore, $S_k\in\{0,\ell\}$. 
  In this case, permutation invariance follows trivially.
\end{example}
\begin{example}[Uncoupled case] \label{ex:uncoupled_case}
  Let the Markov chains $(X_k^{(1)})_{k\in\mathbb{N}_0},\dots,(X_k^{(\ell)})_{k\in\mathbb{N}_0}$ be independent identically distributed, such that each transition matrix is specified by $\lambda,\eta \in (0,1)$ with
  \[
  \mathbb{P}(X_{k+1}^{(i)}=0\mid X_k^{(i)}=0) := \lambda,\qquad
  \mathbb{P}(X_{k+1}^{(i)}=1\mid X_k^{(i)}=1) := \eta
  \]
  for all $i\in[\ell\,]$. Thus, the entries of the vector process $\Xmp$ are independent. We refer to this scenario as the uncoupled case.	
  For the transition matrix, denoted by $M^{\rm(UC)}=(m_{x,y}^{\rm(UC)})_{x,y\in\{0,1\}^{\ell}}$, of the Markov chain $\Xmp$ we have
  \[
  m^{\rm(UC)}_{x,y} = \lambda^{\vert \{i\colon x^{(i)}=y^{(i)}=0  \} \vert} (1-\lambda)^{\vert \{i\colon x^{(i)}=0,\,y^{(i)}=1  \} \vert}
  \eta^{\vert \{i\colon x^{(i)}=y^{(i)}=1  \} \vert} (1-\eta)^{\vert \{i\colon x^{(i)}=1,\,y^{(i)}=0  \} \vert},
  \]
  where $x=(x^{(1)},\dots,x^{(\ell)})^T$, $y=(y^{(1)},\dots,y^{(\ell)})^T$ with $x,y\in\{0,1\}^\ell$, and the cardinality of a set $A$ is denoted by $\vert A \vert$. For any permutation matrix 
  $P \in\{0,1\}^{\ell\times\ell}$
  note that 
  \begin{equation} 	\label{eq:permut_inv}
  m^{\rm(UC)}_{x,y} = m^{\rm(UC)}_{P x,P y},
  \end{equation}
  such that the permutation invariance property holds.
\end{example}

The Markov chains of the previous examples are complementary to each other in the sense that in the coupled case there is total dependence within the single signals and in the uncoupled case there is total independence between the signals. The idea of the CK model is to consider a Markov chain model with convex combination of the transition matrices of the former examples:

\begin{defi}[CK model, see \citep{chung1996}] \label{ex:lin-inter}
  For $\kappa\in[0,1]$, suppose that the transition matrix $M^{\rm(CK)}=(m^{\rm(CK)}_{x,y})_{x,y\in\{0,1\}^{\ell}}$ of $\Xmp$, is given by
  \[
  M^{\rm(CK)}= \kappa M^{\rm(FC)} + (1-\kappa) M^{\rm(UC)},
  \]
  where $M^{\rm(FC)}$ and $M^{\rm(UC)}$ denote the transition matrices of the previous examples.
\end{defi}

\begin{rem}
  By the fact that the permutation invariance property holds in the examples above we have for any $x,y\in \mathcal{Z}_i$ and permutation matrix $P\in\{0,1\}^{\ell\times \ell}$ for a Markov chain $\Xmp$ within the CK model that
  \begin{align*}
  \mathbb{P}(X_{k+1}=y\mid X_k=x) 
  & = m^{\rm(CK)}_{x,y} = \kappa m^{\rm(FC)}_{x,y} + (1-\kappa) m^{\rm(UC)}_{x,y}
  = \kappa m^{\rm(FC)}_{P x,P y} + (1-\kappa) m^{\rm(UC)}_{P x,P y}\\
  & = m^{\rm(CK)}_{Px,Py} = \mathbb{P}(X_{k+1}=P y\mid X_k= P x).
  \end{align*}
  Thus, the Markov chain $\Xmp$ with transition matrix $M^{\rm(CK)}$ is also permutation invariant.
\end{rem}

The vector Markov chain within the CK model is able to cover both, fully coupled and uncoupled independent behavior. The parameter $\kappa$ allows a balancing between both scenarios. In the field of ion channel research, it is the main competitor to the vector norm dependent model that we introduce now.

\subsection{Vector norm dependency}
\label{subsec: VND}

In the following a more sophisticated model of an aggregated state dependency of a Markov chain $\Xmp$ which satisfies the lumping property is suggested.

\begin{defi}[Vector Norm Dependency] \label{def:vnd}
  A vector Markov chain $\Xmp$ on $\{0,1\}^\ell$ with transition matrix $M^{\rm (VND)}=(m^{\rm(VND)}_{x,y})_{x,y\in\{0,1\}^\ell}$ is called vector norm dependent (VND) if for all $i\in\{1,\dots,\ell\}$, $r\in [\ell\,]$, $k\in\mathbb{N}_0$ and $b\in \{0,1\}$, the expression
  \begin{align}\label{eq:vnd_const_map}
  \mathbb{P}\left(X_{k+1}^{(i)}=b\middle| X_{k}^{(i)}=b,\Vert X_k \Vert_1 = r	\right)
  \end{align}
  is independent of $i$ and $k$, and
  \begin{align}\label{eq:m_vnd_def}
  m^{\rm(VND)}_{x,y} = \prod_{i=1}^{\ell} \mathbb{P}\left( X_{k+1}^{(i)}=y^{(i)} \middle| X_k^{(i)} = x^{(i)}, \Vert X_k \Vert_1 = \Vert x \Vert_1 \right),
  \end{align}
  for any $x=(x^{(1)},\dots,x^{(\ell)})^T$ and $y=(y^{(1)},\dots,y^{(\ell)})^T$ with $x,y\in\{0,1\}^\ell$. 
\end{defi}

A vector norm dependency structure as introduced in Definition \ref{def:vnd} refers to the fact that the entries of $X_{k+1}$ of the vector Markov chain $\Xmp$ might depend on $\Vert X_k \Vert_1$, that is, a transition from the $i$-th entry $X_k^{(i)}$ to $X_{k+1}^{(i)}$ is allowed to depend in an explicit way on 
$\Vert X_k \Vert_1$. In our ion channel application, this means that the probability of every open channel to close and every closed channel to open may depend on the number of open channels.
\begin{rem} \label{rem:param_vnd}
 For a vector norm dependent Markov chain with transition matrix $M^{\rm (VND)}=(m^{\rm(VND)}_{x,y})_{x,y\in\{0,1\}^\ell}$, equation \eqref{eq:vnd_const_map} says that for all $i,j\in\{1,\dots,\ell\}$, $r\in [\ell\,]$, $k\in\mathbb{N}_0$ and $b\in \{0,1\}$ the following holds
  \begin{align*}
  \mathbb{P}\left(X_{k+1}^{(i)}=b\middle| X_{k}^{(i)}=b,\Vert X_k \Vert_1 = r	\right) 
  & =
  \mathbb{P}\left(X_{k+1}^{(j)}=b\middle| X_{k}^{(j)}=b,\Vert X_k \Vert_1 = r	\right). 
  \end{align*}	
  Observe that within the VND Markov chain model the transition matrix $M^{\rm (VND)}$ is determined by $2\ell$ parameters, $\lambda_j$ and $\eta_{j+1}$, $j=0,\ldots,\ell-1$ which are given by
  \begin{align}
  \label{eq:trans_p_r}
  \lambda_j & := \mathbb{P}\left(	X_{k+1}^{(i)} = 0\middle| X_{k}^{(i)}=0,\Vert X_k \Vert_1 = j	\right),	\\
  \label{eq:trans_q_r}
  \eta_{j+1} & := \mathbb{P}\left(	X_{k+1}^{(i)} = 1\middle| X_{k}^{(i)}=1,\Vert X_k \Vert_1 = j+1	\right), 
  \end{align}
 	since, for $r\in\{0,\dots,\ell-1\}$ it follows from equation \eqref{eq:m_vnd_def} that
  \begin{align}\label{eq:m_vnd_formula}
  m^{\rm(VND)}_{x,y} &  = 
  \lambda_r^{\vert\{ i \colon x^{(i)}=y^{(i)}=0  \}\vert} 
  (1-\lambda_r)^{\vert\{ i \colon x^{(i)}=0, y^{(i)}=1  \}\vert} \\
  \notag
  & \qquad \qquad
  \times \eta_{r+1}^{\vert\{ i \colon x^{(i)}=y^{(i)}=1  \}\vert}
  (1-\eta_{r+1})^{\vert\{ i \colon x^{(i)}=1, y^{(i)}=0  \}\vert},
  \end{align}
  where $x=(x^{(1)},\dots,x^{(\ell)})^T$ and $y=(y^{(1)},\dots,y^{(\ell)})^T$ with $x,y\in\{0,1\}^\ell$.
  Thus, the entries of $M^{(\rm VND)}$ are determined by the $2\ell$ numbers $\lambda_0,\dots,\lambda_{\ell-1}, \eta_1,\dots,\eta_{\ell}$ and, in this sense, the number of parameters is $2\ell$. This is, for instance, in contrast to Example~\ref{ex:uncoupled_case} where only two parameters determine the transition matrix $M^{\rm(UC)}$. Moreover, 
  the $2\ell$ parameters have a direct interpretation, as the probability of each signal to stay in state ``\,$0$'' or ``\,$1$'' given the previous value of the signal and the total number of signals in state ``\,$1$''. 
\end{rem}

From the observation derived in equation \eqref{eq:m_vnd_formula} 
we immediately get the following result.
\begin{prop}  \label{prop:vnd_perm_inv}
  A vector norm dependent Markov chain $\Xmp$ 
  is permutation invariant.
\end{prop}

Now we provide a characterization of norm dependent Markov chains in terms of elementary and easily interpretable properties. To this end we introduce conditional independence of a vector Markov chain.

\begin{defi}[Conditional independence] \label{def:cond_ind}
  We call $\Xmp$ \emph{conditionally independent} w.r.t. the past, if
  \begin{align*}
  \mathbb{P}(X_{k+1}=y\mid X_k=x)&=\prod_{i=1}^\ell \mathbb{P}(X_{k+1}^{(i)}=y^{(i)}\mid X_k=x),
  \end{align*}
  for any $k\in\mathbb{N}_0$ and for all $x,y\in\set$ where $y=(y^{(1)},\dots,y^{(\ell)})^T$.
\end{defi}

In our application, this condition translates to saying that, given the state at the previous time point, all channels behave statistically independent. A dependency between channels is therefore only possible through the state at the previous time point, so interaction between channels always occurs with temporal delay. 

\begin{rem}
	\label{rem: FC_not_cond_ind}
	 Note that the property of the former definition is not always satisfied. For example,
	 the Markov chain $(X_k)_{k\in\mathbb{N}_0}$ with transition matrix $M^{(\text{FC})}$ that has been introduced in Example~\ref{ex: FC-case} 
	 is not conditionally independent. This can be seen by considering the case $\ell=2$. For $x = (1,0)^T$ we have
		\begin{align*}
		\mathbb{P}(X_{k+1}=(0,0)^T\mid X_k=x) = \mathbb{P}(X_{k+1}=(1,1)^T\mid X_k=x)&=0.5, \\
		\mathbb{P}(X_{k+1}=(0,1)^T\mid X_k=x) = \mathbb{P}(X_{k+1}=(1,0)^T\mid X_k=x)&=0.
		\end{align*}
		Now assume that there are probabilities
		\begin{align*}
		p_1 :=&~ \mathbb{P}(X_{k+1}^{(1)}=1\mid X_k=x), & p_2 :=&~ \mathbb{P}(X_{k+1}^{(1)}=0\mid X_k=x),\\
		p_3 :=&~ \mathbb{P}(X_{k+1}^{(2)}=1\mid X_k=x), & p_4 :=&~ \mathbb{P}(X_{k+1}^{(2)}=0\mid X_k=x) \, .
		\end{align*}
		Then, conditional independence requires the constraints $p_1 p_3 = p_2 p_4 = 0.5$ and $p_1 p_4 = p_2 p_3 = 0$, which are incompatible.
\end{rem}
Interestingly,
we can characterize vector norm dependency in terms of conditional independence and permutation invariance. This is formulated in the following result, which is proven in Appendix~\ref{suppl_sec:proof_app_lem}.

\begin{thm}[Characterization of VND Markov chain] 
  \label{thm:vnd-char}
  For a vector Mar\-kov chain $\Xmp$ assume that the initial distribution is permutation invariant\footnote{The distribution of $X_0$ is permutation invariant if $\mathbb{P}(X_0=y) = \mathbb{P}(X_0=Py)$ for any $y\in \{0,1\}^\ell$ and any permutation matrix $P\in \{0,1\}^{\ell\times\ell}$.}.
  Then, the following statements are equivalent:
  \begin{enumerate}
    \item The Markov chain $\Xmp$ is vector norm dependent;
    \item The Markov chain $\Xmp$ is permutation invariant and conditionally independent.
  \end{enumerate}
\end{thm}

Due to this theorem the VND model is only applicable if permutation invariance and conditional independence are sensible assumptions. For example, within the ion channel setting this is very useful, since both of these assumptions are readily interpretable: Permutation invariance can typically be taken for granted if only a single type of ion channel is present in the membrane. Conditional independence implies that no direct interaction between ion channels occurs, but interactions can be mediated by a collective effect, which can be interpreted as an effective potential. In particular, if the VND model provides a good fit, this can be seen as an indication that direct channel interactions are not required to explain channel conductivity in a system. Instead indirect channel interaction may be mediated by a large scale property like ion concentration or electrostatic potential.

The CK model satisfies permutation invariance as shown above, but since the fully coupled Markov model violates conditional independence, see Remark~\ref{rem: FC_not_cond_ind}, the CK model also violates conditional independence and is not included in the class of VND models, unless $\kappa=0$, which is the uncoupled case.

We stress that within the setting of a VND Markov chain a great variety of emitter interactions can be modeled. Here we identify two archetypes of such interactions that we also exhibit in our ion channel data analysis.
\begin{defi} \label{def:coop-comp}
	We call a VND Markov chain $\Xmp$ \emph{cooperative} if for all $a \in \{1, \dots, \ell-1\}$ it	holds
	\begin{equation} \label{eq: cooperative}
	\frac{\mathbb{P}\left(	X_{k+1}^{(i)} = 1\middle| X_{k}^{(i)}=0,\Vert X_k \Vert_1 = a	\right)}{\mathbb{P}\left(	X_{k+1}^{(i)} = 1\middle| X_{k}^{(i)}=0,\Vert X_k \Vert_1 = 0\right)}  
	> 1,\quad 
	\frac{\mathbb{P}\left(	X_{k+1}^{(i)} = 0\middle| X_{k}^{(i)}=1,\Vert X_k \Vert_1 = a	\right)}{\mathbb{P}\left(	X_{k+1}^{(i)} = 0\middle| X_{k}^{(i)}=1,\Vert X_k \Vert_1 = \ell\right)}  
	> 1,
	\end{equation}
	and we call it \emph{competitive} if for all $a \in \{1, \dots, \ell-1\}$ it holds
	\begin{equation} \label{eq: competitive}
	\frac{\mathbb{P}\left(	X_{k+1}^{(i)} = 1\middle| X_{k}^{(i)}=0,\Vert X_k \Vert_1 = 0\right)}{\mathbb{P}\left(	X_{k+1}^{(i)} = 1\middle| X_{k}^{(i)}=0,\Vert X_k \Vert_1 = a	\right)}  
	> 1,\quad 
	\frac{\mathbb{P}\left(	X_{k+1}^{(i)} = 0\middle| X_{k}^{(i)}=1,\Vert X_k \Vert_1 = a+1	\right)}{\mathbb{P}\left(	X_{k+1}^{(i)} = 0\middle| X_{k}^{(i)}=1,\Vert X_k \Vert_1 = 1\right)}  
	> 1.
	\end{equation}
\end{defi}

By the fact that $\Xmp$ is a VND Markov chain we have that the ratios in \eqref{eq: cooperative} and \eqref{eq: competitive} do not depend on $i\in\{1,\dots,\ell\}$ and $k\in\mathbb{N}_0$. Moreover, note that \eqref{eq: cooperative} is equivalent to $\frac{1-\lambda_{a}}{1- \lambda_0}>1$ and $\frac{1-\eta_a}{1-\eta_\ell}>1$ for all $a\in\{1,\dots,\ell-1\}$ as well as that \eqref{eq: competitive} is equivalent to $\frac{1- \lambda_0}{1-\lambda_{a}}>1$ and $\frac{1-\eta_{a+1}}{1-\eta_1}>1$ for all $a\in\{1,\dots,\ell-1\}$.
In a cooperative VND Markov chain model, signals switch to ``\,$1$'' more easily, as soon as at least one signal is ``\,$1$'' and switch to ``$0$'' more easily as soon as at least one signal is ``\,$0$''. The result is that the states with either all signals ``\,$0$'' or all signals ``\,$1$'' are more likely visited than for independent signals. In a competitive model, a signal is more likely to switch to ``\,$1$'' while no other is ``\,$1$'' and it is more likely to switch to ``\,$0$'' if more than one signal is ``\,$1$''. In this case, the state with one signal being ``\,$1$'' is more likely visited than for independent signals. It is obvious that these two properties are mutually exclusive. In Appendix~\ref{suppl_sec:comp-coop-def} we discuss generalized terminology regarding cooperative and competitive behavior.
Note that even the generalized notions of cooperative and competitive gating are not exhaustive, i.e., there are Markov chains, even within the class of VND chains, with more complex coupling behavior which elude such a simple classification.

We end this section by providing a representation of the transition matrix of the sum Markov chain $(S_k)_{k\in\mathbb{N}}$ based on a vector norm dependent Markov chain $\Xmp$. For the proof of the following result we refer to the end of Appendix~\ref{suppl_sec:proof_app_lem}.
\begin{prop}
  \label{prop:repres_Q_vnd}
  Given a vector norm dependent Markov chain $\Xmp$ with transition matrix $M^{(\rm VND)}$ determined by the parameters $\lambda_0,\dots,\lambda_{\ell-1},\eta_1,\dots,\eta_{\ell}\in [0,1]$, see Remark~\ref{rem:param_vnd}. Then,
  the transition matrix $Q=(q_{i,j})_{i,j\in[\ell\,]}$ of the corresponding sum Markov chain $(S_k)_{k\in\mathbb{N}}$ is given by
  \begin{align*}
  q_{i,j} := \sum_{r=\max\{0,i-j\}}^{\min\{i,\ell-j\}} {i\choose r} {\ell-i\choose j-i+r} \eta_{i}^{i-r}(1-\eta_{i})^{r} \lambda_i^{\ell-j-r}(1-\lambda_i)^{j-i+r},
  \end{align*}
  for any $i,j\in[\ell\,]$, where for completeness we set $\eta_0:=1$ and $\lambda_\ell:=1$.
\end{prop}

\subsection{Inverse lumping and identifiability}
\label{subsec: inv_lump_ident}

Let $\Xmp$ be a vector Markov chain on $\{0,1\}^\ell$ with transition matrix $M=(m_{x,y})_{x,y\in \{0,1\}^\ell}$. Suppose that $\Xmp$ satisfies the lumping property and therefore the corresponding sum process $(S_k)_{k\in\mathbb{N}}$ is a Markov chain with transition matrix $Q=(q_{i,j})_{i,j\in [\ell\,]}$, see Theorem~\ref{thm:q_dep_on_m}. The transition matrix $Q$ is completely determined by $M$, see equation~\eqref{eq:q_dep_on_m}. This fact has been already used in \cite{chung1996} leading to the CK model of Definition~\ref{ex:lin-inter}. However, we are interested in reversing the perspective. Namely, we ask the following question: \emph{Can we uniquely recover $M$ from $Q$?} We refer to this as the \emph{inverse lumping problem}. 

Without further conditions on $\Xmp$ this is not possible. Indeed the lumping property is not sufficient and even if we have permutation invariance we know from Proposition~\ref{prop:perm-inv-param} that the number of parameters determining $M$ is $\ell(\ell +1)(\ell+5)/6$, which is larger than the number of parameters that determine $Q$. This shows that the property of permutation invariance is not sufficient for an affirmative answer to the inverse lumping problem. We illustrate this in the case $\ell=2$.
\begin{example}
  Consider the same setting as in Example~\ref{exa_cont}, i.e.,
  let $\Xmp$ be a permutation invariant Markov chain and $Q$ be the transition matrix of $(S_k)_{k \in \mathbb{N}}$. Then, by Theorem~\ref{thm:q_dep_on_m} we have
  \begin{equation*}
  Q=\begin{pmatrix}
  m_{(0,0),(0,0)} & 2m_{(0,0),(1,0)} &m_{(0,0),(1,1)}\\
  m_{(1,0),(0,0)} & m_{(1,0),(1,0)} + m_{(1,0),(0,1)}&m_{(1,0),(1,1)}\\
  m_{(1,1),(0,0)} & 2m_{(1,1),(1,0)} &m_{(1,1),(1,1)}
  \end{pmatrix}.\end{equation*}
  Note that $Q$ is parameterized with $7$ parameters. In particular, knowing $Q$ the transition matrix $M$ cannot be fully recovered, since we are not able to identify $m_{(1,0),(1,0)}$ and $m_{(1,0),(0,1)}$.
\end{example}
This indicates the need to add structural assumptions on $\Xmp$. Suppose now that $\Xmp$ is permutation invariant and conditionally independent. Then, for suitable initial distributions, this leads to a vector norm dependent Markov chain, see Theorem~\ref{thm:vnd-char}. The number of parameters determining $M$ in this setting is $2\ell$, see Remark~\ref{rem:param_vnd}, which is smaller or equal than $\ell(\ell+1)$. Therefore, a solution of the inverse lumping problem is not immediately excluded. 
Moreover, the transition matrix $Q^{\rm(VND)}$ can be explicitly stated in terms of the parameters $\lambda_j$ and $\eta_{j+1}$ for $j=0,\dots,\ell-1$ of the VND Markov chain, see Proposition \ref{prop:repres_Q_vnd}. By equation \eqref{eq:m_vnd_formula} also $M^{(\rm VND)}$ can be explicitly stated in terms of the parameters $\lambda_j$ and $\eta_{j+1}$. This enables one to extract knowledge about $M^{(\rm VND)}$ from $Q^{\rm(VND)}$ to the extent that the sum process determines the parameters of $M^{\rm (VND)}$ uniquely. For the proof of the following result we refer to Appendix~\ref{sec:proof-params}.

\begin{thm}	[Inverse lumping identifiability for VND Markov chains] \label{thm:params-unique}
  Let $(S_k)_{k\in\mathbb{N}}$ be a sum Markov chain with transition matrix $Q^{\rm(VND)}$ on $[\ell\,]$ based on a vector norm dependent Markov chain $\Xmp$ with transition matrix $M^{\rm(VND)}$. If $\ell$ is odd, then the parameters $\lambda_j$ and $\eta_{j+1}$ with $j=0,\dots,\ell-1$ defining $M^{(\rm VND)}$ are uniquely determined by the entries of $Q^{\rm(VND)}$. If $\ell$ is even, the same holds true provided that $\lambda_{\ell/2} \geq 1-\eta_{\ell/2}$.
\end{thm}

Using equation \eqref{eq:m_vnd_def}, this means for a VND Markov chain we can recover the transition matrix $M^{\rm (VND)}$ from the transition matrix $Q^{\rm (VND)}$ of the corresponding sum process. This paves the way for estimating these parameters from data, which we address in the following section. In the setting of the previous theorem we can therefore give an affirmative answer to the question of the inverse lumping problem.

Let us emphasize that the VND modeling is flexible enough for describing the system behavior we are interested in and at the same time is specific enough so that we can recover the parameters uniquely from the transition matrix $Q^{\rm(VND)}$. This conclusion cannot be reached by the lumping property or even permutation invariance alone, since the number of free parameters of $M$ is larger than the number of entries of $Q$ (see Remarks \ref{rem:par_lum} and \ref{rem:par_pi}), which leads to an underdetermined system.

\section{VND Hidden Markov model estimation} \label{sec:hmm-est}
We provide a short review on homogeneous HMM, define basic concepts and explain our HMM setting. Additionally, we introduce a customized Baum-Welch algorithm, which is adjusted to account for our specific modeling.

\subsection{VND Hidden Markov model and vector norm dependency}
Assume that
the measurable space $(\Omega,\mathcal{F})$ is equipped with  
a family of probability measures $(\mathbb{P}_\theta)_{\theta\in\Theta}$, where $\Theta \subseteq \mathbb{R}^{d} $ for some $d\in \mathbb{N}$ denotes an underlying parameter set. Let $\ell\in\mathbb{N}$ (e.g., corresponding to the number of channels) and let $(S_k,Y_k)_{k\in\mathbb{N}}$ be a bivariate stochastic process defined on $(\Omega,\mathcal{F})$, where $(S_k)_{k\in\mathbb{N}}$ is a Markov chain on the finite state space $[\ell\,]$ and $(Y_k)_{k\in\mathbb{N}}$, conditioned on $(S_k)_{k\in\mathbb{N}}$, is a real-valued, independent sequence of random variables. Taking the parameterized family of probability distributions into account, this leads to a parameterized homogeneous HMM $(S_k,Y_k)_{k\in\mathbb{N}}$. Given observed data $y_1,\dots,y_K\in \mathbb{R}$ for $K\in\mathbb{N}$, i.e., realizations of $Y_1,\dots,Y_K$, the goal is to determine the ``true'' underlying parameter $\theta^*\in \Theta$. In this context we call $(S_k)_{k\in\mathbb{N}}$ the hidden Markov chain and the distribution of $Y_k$ given $S_k$ the \emph{emission distribution}. Furthermore, assume that the parameter set can be represented as $\Theta = \Theta_H \times \Theta_E$, where $\Theta_H$ corresponds to the part of the parameters which come from the hidden Markov chain and $\Theta_E$ denotes the part which comes from the emission distribution.

Suppose now that $(S_k)_{k\in\mathbb{N}}$ is the sum process based on a vector norm dependent Markov chain $(X_k)_{k\in\mathbb{N}_0}$ on $\{0,1\}^\ell$. 
Set $\Theta_H = [0,1]^{2\ell}$, such that an element $\theta_H\in\Theta_H$ is given by $\theta_H=(\lambda_0,\dots,\lambda_{\ell-1},\eta_1,\dots,\eta_\ell)$, where $\lambda_r$ and $\eta_{r+1}$ with $r=0,\dots,\ell-1$ determine the transition matrix $Q^{\text{(VND)}}(\theta_H)=(q^{\text{(VND)}}_{i,j}(\theta_H))_{i,j\in[\ell\,]}$ of
the Markov chain $(S_k)_{k\in\mathbb{N}}$ as in Proposition~\ref{prop:repres_Q_vnd}. In formulas, for any $k\in\mathbb{N}$ and any $\theta\in\Theta$ we have
\[
\mathbb{P}_\theta (S_{k+1}=s_{k+1} \mid S_{k}=s_{k},\ldots,S_1=s_{1}) 
= \mathbb{P}_\theta (S_{k+1}=s_{k+1}| S_{k} = s_{k})
= q^{(\text{VND})}_{s_k,s_{k+1}}(\theta_H),
\]
where $(s_1,\dots,s_{k+1})\in [\ell]^{k+1}$ and $\theta=(\theta_H,\theta_E)\in \Theta$. In particular, note that the transition matrix $Q^{\text{(VND)}}(\theta_H)$ does not depend on $k$, which means that the hidden Markov chain is homogeneous.

Additionally, for $\theta_E \in \Theta_E$ we assume that the emission distribution is determined by a strictly positive probability Lebesgue density function  $g_{\theta_E}\colon \mathbb{R}\times [\ell\,] \to (0,\infty)$, such that
\begin{align*}
\mathbb{P}_\theta (Y_k\in A\mid S_k = s_{k},\ldots,S_1=s_{1}, Y_{k-1} = y_{k-1},\ldots,Y_1=y_{1}) & = \mathbb{P}_\theta(Y_k\in A \mid S_{k} = s_{k}) \\
& = \int_A g_{\theta_E}(y,s_k) \dint y,
\end{align*}
for any $k \in\mathbb{N}\setminus\{1\}$, $(s_1,\dots,s_k)^T \in[\ell\,]^k$, $(y_1,\dots,y_{k-1})^T\in \mathbb{R}^{k-1}$, any Borel set $A\subseteq \mathbb{R}$ and $\theta=(\theta_H,\theta_E)\in \Theta$. Let us emphasize here that we consider only homogeneous emission measures, i.e., $g_{\theta_E}$ does not depend on $k$ (in contrast to the considerations in \cite{diehn2019}). As a concrete setting, we provide a Gaussian standard scenario.

\begin{example}
  \label{ex:Gauss_emission}
  With $\mathcal{N}(m,v^2)$ and $m\in\mathbb{R}$ as well as $v>0$ we denote the normal distribution with mean $m$ and variance $v^2$.
  Let $\Theta_E=\mathbb{R}^2\times (0,\infty)^{\ell+1}$ and $\theta_E\in \Theta_E$ with $\theta_E = (\mu,\nu,\sigma_0,\dots,\sigma_{\ell})$. For  $k\in\mathbb{N}$ 
  consider
  \[
  Y_k = \mu + S_k \nu + \sigma_{S_k} \xi
  \] 
  with $\xi\sim\mathcal{N}(0,1)$, i.e, $Y_k\sim \mathcal{N}(\mu+S_k\nu,\sigma^2_{S_k})$. Therefore, for any $j\in[\ell\,]$ we have
  \[
  g_{\theta_{E}}(y,j) = \frac{1}{\sqrt{2\pi \sigma_j^2}} \exp\left(-\frac{\vert y-\mu-j\nu \vert^2}{2\sigma_{j}^2}\right). 
  \]
\end{example}

Assume that $\pi=(\pi^{(0)},\dots,\pi^{(\ell)})\in [0,1]^{[\ell\,]}$ is the probability vector that provides the initial distribution of the Markov chain $(S_k)_{k\in\mathbb{N}}$. By convention, let $\mathbb{P}_\theta(S_1=s \mid S_0):=\pi^{(s)}$ for any $s\in [\ell\,]$.
For simplicity, we assume $\pi$ to be known. 
Furthermore, to shorten the notation for a finite sequence $z_1,\dots,z_K$ we write $z_{1:K}$. (Thus, the event $\{S_{1:K} = s_{1:K}\}$ coincides with $\{S_1 = s_1, \ldots, S_k = s_k\}$.)
Then, 
the probability of 
$Y_{1:K}$ from the HMM being in a Borel set $A\subseteq \mathbb{R}^K$ is determined by
\begin{align}\label{eq:likelihood}
\notag
\mathbb{P}_\theta(S_{1:K} = s_{1:K}, Y_{1:K} \in A) 
= & \int_A \prod_{k=1}^K g_{\theta_E}(y_k,s_k) \mathbb{P}_\theta (S_k = s_k | S_{k-1}=s_{k - 1}) 
\dint y_{1:K}
\\
=&  \int_A g_{\theta_E}(y_1,s_1) 
\pi^{(s_1)} \times \prod_{k=2}^K g_{\theta_E}(y_k,s_k) q^{(\text{VND})}_{s_{k-1},s_k}
\dint y_{1:K}. 
\end{align}

\subsection{Parameter estimation and customized Baum-Welch algorithm}\label{sec:BW_alg}

For parameter estimation within the previously described HMM setting we modify the well-known  Baum-Welch algorithm, whose convergence was discussed in \cite{baum1970}. It is based on the expectation maximization (EM) paradigm, which means that an expectation computation step is followed by a maximization step. In order to discuss the algorithm we provide for $s_{1:K}\in [\ell\,]^K$ and $y_{1:K}\in \mathbb{R}^K$ the log-likelihood function $\theta \mapsto \mathcal{L}(\theta;s_{1:K},y_{1:K})$ within our parameterization. From \eqref{eq:likelihood} we can deduce for $\theta=(\theta_H,\theta_E)\in \Theta$ that 
\begin{align}\label{eq:like2}
\mathcal{L}(\theta,s_{1:K}, y_{1:K}) &= \log\pi^{(s_1)} + \sum_{k=1}^{K-1}\log q^{(\text{VND})}_{s_k,s_{k+1}}(\theta_H) + \sum_{k=1}^K\log g_{\theta_E}(y_k,s_k).
\end{align}
Note that we can simplify the standard algorithm by using the parameterized components of the transition matrix $Q^{(\text{VND})}(\theta_H)$ in the second term of the log-likelihood function. This is in contrast to settings where all $(\ell+1)\ell$ entries of possible transition matrices determine the parameters of the hidden Markov chain.

This modified Baum-Welch algorithm leads to a sequence $(\theta_t)_{t\in\mathbb{N}}\subset \Theta$ and consists of the following steps which are performed with increasing iteration index $t\in \mathbb{N}$
until a convergence criterion is reached:
\begin{enumerate}
  \item For all $k=1,\dots,K$ compute the so-called univariate and bivariate filtering distributions, given by $s\mapsto \mathbb{P}_{\theta_t}(S_k=s\mid Y_{1:K}=y_{1:K})$
  and $(r,s)\mapsto \mathbb{P}_{\theta_t}(S_{k}=r,S_{k+1}=s\mid Y_{1:K}=y_{1:K})$ with $r,s\in [\ell\,]$, by a forward-backward algorithm using the parameter $\theta_t\in \Theta$ determined in the previous iteration.
  \item Expectation step: Given $\theta_t$, and of course $y_{1:K}$, define $\theta \mapsto f(\theta;\theta_t,y_{1:K})$  with $\theta=(\theta_H,\theta_E)\in \Theta$ by
  \begin{align*}
  f(\theta;\theta_t,y_{1:K})
  = & \sum_{k=1}^K \sum_{s_k\in [\ell\,]} \log g_{\theta_E}(y_k,s_k) \mathbb{P}_{\theta_t}(S_k=s_k\mid Y_{1:K}=y_{1:K})\\
  +\sum_{k=1}^{K-1}  \sum_{s_k,s_{k+1}\in [\ell\,]} & 
  \log q^{(\text{VND})}_{s_k,s_{k+1}}(\theta_H)
  \mathbb{P}_{\theta_t}(S_k=s_k,S_{k+1}=s_{k+1}\mid Y_{1:K}=y_{1:K}).	
  \end{align*}
  (We obtain this function by taking the expectation of \eqref{eq:like2} w.r.t. the distribution of $S_{1:K}$ given $Y_{1:K}=y_{1:K}$ under $\mathbb{P}_{\theta_t}$ and neglect multiplicative constants as well as terms that do not depend on $\theta$.)
  \item Maximization step: Compute $\theta_{t+1} := \text{arg}\max_{\theta\in \Theta} f(\theta;\theta_t,y_{1:K})$. 
\end{enumerate}

As mentioned above, in \cite{baum1970}, it is shown that the sequence $(\theta_t)_{t\in\mathbb{N}}\subset \Theta$ converges under weak assumptions to a local maximum of the likelihood. In particular, in the Gaussian emission distribution setting of Example~\ref{ex:Gauss_emission} and our VND modeling convergence is achieved. Since the MLE is asymptotically efficient, see \cite{BRR1998}, it achieves a better estimation accuracy than the proposed least-squares procedure in \cite{chung1996}. Integrating the parametric form in the expectation step has the benefit that in each iteration the parameters remain in the parameter space. In contrast to the classical Baum-Welch algorithm, we do not obtain a closed expression for $\theta_{t+1}$ in the third step, since the parametric form of $q^{(\text{VND})}_{s_{k},s_{k+1}}(\theta_{t+1})$ is convoluted (see Proposition~\ref{prop:repres_Q_vnd}). Instead, we propose to solve this maximization problem numerically.

It is not \textit{a priori} clear if the function $\theta \mapsto f(\theta,\theta_t,y_{1:K})$ is strictly unimodal. If multiple local minima exist, this would impede the numerical optimization. From \eqref{eq:like2} we observe that all terms except the VND parameter term $\sum_{k=1}^{K-1}\log q^{(\text{VND})}_{s_k,s_{k+1}}(\theta_H)$ are strictly concave in the Gaussian case and any case of emission densities with strictly concave likelihood. In Appendix~\ref{suppl_sec:max} we investigate the VND parameter term numerically for a model with $\ell=2$. As it can be seen, the likelihood is unimodal in $\lambda_0$ and $\eta_2$ but the joint likelihood for $\lambda_1$ and $\eta_1$ is bimodal. These two modes correspond exactly to the parameter ambiguity in Theorem \ref{thm:params-unique}. Thus, restricting to $\lambda_1\geq 1- \eta_1$, there are no local maxima which impede parameter estimation.

\subsection{Model selection}\label{sec:modelsel}

In many applications the exact number of emitters $\ell\in\N$ is unknown and is often determined heuristically by visual inspection\footnote{Visual inspection refers to counting the number of visible distinct current levels within the given data.} or with some uncertainty by additional experiments and expert knowledge, as discussed in Section~\ref{subsec: intro_ion_channels}. Therefore, the question naturally emerges, whether the number of channels can be reliably estimated from the data. This estimation problem is an instance of model selection, a topic which in the general context of HMMs has been widely studied. It is known to be notoriously difficult as a reasonable balancing between model complexity and statistical precision has to be found, which in general depends on the unknown model itself. Commonly used methods are penalized likelihood criteria, such as the  Bayesian information criteria (BIC) and various cross validation approaches. A survey of multiple methods is given in \cite{celeux2008}.

Theoretical properties of these criteria are generically proven in the context of nested models. For some examples, see \citep{csiszar2000, gassiat2003, celeux2008, lehericy2019, yonekura2021}. 
Unfortunately, the VND Markov chain models for different numbers of emitters are not nested, but we still find the aforementioned model selection methodology useful in the context of our ion channel data, see Section~\ref{subsec: est_num_chan} for the performance of several of such criteria in our context. In particular, they return the largest number of open channels occurring in the data as proxy for $\ell$, thus apparently confirming the visual inspection.

In general, any reasonable estimation method of $\ell$, based solely on finitely many observations $y_1,\dots,y_K$, can only return an estimate of the maximal observed number of emitters. Especially in the case of competitive interaction, this is likely to be smaller than the actual number $\ell$. Already in the simplified setting of $\ell$ emitters that act identically state- and time-independently of each other, the estimation of $\ell$ from superpositional signals corresponds to the approximation of the number of trials $\ell$ in a binomial model with unknown success probability $p$ based on finitely many observations.
Here, it is known that without further prior information it is not possible to recover $\ell$ consistently when $p$ is too small \citep{Schmidt2021}. In our VND Markov chain models this deficiency leads to an underestimation of the true number of emitters. A simulation study illustrating this effect for parameters chosen according to the estimates from the ion channel data is presented in Appendix~\ref{suppl_sec: est_num_chan}. 
Notably, we found that the estimated parameters, even if $\ell$ is underestimated, still reliably identify the interaction archetypes presented in Definition~\ref{def:coop-comp}. In this sense, the estimated VND parameters are robust under the underestimation of the number of channels, see Section~\ref{subsec:robust}, which 
is of particular practical relevance.

\section{Application to ion channels} \label{sec:ion-channels}
Gating dynamics of ligand-gated ion channels vary depending on binding-ligand concentration (e.g. the cytosolic and/or intraorganelle luminal Ca${}^{2+}$ concentration). Here we model a set of multiple ion channels in an artificial lipid bilayer by a Gaussian emission HMM as explained in Example~\ref{ex:Gauss_emission}.
Our goal is to model the multiple channels in the bilayer as well as the dependencies between them. These dependencies are not necessarily caused by direct physical interaction of channels but may be mediated indirectly by joint environment factors, such as an overall increase in Ca${}^{2+}$ concentration on either side of the membrane.

\subsection{Data sets}\label{subsec:data_sets}
Our present application is in RyR2 channels, as introduced in Section~\ref{subsec: intro_ion_channels}. 
On these channels, extensive research about the gating mechanism at low levels of Ca${}^{2+}$ where only single channels are open at a time has been conducted. 
The next step is to investigate the dependencies between multiple channels since the probability of subcellular calcium release may depend significantly on those channel dependencies. Several investigations of local channel clustering have been conducted on living cells, cf. \citep{walker2014,walker2015}. While nonindependent gating has been reported for RyR2 channels in some studies, the question if and how cooperative gating occurs in subcellular channel clusters is not conclusively answered, cf. \citep{marx2001,laver2004,chen2009,walker2014,williams2018}.

We investigate two data sequences, which were measured subsequently on the same system of wild-type RyR2 channels in a synthetic lipid bilayer with different luminal (trans) concentrations of Ca${}^{2+}$ ions of 5 mM (data set 1) and 10 mM (data set 2) at a constant cytosolic (cis) Ca${}^{2+}$ concentration of 150 nM. Additionally, 53 mM Ba${}^{2+}$ was present on the trans side of the lipid bilayer and used as a principal charge carrier, since higher conductance of RyR2 for Ba${}^{2+}$ results in better signal-to-noise ratio.
Materials and methods are discussed in more detail in Appendix~\ref{suppl_sec:exp_meth}.

\subsection{Determining a suitable model}\label{subsec:det_suit_model}

By counting the number of current levels which are clearly distinguishable in the data, we can see at most three channels open at any one time, thus we determine $\ell=3$, i.e., we assume that three channels are present in the membrane. This is also supported by the formal model selection, see Section~\ref{subsec: est_num_chan}. Based on that we compare and investigate three models. The baseline model, against which we compare the two other models is the model of uncoupled independent channels (UC). Here, the sum Markov chain $(S_k)_{k\in\mathbb{N}}$ is based on $(X_k)_{k\in\mathbb{N}_0}$ considered in Example~\ref{ex:uncoupled_case} and the corresponding transition matrix, using shorthand notation ${\overline{\lambda}_j := (1- \lambda_j)}$ and $\overline{\eta}_j := (1- \eta_j)$, is given by
\begin{align*}
Q^{\textnormal{\rm(UC)}} := \begin{pmatrix}
\lambda_0^3 & 3\lambda_0^2\overline{\lambda}_0 & 3\lambda_0\overline{\lambda}_0^2 &\overline{\lambda}_0^3 \\
\lambda_0^2\overline{\eta}_1 & \lambda_0^2 \eta_1 + 2\lambda_0\overline{\lambda}_0\overline{\eta}_0 & 2\lambda_0\overline{\lambda}_0\eta_1 + \overline{\lambda}_0^2 \overline{\eta}_1 & \overline{\lambda}_0^2 \eta_1\\
\lambda_0\overline{\eta}_1^2 & 2\lambda_0\eta_1\overline{\eta}_1 + \overline{\lambda}_0\overline{\eta}_1^2 & \lambda_0\eta_1^2 + 2\overline{\lambda}_0\eta_1\overline{\eta}_1 & \overline{\lambda}_0 \eta_1^2\\
\overline{\eta}_1^3 & 3\eta_1\overline{\eta}_1^2 & 3\eta_1^2\overline{\eta}_1 & \eta_1^3
\end{pmatrix} \, .
\end{align*}
The other two models are the CK model from Definition~\ref{ex:lin-inter} by \cite{chung1996} and our VND model described in Section~\ref{sec:hmm-est}.
In the CK model,
the transition matrix of $(S_k)_{k\in\mathbb{N}}$  takes the form
\begin{align*}
Q^{\textnormal{\rm(CK)}} := (1- \kappa) Q^{\textnormal{\rm(UC)}} + \kappa \begin{pmatrix}
\lambda_0 & 0 & 0 & \overline{\lambda}_0 \\
1/2 & 0 & 0 & 1/2 \\
1/2 & 0 & 0 & 1/2 \\
\overline{\eta}_1 & 0 & 0 & \eta_1
\end{pmatrix},
\end{align*}
where additional to the parameters $\lambda_0$ and $\eta_1$ a coupling parameter $\kappa\in [0,1]$ appears.
In the VND model, the Markov chain $(S_k)_{k\in\mathbb{N}}$ is given by the sum process based on a vector norm dependent Markov chain, recall Definition \ref{def:vnd}. For $\ell=3$ its transition matrix is given by
\begin{align*}
Q^{\textnormal{\rm(VND)}} := \begin{pmatrix}
\lambda_0^3 & 3\lambda_0^2\overline{\lambda}_0 & 3\lambda_0\overline{\lambda}_0^2 &\overline{\lambda}_0^3 \\
\lambda_1^2\overline{\eta}_1 & \lambda_1^2 \eta_1 + 2\lambda_1\overline{\lambda}_1\overline{\eta}_1 & 2\lambda_1\overline{\lambda}_1\eta_1 + \overline{\lambda}_1^2 \overline{\eta}_1 & \overline{\lambda}_1^2 \eta_1\\
\lambda_2\overline{\eta}_2^2 & 2\lambda_2\eta_2\overline{\eta}_2 + \overline{\lambda}_2\overline{\eta}_2^2 & \lambda_2\eta_2^2 + 2\overline{\lambda}_2\eta_2\overline{\eta}_2 & \overline{\lambda}_2 \eta_2^2\\
\overline{\eta}_3^3 & 3\eta_3\overline{\eta}_3^2 & 3\eta_3^2\overline{\eta}_3 & \eta_3^3
\end{pmatrix} \, .
\end{align*}
The question naturally emerges, whether the CK or the VND model provide a better explanation for the data than the UC model.

\begin{figure}[h!]
  \centering
  \includegraphics[width=0.95\textwidth]{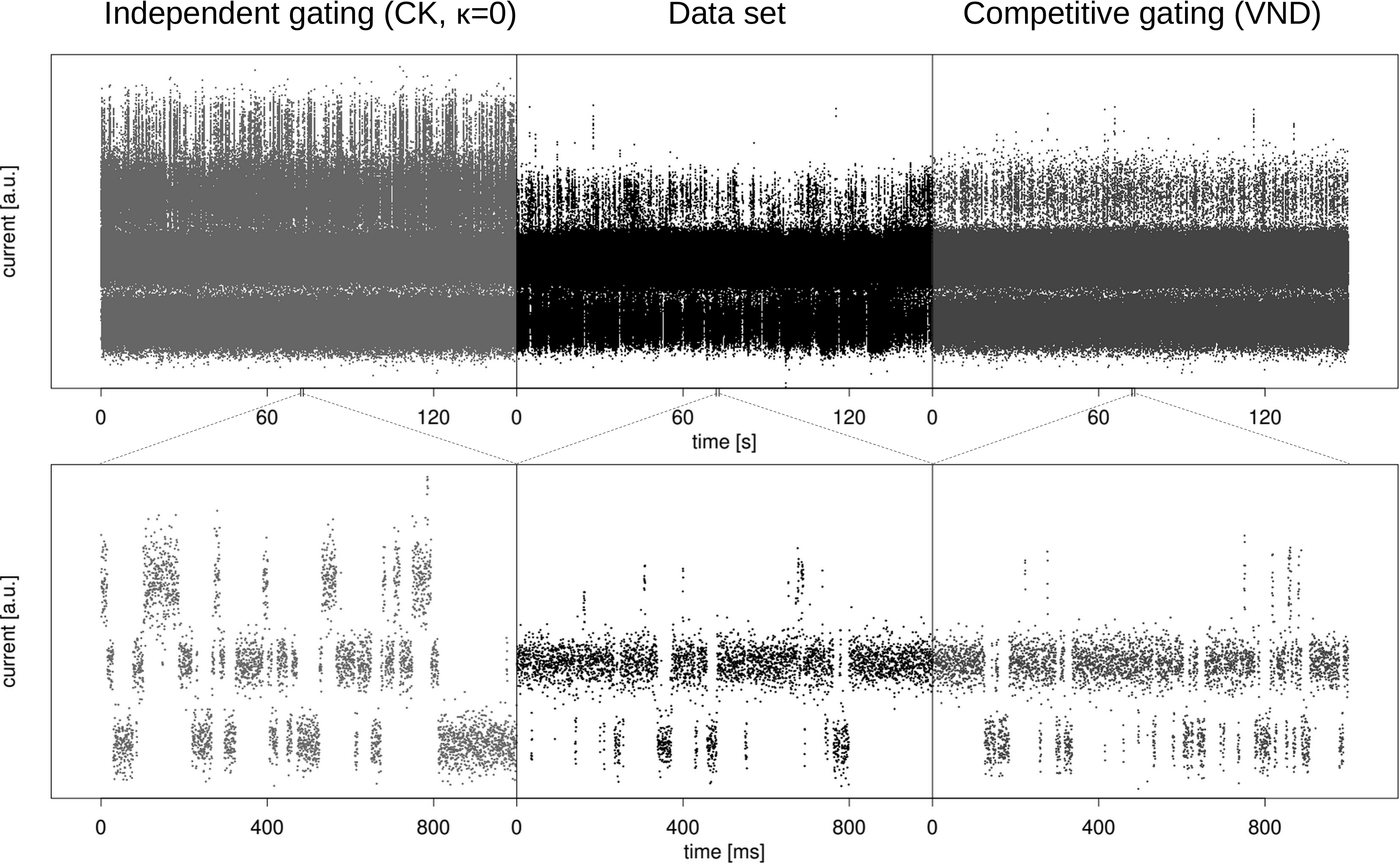}
  \caption{A current measurement on multiple channels in a synthetic lipid bilayer is compared to simulated data from two models at two scales. The zoom-ins show a time interval of 1 second length at the same randomly chosen time for all three trajectories and thus reflect typical behavior of the models and data at short time scales. The data is displayed in the middle, a simulated trace from the Chung-Kennedy model using the estimated parameter value $\widehat{\kappa}=0$, which amounts to independent channels, is displayed to the left and a simulated trace from the VND model with competitive gating is displayed on the right. The parameters used for the simulations are the estimated parameters from the data under the respective models. Due to signal filtering, the real data contain more intermediate values between the levels. It is obvious, especially from the zoom-in plots in the lower panel, that the competitive gating VND model reflects the channel gating behavior much more faithfully than the CK model presented in Definition~\ref{ex:lin-inter}, which predicts uncoupled channels. \label{fig:intro}}
\end{figure}

\begin{figure}[h!]
  \centering
  \subcaptionbox*{state $0$}[0.31\textwidth]{\includegraphics[width=0.3\textwidth]{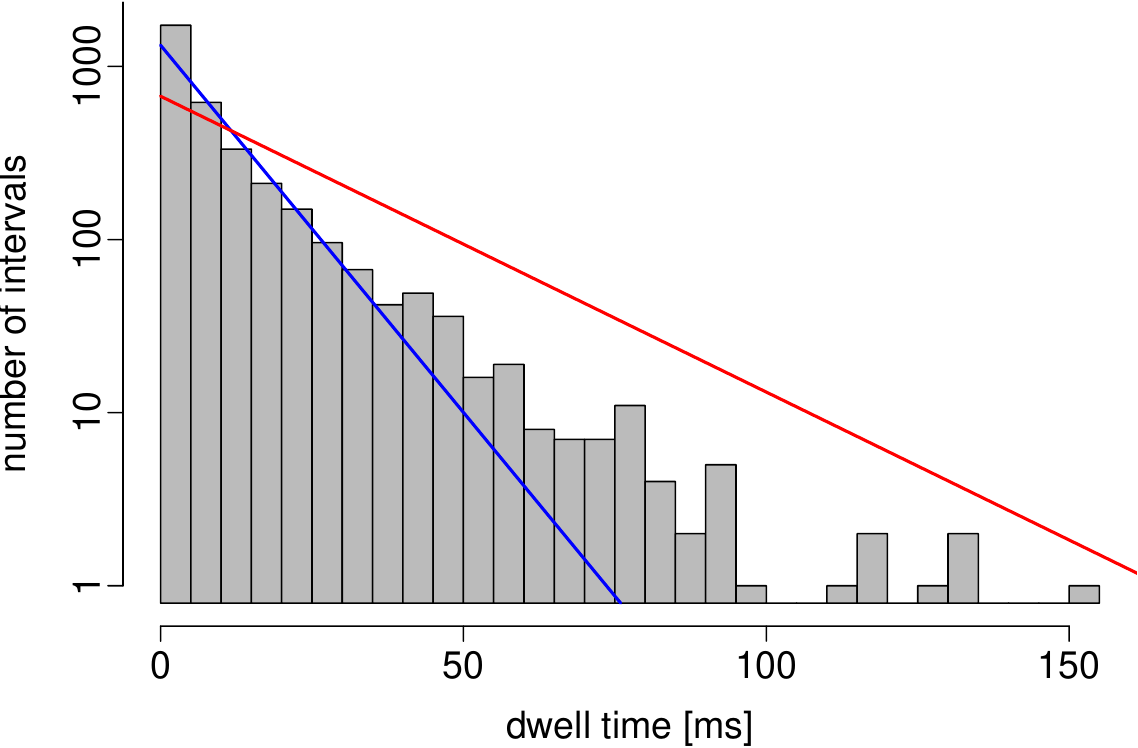}}
  \subcaptionbox*{state $1$}[0.31\textwidth]{\includegraphics[width=0.3\textwidth]{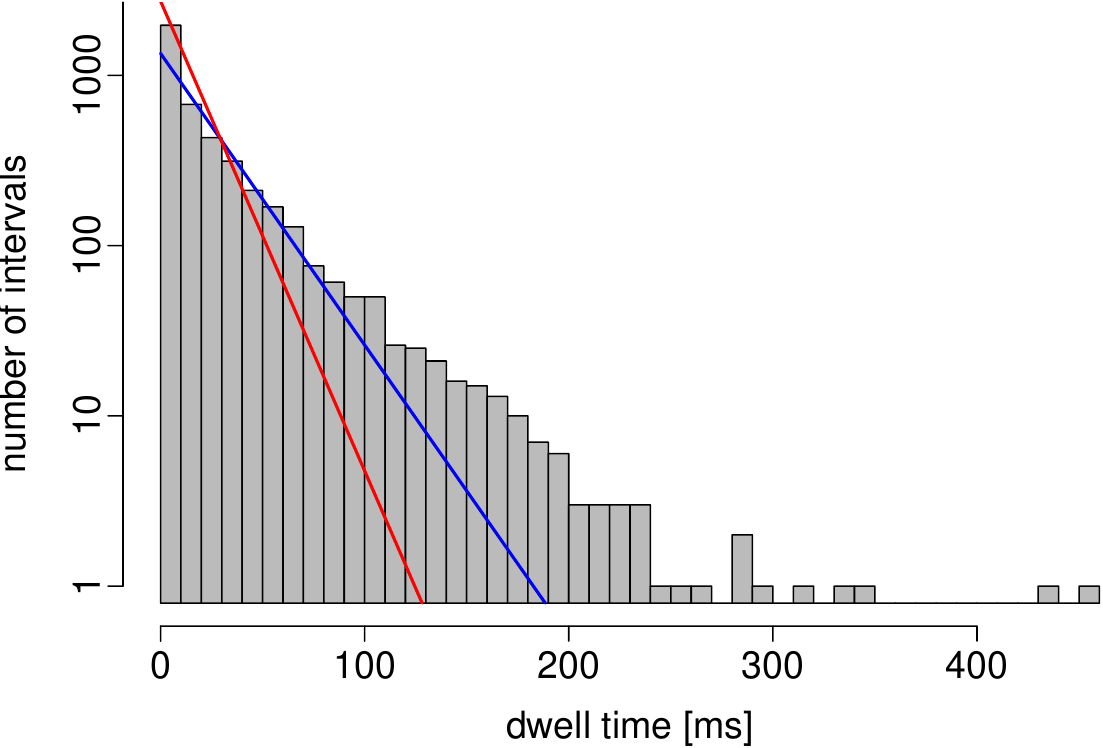}}
  \subcaptionbox*{state $2$}[0.31\textwidth]{\includegraphics[width=0.3\textwidth]{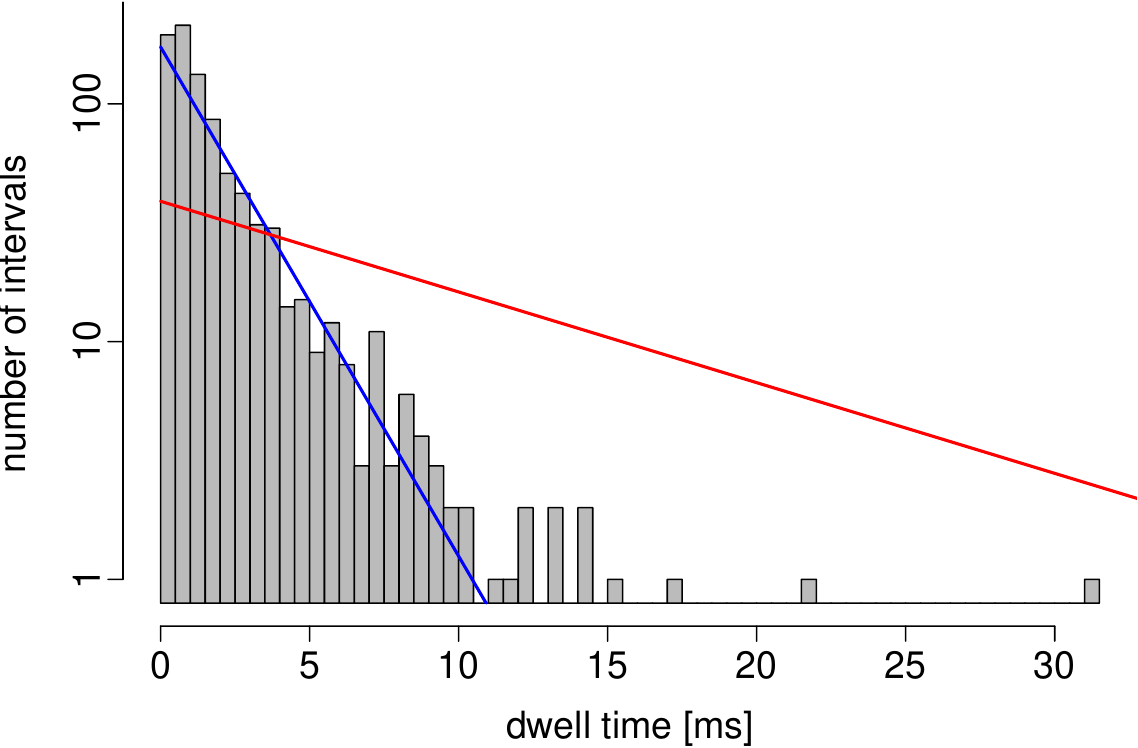}}
  \hspace*{0.31\textwidth}
  \caption{Dwell times in the three lower states compared with model predictions. State $3$ is visited so rarely that no meaningful histogram emerges. The blue line indicates the prediction from the estimated parameters of the VND model, the red curve corresponds to the estimate by the CK model. The histograms display dwell times extracted from the Viterbi path. Note that the vertical axis is logarithmic, which means that the excess of long dwell times over model predictions is less pronounced than it appears. One can clearly see that the predictions of the VND model fit the data much better than the predictions of the CK model. See Appendix~\ref{suppl_sec:dwell_time} for additional illustrations. \label{fig:dwell_times}}
\end{figure}

We provide a qualitative comparison of simulated traces from the two models with the data at two different time scales in Figure~\ref{fig:intro}, which indicates that the VND model visually appears to achieve a better modeling of the data. See also Figure~\ref{fig:plausible_single_traces} in Appendix~\ref{suppl_sec: Intro}. Comparing dwell time distributions from the estimated parameters from both models to the Viterbi path determined from the data in Figure~\ref{fig:dwell_times} also shows a much better fit for the VND model. For a quantitative comparison of the two models, we use the Bayesian information criterion.

\begin{table}
  \caption{Differences of Bayesian information criteria (BIC) between UC and CK model in the first row and UC and VND model in the second row for two ion channel data sequences, provided by the HRCG and described in Section~\ref{subsec:data_sets}. The symbols $\mathcal{L}^{\rm(CK)}$, $\mathcal{L}^{\rm(UC)}$ and $\mathcal{L}^{\rm(VND)}$ denote estimates of log-likelihood for the three models CK, UC, and VND, respectively. The sign of the BIC has been reversed such that larger values indicate a better fit than the UC model. The VND model clearly yields the highest score, i.e., smallest BIC, while the added parameter in the CK model does not achieve a sufficient increase in the likelihood to improve the BIC.
  	\label{tab:likelihood_ratio_with_BIC}}
  \begin{center}
    \fbox{%
      \begin{tabular}{l|c|c}
        negative BIC & Data Set 1 & Data Set 2\\
        \hline
        $2\left( \mathcal{L}^{\rm(CK)} - \mathcal{L}^{\rm(UC)} \right) - \ln(K)$    &    $-88.85$   &  $-42.73$ \\
        $2\left( \mathcal{L}^{\rm(VND)} - \mathcal{L}^{\rm(UC)} \right) - 4\ln(K)$  &   $6526.20$   & $7453.29$
      \end{tabular}
    }
  \end{center}
\end{table}

Table~\ref{tab:likelihood_ratio_with_BIC} shows that the CK model, which can only reveal clustered gating, does not yield an improvement over uncoupled channels, while the VND model clearly does, cf. also Figure~\ref{fig:intro}.

In summary, one can clearly see that the VND model describes the data much better than the uncoupled and the CK model, in accordance with the results of the BIC. In order to see whether the estimated parameters $\widehat{\lambda}_0, \widehat{\lambda}_1, \widehat{\lambda}_2, \widehat{\eta}_1, \widehat{\eta}_2$ and $\widehat{\eta}_3$ support the conclusion of competitive gating, c.f. Definition~\ref{def:coop-comp}, we investigate the ratios $\frac{1-\widehat\lambda_0}{1-\widehat\lambda_1}$, $\frac{1-\widehat\lambda_0}{1-\widehat\lambda_2}$, $\frac{1-\widehat\eta_2}{1-\widehat\eta_1}$ and $\frac{1-\widehat\eta_3}{1-\widehat\eta_1}$. For competitive gating, we expect all four of these quotients to be larger than $1$ since this would indicate that transitions into the state with one channel open are preferred relative to transitions out of this state. In turn, $\frac{1-\widehat\lambda_0}{1-\widehat\lambda_1}$, $\frac{1-\widehat\lambda_0}{1-\widehat\lambda_2}$, $\frac{1-\widehat\eta_3}{1-\widehat\eta_2}$ and $\frac{1-\widehat\eta_3}{1-\widehat\eta_1}$ being smaller than $1$ would indicate cooperative gating, see Definition~\ref{def:coop-comp}.

\begin{table}
  \caption{Estimated parameters for both data sets described in Section~\ref{subsec:data_sets}. 
    From the ratios $\frac{1-\widehat\lambda_0}{1-\widehat\lambda_1}$, $\frac{1-\widehat\lambda_0}{1-\widehat\lambda_2}$, $\frac{1-\widehat\eta_2}{1-\widehat\eta_1}$ and $\frac{1-\widehat\eta_3}{1-\widehat\eta_1}$, which are all much larger than $1$, we can conclude that the gating is competitive.}
  \begin{center}
    \begin{tabular}{l|c|c|c|c|c|c}
      & $1 - \widehat{\lambda}_0$ & $\frac{1-\widehat\lambda_0}{1-\widehat\lambda_1}$ & $\frac{1-\widehat\lambda_0}{1-\widehat\lambda_2}$ & $1 - \widehat{\eta}_1$ & $\frac{1-\widehat\eta_2}{1-\widehat\eta_1}$ & $\frac{1-\widehat\eta_3}{1-\widehat\eta_1}$ \\
      \hline
      Data Set 1 & 0.0082 & 7.800 & 10.590 & 0.0078 & 8.131 &  7.141 \\
      Data Set 2 & 0.0108 & 8.010 &  5.218 & 0.0047 & 9.835 & 12.664
    \end{tabular}
  \end{center}
  \label{tab:estimated_parameters}
\end{table}
In Table \ref{tab:estimated_parameters} we see that $\frac{1-\widehat\lambda_0}{1-\widehat\lambda_1} \gg 1$, $\frac{1-\widehat\lambda_0}{1-\widehat\lambda_2} \gg 1$, $\frac{1-\widehat\eta_2}{1-\widehat\eta_1} \gg 1$ and $\frac{1-\widehat\eta_3}{1-\widehat\eta_1} \gg 1$ for both data sets, which strongly indicates competitive gating. This finding appears to be compatible with the observations by \cite{PDNEFC2012}, which were interpreted in the sense that one channel can act as a ``driver'' for the other channels. In this hypothetical scenario, the ``driver channel'' would have a high probability of opening and low probability of closing, so it will be open most of the time. Other channels would only open while the driver channel is open, albeit with much lower opening probability and much higher closing probability, thus remaining open for much shorter time intervals.


\subsection{Estimating the number of channels}
\label{subsec: est_num_chan}

In Section~\ref{subsec:det_suit_model}, we have chosen $\ell=3$ by simply counting clearly distinguishable levels in the data. In this section, we aim to investigate model selection approaches to select $\ell$, i.e., to determine a methodical,
`data driven' estimate of the number of channels in the VND model. 
As already highlighted in Section \ref{sec:modelsel}, determining $\ell$ from the data is challenging since the entries of $Q^{\textnormal{\rm(VND,$\ell$)}}$ depend only very weakly on $\ell$, in analogy to estimating the number of trials within a binomial distribution with small success probability. In Figure~\ref{fig:number_of_channels} we present the BIC, along with the \emph{penalized maximum likelihood criterion} (PML), and the \emph{integrated complete likelihood criterion} (ICL), which are variants of the BIC using marginal likelihoods. 
We used those criteria and \emph{odd even half sampling} (OEHS) cross validation for the VND model with $\ell \in \{2, \dots, 20\}$ for the assessment of the number of channels.  To have all methods on the same scale, we calculate relative quantities
\begin{align*}
  \mathrm{BIC}(3) &- \mathrm{BIC}(\ell), & \mathrm{PML}(3) &- \mathrm{PML}(\ell), & \mathrm{ICL}(3) &- \mathrm{ICL}(\ell), & \mathrm{OEHS}(3) &- \mathrm{OEHS}(\ell) \, .
\end{align*}
In consequence, the displayed values for $\ell = 3$ are all zero and signs of the displayed values are such that better fits amount to higher values. 
Details to all methods can be found in \citep{celeux2008}.

\begin{figure}[h!]
  \centering
  \subcaptionbox{Data set 1}[0.45\textwidth]{\includegraphics[width=0.45\textwidth]{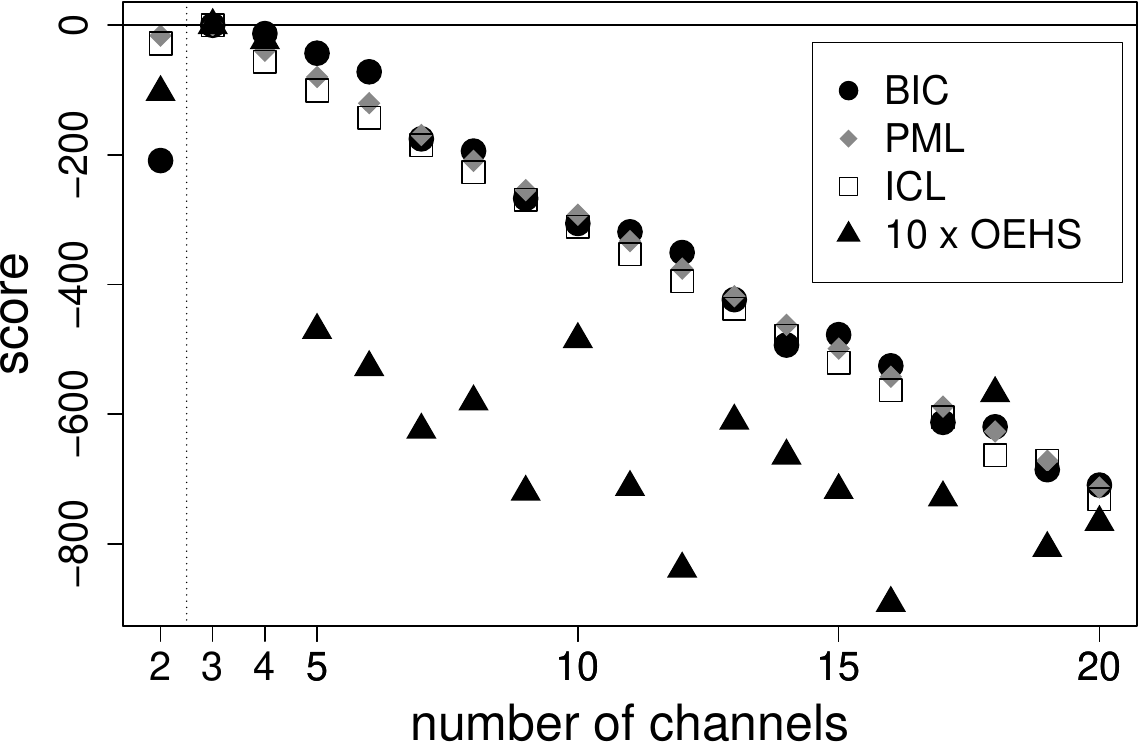}}
  \hspace*{0.02\textwidth}
  \subcaptionbox{Data set 2}[0.45\textwidth]{\includegraphics[width=0.45\textwidth]{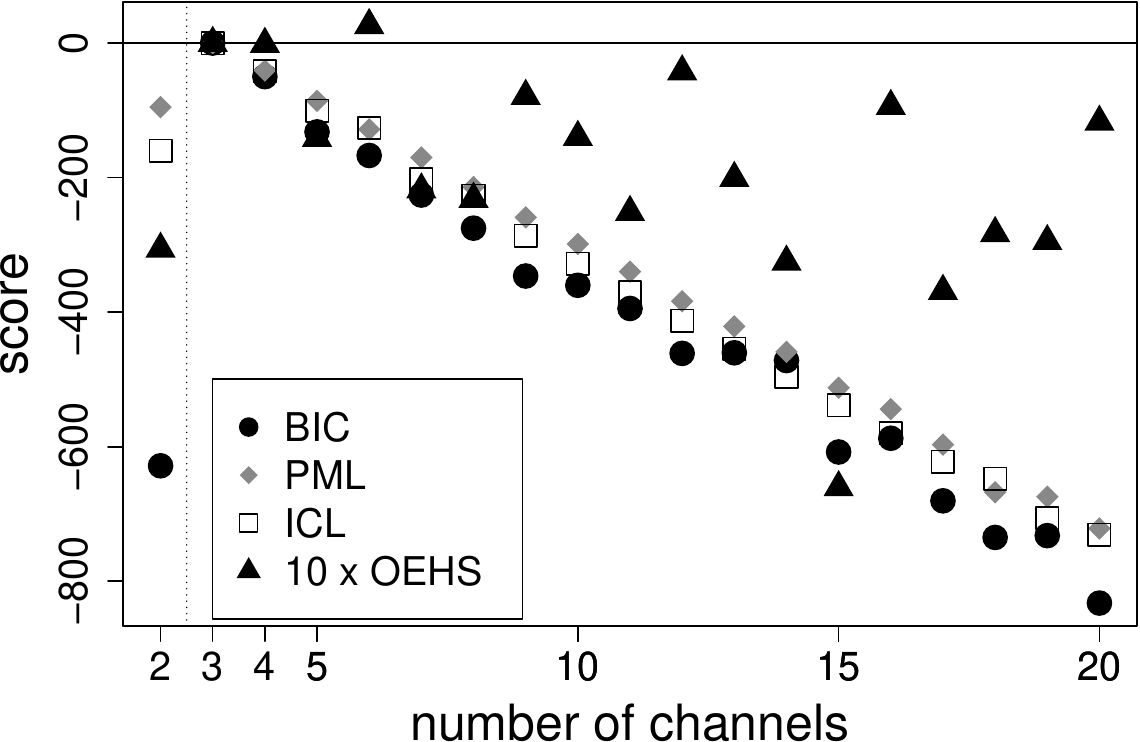}}
  \caption{Several model selection criteria for the VND model with $\ell \in \{2, \dots, 20\}$ applied to the data sets described in Section~\ref{subsec:data_sets} are presented. In order to have comparable absolute values, we display $\mathrm{BIC}(3) - \mathrm{BIC}(\ell)$, $\mathrm{PML}(3) - \mathrm{PML}(\ell)$, $\mathrm{ICL}(3) - \mathrm{ICL}(\ell)$ and $\mathrm{OEHS}(3) - \mathrm{OEHS}(\ell)$. The values for $\ell = 2$ were divided by $10$ so they do not dominate the vertical axis scaling. Values for the cross validation for $\ell \ge 3$ have been inflated by a factor of $10$ to be more easily distinguishable. For both data sets, one can conclude that the model with $\ell=3$ performs best.
  	 \label{fig:number_of_channels}}
\end{figure}

We can see in Figure~\ref{fig:number_of_channels} that the BIC-like measures clearly prefer the model with $\ell=3$ due to the penalty on additional parameters, in accordance with visual inspection. The cross validation results are much less striking but for data set 1 the $\ell=3$ model is still preferred. In the case of data set 2, the $\ell = 6$ model performs slightly better than $\ell = 3$, however in light of the BIC results, we would still suggest the $\ell = 3$ model in that case. This appears to confirm our choice of $\ell =3$ channels above. 
In Appendix~\ref{suppl_sec: est_num_chan} we show that the similarity of VND models for different values of $\ell$ makes determining the number of channels difficult if only at most $j \ll \ell$ channels are open at a time during the measurement and instead the number of channels is generally underestimated to be $\ell = j$.

\subsection{Robustness of results to number of channels}\label{subsec:robust}

In the present experimental system, the number of actively gating channels $\ell$ was estimated using an additional experiment. The Ca${}^{2+}$ concentration on the cis side was increased to 5 $\upmu$M and 1mM of adenosine triphosphate (ATP) was added to fully activate the channels. The maximally elicited transmembrane current was then divided by the single-channel current amplitude. Here, the number of experimentally detected channels results in $\ell= 15$. As the model selection criteria indicate $\ell = 3$, and $\ell=6$ for cross validation and data set 2, they underestimate the number of channels. As discussed in Section~\ref{sec:modelsel}, this situation is not unexpected, since model selection criteria typically estimate the number of channels $\ell$ in the VND model to be the largest number of channels open at a time. If only a smaller number of channels $j < \ell$ is open at a time during the measurement, the number of channels is thus likely to be underestimated. Therefore, we investigate the converse situation namely how robust results of the VND model are, if the number of channels is underestimated. Note that overestimating the number of channels rarely occurs in practice, if the number of levels of a single channel are approximately known.

In order to inspect a situation close to the experimental data, we simulate $1\,000\,000$ data points from systems with different numbers of channels $\ell$ in the UC model and the VND model with different sets of parameters which were chosen such that the highest number of channels open at the same time was $3$. Then we fit a VND model with $3$ channels to the data and inspect the estimated parameters for signs of competitive or cooperative gating, c.f. Definition~\ref{def:coop-comp}. In order to acquire variance estimates, $100$ repetitions were done for each set of parameters. We investigate the ratios $\frac{1-\widehat\eta_3}{1-\widehat\eta_1}$, $\frac{1-\widehat\eta_2}{1-\widehat\eta_1}$, $\frac{1-\widehat\lambda_0}{1-\widehat\lambda_2}$, and $\frac{1-\widehat\lambda_0}{1-\widehat\lambda_1}$, where we expect all of these ratios to be 
larger than $1$ in case of competitive gating since this would indicate that transitions into the state with one channel open are preferred relative to transitions out of this state. Again, we use the shorthand notation $\overline{\lambda} := 1- \lambda$ and $\overline{\eta} := 1- \eta$.

\begin{figure}[h!]
  \centering
  \subcaptionbox{Independent gating}[0.45\textwidth]{\includegraphics[width=0.45\textwidth]{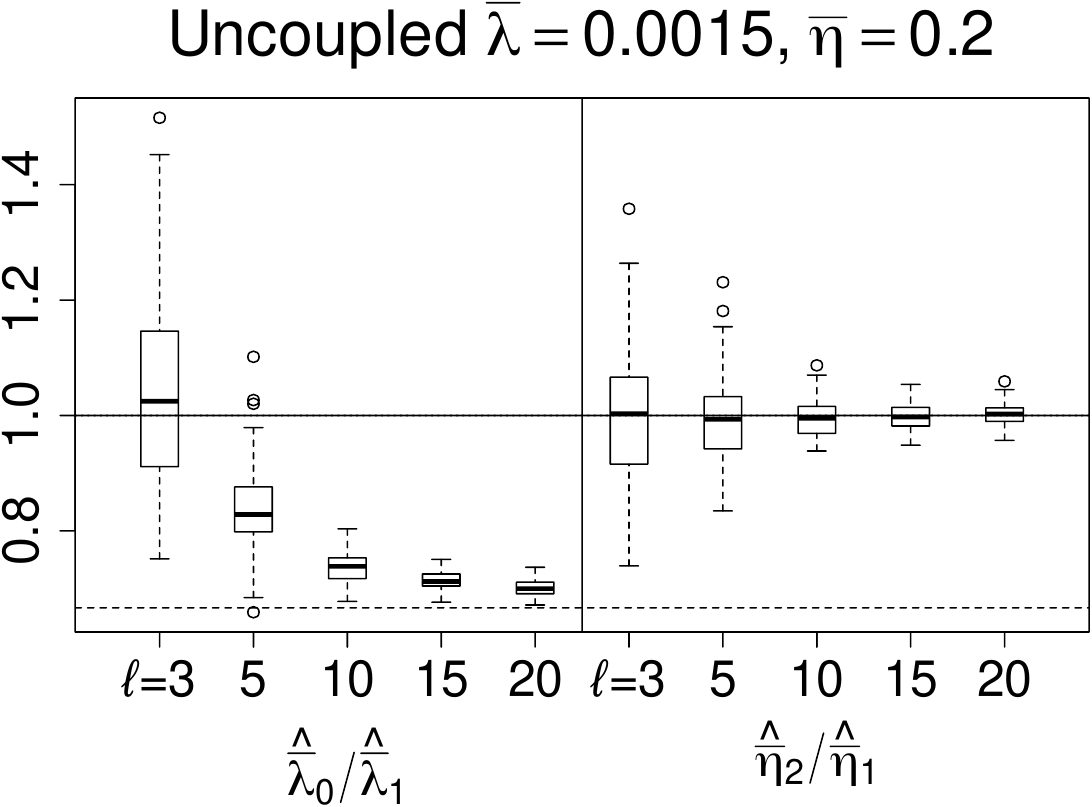}}
  \hspace*{0.02\textwidth}
  \subcaptionbox{Competitive gating}[0.45\textwidth]{\includegraphics[width=0.45\textwidth]{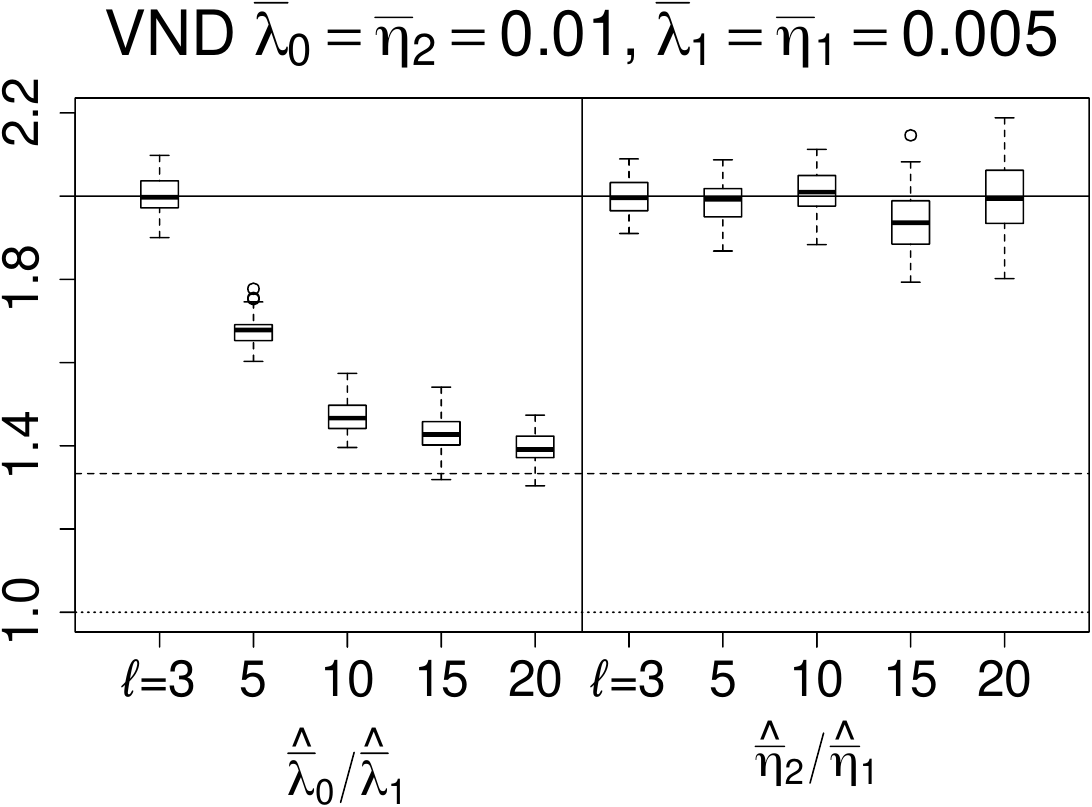}}\\
  \subcaptionbox{Independent gating}[0.45\textwidth]{\includegraphics[width=0.45\textwidth]{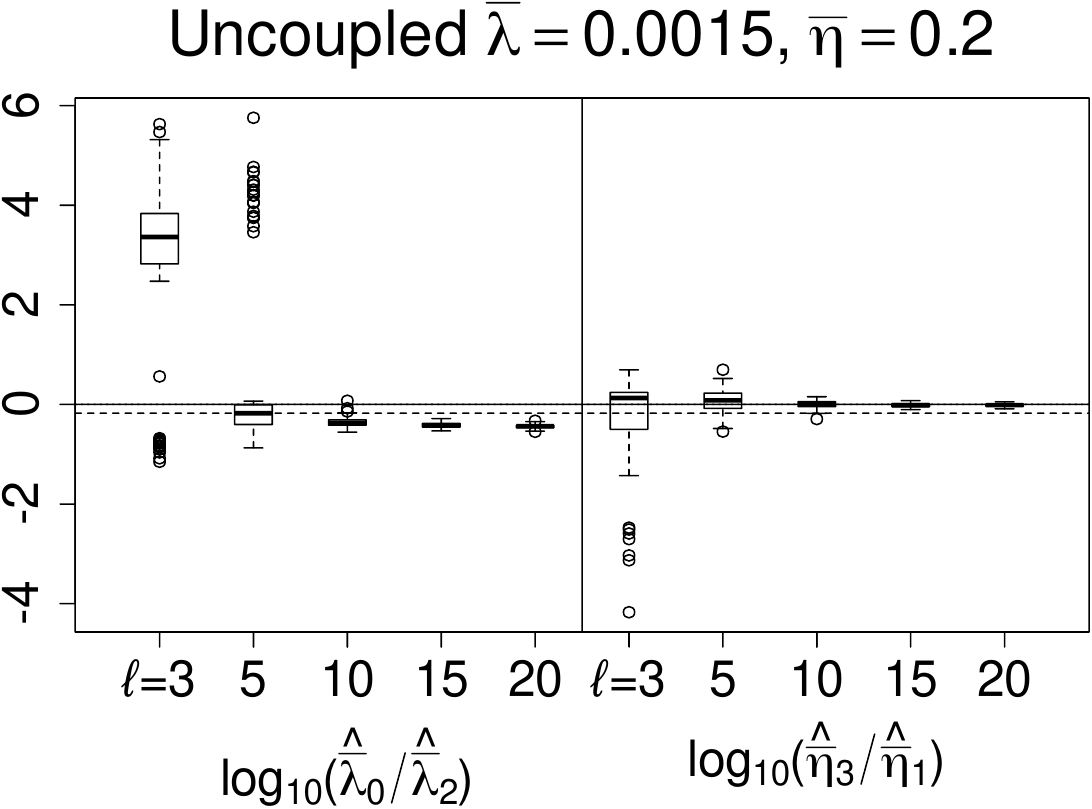}}
  \hspace*{0.02\textwidth}
  \subcaptionbox{Competitive gating}[0.45\textwidth]{\includegraphics[width=0.45\textwidth]{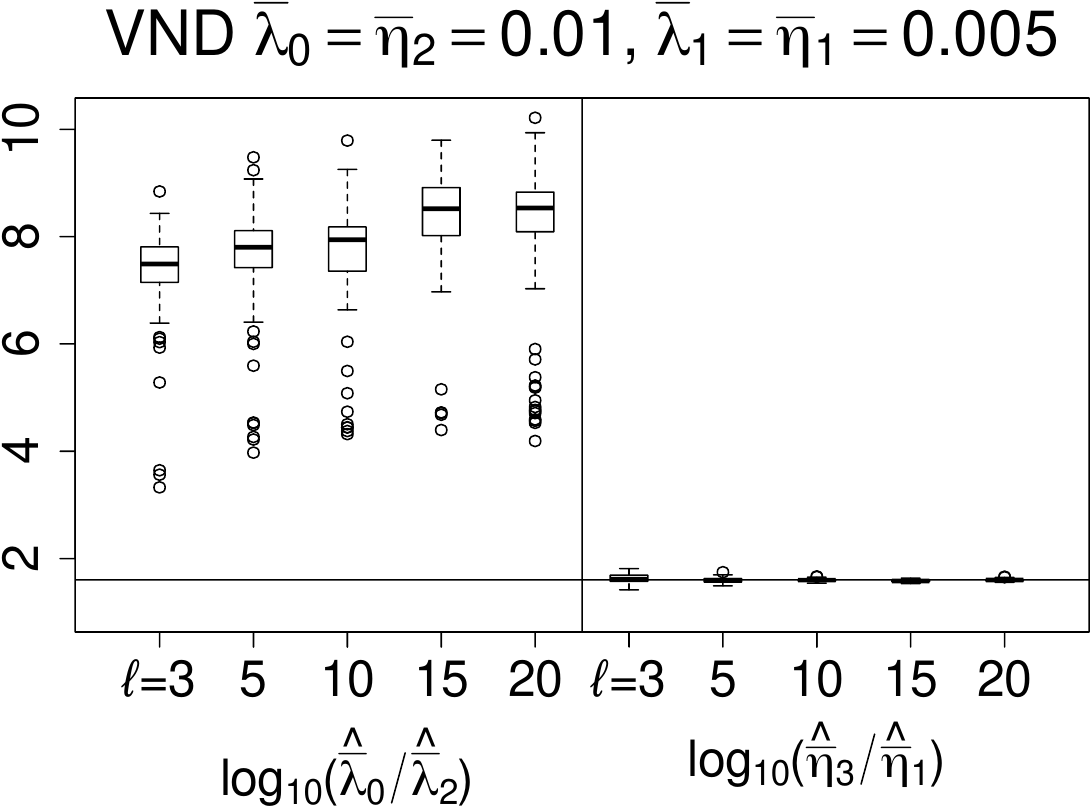}}
  \caption{Boxplots displaying ratios of parameters estimated by a 3 channel VND model for $100$ repetitions of simulations of $1\,000\,000$ data points from a system with $\ell = 3, 5, 10, 15, 20$ channels with the given parameters and $\lambda_k = 1$ for $k > 1$ and $\eta_k = 0.8$ for $k > 2$. For reference, the true value of $\frac{1-\lambda_0}{1-\lambda_1}$ is indicated by a solid line and $\frac{2}{3}\frac{1-\lambda_0}{1-\lambda_1}$ is indicated by a dashed line. 
    The estimated ratio $\frac{1-\widehat\eta_2}{1-\widehat\eta_1}$ is close to the true value in all cases but $\frac{1-\widehat\lambda_0}{1-\widehat\lambda_1}$ is increasingly underestimated with increasing $\ell$, up to a factor of $\sim 2/3$.
    The reason for this is explained in Appendix~\ref{suppl_sec:robustness}. The estimated ratios $\frac{1-\widehat\eta_3}{1-\widehat\eta_1}$ and $\frac{1-\widehat\lambda_0}{1-\widehat\lambda_2}$ exhibit much more variance due to the low population of state $3$. The ratio $\frac{1-\widehat\eta_3}{1-\widehat\eta_1}$ clearly approaches the true values of $1$ and $40$ respectively, however $\frac{1-\widehat\lambda_0}{1-\widehat\lambda_2}$, which should approach $2/3$ and $\infty$, respectively, is underestimated. \label{fig:boxplots}}
\end{figure}
As is clear from the estimated parameters displayed in Figure \ref{fig:boxplots}, the qualitative coupling behavior is recovered. Panel~(a) and panel~(c) show that for uncoupled channels the estimated parameters do not clearly point to either competitive or cooperative gating. Panel~(b) and panel~(d) show that for competitive gating with $\frac{1-\lambda_0}{1-\lambda_1} = \frac{1-\eta_2}{1-\eta_1} = 2$, the probabilities to transition into state $1$ are significantly greater than the probabilities to transition out of this state. This means that the qualitative gating behavior is faithfully recovered even though specific parameter values are not comparable. Since the ratio $\frac{1-\lambda_0}{1-\lambda_1}$ can only be underestimated but not overestimated if the number of channels is underestimated, the conclusion of competitive gating from the data sets we investigate is reinforced. Especially, cooperative gating is conclusively ruled out in our data sets. 
In Appendix~\ref{suppl_sec:robustness} we give a heuristic argument to explain the fact that $\frac{1-\lambda_0}{1-\lambda_1}$ decreases if the ratio between the true number of channels and the number of channels used in the model fit increases. In Appendix~\ref{suppl_sec:coop_gating} we show similar simulations for cooperative gating. In a suitable sense we also find cooperative gating to be robust to underestimation of the number of channels.

\section{Discussion and Outlook}

We have introduced a hidden Markov model for the description of a set of possibly coupled emitters, if only the sum of their signals is measured. We distinguish two main types of emitter interaction within our VND Markov models, namely cooperative and competitive behavior. In case of a competitive model, the resulting signal can often exhibit only few states and thereby disguise the true number of emitters. This can lead to an underestimation of the number of emitters, so it is important to note that the VND model we introduce allows to draw the right conclusion of cooperative or competitive behavior even if the number of emitters is underestimated. In summary, the VND model is a powerful tool to distinguish different types of emitter interaction which is robust to underestimation of the number of emitters.

The VND model faces the typical limitations of homogeneous Markov models in that changes in emitter behavior over time cannot be modeled in this setting. Furthermore, the assumption of permutation invariance states that all emitters are equal in their dynamics and interaction between any pair of emitters is the same. In a hypothetical system of ion channels forming collectively gating clusters but where different clusters gate independently and possibly differently from each other, this assumption would be violated. The requirement of conditional independence goes even further and restricts interaction between emitters to stem exclusively from the overall state of the system which leads to a ``mean field like'' interaction. This assumption is not satisfied by the CK model of clustered gating.

In Section~\ref{subsec:det_suit_model} we show that the VND model can achieve a good fit to real RyR2 ion channel traces and indicates competitive gating. We stress that the competitive phenomenology also appears compatible with the ``driver channel'' hypothesis put forth by \cite{PDNEFC2012} in a system of multiple RyR1 channels, where one channel with very high activity allows for the activation of nearby channels through Calcium transduction. However, from the current time series measured on a membrane with multiple channels as investigated here, one cannot determine whether the finding of competitive gating stems from the presence of a ``driver channel'' or not. The VND model is flexible and can similarly yield a good fit to systems with multiple driver channels, therefore we speculate that it is useful to describe various other ion channel current time series.

In the case of ion channels, the measured signal is usually filtered by a Bessel filter, which leads to additional dependency of the data. This does not affect our simulations, which do not include signal filtering. We disregard filtering in the data sets analyzed here as it will only affect time scales shorter than the filtering time scale of one millisecond. In consequence, the Viterbi algorithm, which aims to determine the underlying Markov state at each time point, may produce flawed results at these short time scales. Especially a transition from state 0 to state 2 or vice versa can easily be mistaken for two transitions, from state 0 to 1 and from 1 to 2 or the other way around. However, such two-step transitions occur only $9$ times upward and $12$ times downward out of a total of $8762$ transitions, thus accounting for $0.5\%$ of all transitions in data set 1 and $21$ times upward and $16$ times downward out of a total of $7341$ transitions, thus accounting for $1\%$ of all transitions in data set 2. Nevertheless, it might be worth extending our model to this situation, see \cite{gunst2001,diehn2019} for existing approaches.

%

%
%

\section*{Acknowledgments}
  We are grateful to Housen Li and Robin Requadt for helpful discussions and validation of our software package.
\section*{Funding}
  The authors acknowledge support of the DFG CRC 803 project Z02, DFG CRC 1456 projects A01, B02, B04, and C06, and the DFG Cluster of Excellence 2067 MBExC. The data set was provided by Lehnart's Lab from the Cellular Biophysics and Translational Cardiology Section in the Heart Research Center G\"ottingen (HRCG). S. E. Lehnart was supported by Deutsche Forschungsgemeinschaft SPP1926 Next Generation Optogenetics.

\appendix

\section{Proofs of Section~2}
For the convenience of the reader we repeat the statements and state their corresponding proofs.

\subsection{Proof of Section~\ref{sec:lump}} 
\label{suppl_sec: proof_lump}

\newtheorem*{thm22}{Theorem \ref{thm:q_dep_on_m}}

\begin{thm22}
  Let $M=(m_{x,y})_{x,y\in\{0,1\}^\ell}$ be the transition matrix of the Markov chain $\Xmp$. If $\Xmp$ satisfies the lumping property, then $(S_k)_{k\in\mathbb{N}}$ is a Markov chain with transition matrix $Q=(q_{i,j})_{i,j\in[\ell\,]}$  given by 
  \begin{equation}
  \label{eq:q_dep_on_m}
  q_{i,j} = \sum_{y\in \mathcal{Z}_j} m_{x,y}
  \end{equation}
  for arbitrary $x\in \mathcal{Z}_i$.
\end{thm22}

The former theorem is an immediate consequence of the following lemma and proposition:
The lemma provides a well known characterization of the lumping property.
\begin{lem} \label{lem: aux_lump}
  For $\Xmp$ satisfying the lumping property is equivalent to
  \begin{equation} \label{eq:equiv_lump}
  \mathbb{P}(S_{k+1}=j\mid S_k=m) =  \mathbb{P}(S_{k+1}=j\mid X_k=x)
  \end{equation}
  for any $k\in\mathbb{N}$, $j,m\in [\ell\,]$ and any $x\in\mathcal{Z}_m$ whenever $\mathbb{P}(S_k=m)\cdot \mathbb{P}(X_k=x)>0$.
\end{lem}
\begin{proof}
  First we show that the lumping property implies \eqref{eq:equiv_lump}.
  We have
  \begin{align*}
  \mathbb{P}(S_{k+1}=j\mid S_k=m) 
  & = \sum_{z \in\mathcal{Z}_m} \frac{\mathbb{P}(S_{k+1}=j,S_k=m,X_k=z)}{\mathbb{P}(S_k=m)} \\
  & = \sum_{z\in\mathcal{Z}_m} \frac{\mathbb{P}(S_{k+1}=j,X_k=z)}{\mathbb{P}(X_k=z)} \frac{\mathbb{P}(X_k=z)}{\mathbb{P}(S_k=m)}\\
  & = \mathbb{P}(S_{k+1}=j\mid X_k=x) \frac{1}{\mathbb{P}(S_k=m)} \sum_{z\in\mathcal{Z}_m} \mathbb{P}(X_k=z)\\
  & = \mathbb{P}(S_{k+1}=j\mid X_k=x),
  \end{align*}
  where we used the lumping property in the second last equality. The other direction is obvious, since in the right-hand side of \eqref{eq:equiv_lump} we can substitute $x$ by any $y\in \mathcal{Z}_m$.
\end{proof}
The proposition provides the fact that 
the lumping property implies that $(S_k)_{k\in\mathbb{N}}$ is again a Markov chain, see  
\cite[Theorem~6.3.2]{kemeny1976}. For convenience of the reader we add the proof.
\begin{prop} \label{thm:lump_leads_to_MC}
  The sequence of random variables $(S_k)_{k\in\mathbb{N}}$ is a homogeneous Markov chain if $\Xmp$ satisfies the lumping property.
\end{prop}

\begin{proof}
  \label{suppl_sec: proof_lump_leads_to_MC}
  We verify the Markov property. 
  For $k\in\mathbb{N}$ and arbitrary $i_1,\dots,i_{k+1}\in [\ell\,]$
  we have
  \begin{align}
  & \mathbb{P}(S_{k+1}=i_{k+1}\mid S_k=i_k,\dots,S_1=i_1)\cdot \mathbb{P}(S_k=i_k,,\dots,S_1=i_1)\nonumber\\
  = & \sum_{x_1\in \mathcal{Z}_{i_1}} \dots \sum_{x_{k+1}\in\mathcal{Z}_{i_{k+1}}} \nonumber\\
  & \qquad \qquad
  \mathbb{P}(S_{k+1}=i_{k+1},X_{k+1}=x_{k+1},S_k=i_k,X_k=x_k,\dots,S_1=i_1,X_1=x_1)\nonumber\\
  = & \sum_{x_1\in \mathcal{Z}_{i_1}} \dots \sum_{x_{k+1}\in\mathcal{Z}_{i_{k+1}}}
  \mathbb{P}(X_{k+1}=x_{k+1},X_k=x_k,\dots,X_1=x_1)  \nonumber\\
  = & 
  \sum_{x_1\in \mathcal{Z}_{i_1}} \dots \sum_{x_{k+1}\in\mathcal{Z}_{i_{k+1}}}
  \mathbb{P}(X_{k+1}=x_{k+1}\mid X_k=x_k) \mathbb{P}(X_k=x_k,\dots,X_1=x_1). \label{eqn:lump_prop_proof}
  \end{align}
  Here in the last equality we used the Markov property of the Markov chain $\Xmp$.
  For $x_k\in\mathcal{Z}_{i_k}$ we have by Lemma~\ref{lem: aux_lump}
  \begin{align*}
  \sum_{x_{k+1} \in \mathcal{Z}_{i_{k+1}}}
  \mathbb{P}(X_{k+1}=x_{k+1}\mid X_k=x_k) & 
  = \mathbb{P}(S_{k+1}=i_{k+1}\mid X_k=x_k)\\
  & = \mathbb{P}(S_{k+1}=i_{k+1}\mid S_k=i_k).
  \end{align*}
  By plugging this in (\ref{eqn:lump_prop_proof}) and the fact that 
  \[
  \mathbb{P}(S_k=i_k,,\dots,S_1=i_1) = \sum_{x_1\in \mathcal{Z}_{i_1}} \dots \sum_{x_{k}\in\mathcal{Z}_{i_{k}}} \mathbb{P}(X_k=x_k,\dots,X_1=x_1)
  \]
  the assertion is proven.
\end{proof}

\subsection{Proofs of Section~\ref{sec:perm_inv}}
\label{suppl_sec: proof_pi_lump}

\newtheorem*{prop26}{Proposition \ref{prop:pi_lump}}

\begin{prop26}
  If $\Xmp$ is permutation invariant, it satisfies the lumping property.
\end{prop26}
\begin{proof}
  Let $i,j\in [\ell\,]$ and $x,x' \in\mathcal{Z}_i$. 
  Note that $x,x'\in\mathcal{Z}_i$ implies that the vectors $x,x'$ have the same number of ``$1$'' entries and therefore, there exists a permutation matrix $P\in\{0,1\}^{\ell\times\ell}$ such that $P x = x'$. Then
  \begin{align*}
  \mathbb{P}(S_{k+1}=j\mid X_k=x)&=\sum_{y\in\mathcal{Z}_j}\mathbb{P}(X_{k+1}=y \mid X_k=x)\\
  &=\sum_{y\in\mathcal{Z}_j}\mathbb{P}(X_{k+1}=P y\mid X_k=P x)\\
  &=\sum_{y\in\mathcal{Z}_j}\mathbb{P}(X_{k+1}=y\mid X_k=x')
  =\mathbb{P}(S_{k+1}=j|X_k=x'). 
  \end{align*}
\end{proof}

Permutation invariance can lead to a considerable simplification of the transition matrix of the Markov chain $\Xmp$, which we justify with the following lemma that can be straightforwardly used to determine a first upper bound of the number of free parameters of the transition matrix.

\begin{lem}\label{lem:carPI}
  For $x_1,y_1,x_2,y_2\in\{0,1\}^\ell$ with 
  \begin{equation}\label{eq:pereq}
  \Vert x_1 \Vert_1 = \Vert x_2 \Vert_1,\quad \Vert y_1 \Vert_1 = \Vert y_2 \Vert_1,\quad \mbox{and }\quad \Vert x_1-y_1 \Vert_1 = \Vert x_2-y_2 \Vert_1,
  \end{equation}
  there exists a permutation matrix $P\in\{0,1\}^{\ell\times\ell}$ such that
  \begin{equation}
  \label{eq:ex_P}
  x_1 = Px_2 \quad \text{and} \quad y_1=P y_2.
  \end{equation}
  In particular, this implies for a permutation invariant Markov chain $\Xmp$ that
  \[
  \mathbb{P}(X_{k+1}=y_1\mid X_k=x_1) = \mathbb{P}(X_{k+1}=y_2\mid X_k=x_2),
  \quad k\in\mathbb{N}_0.
  \]
\end{lem}
\begin{proof}
  For $x,y\in \set$ define the $\{0,1\}^2$-valued vector 
  \[
  \mathbf{z}_{x,y} := (x,y)^T = ((x^{(1)},y^{(1)}),\dots,(x^{(\ell)},y^{(\ell)}))^T. 
  \]	
  Then, there exists a permutation matrix $P_{x,y}\in \{0,1\}^{\ell\times \ell}$ with 
  $n^{(1)},\dots,n^{(4)}\in [\ell\,]$
  and
  $\sum_{i=1}^{4}n^{(i)} = \ell$
  such that the entries of $P_{x,y}\, \mathbf{z}_{x,y}$ satisfy
  \[
  ( P_{x,y}\, \mathbf{z}_{x,y} )^{(i)} = 
  \begin{cases}
  (0,0) & 1\leq i \leq n^{(1)}\\
  (0,1) & 1+n^{(1)}\leq i \leq n^{(1)}+n^{(2)}\\
  (1,0) & 1+n^{(1)}+n^{(2)} \leq i \leq 
  \sum_{j=1}^{3} n^{(j)}
  \\
  (1,1) & 
  1+\sum_{j=1}^{3} n^{(j)}
  \leq i \leq \ell.
  \end{cases}
  \]	
  Note that the permutation $P_{x,y}$ orders the pairs within the vector $\mathbf{z}_{x,y}$ according to the lexicographic semiorder.	
  %
  
  For the given $x_1,y_1,x_2,y_2\in\{0,1\}^\ell$ satisfying \eqref{eq:pereq} we consider the permutation matrices $P_{x_{1},y_{1}}$ with the numbers $n_1^{(j)} \in [\ell\,]$ for $j=1,\dots,4$ and $P_{x_{2},y_{2}}$ with $n_2^{(j)} \in [\ell\,]$, such that
  $\sum_{j=1}^{4} n_i^{(j)} =\ell$ for $i=1,2$.
  Observe that if we are able to verify that $P_{x_1,y_1}\,\mathbf{z}_{x_1,y_1} = P_{x_2,y_2}\,\mathbf{z}_{x_2,y_2}$, then \eqref{eq:ex_P} follows with
  \[
  P = P_{x_1,y_1}^{-1} P_{x_2,y_2}. 
  \]
  Therefore, to prove the statement of \eqref{eq:ex_P} it is sufficient to show that $n_1^{(j)}=n_2^{(j)}$ for $j=1,\dots,4$.
  %
  For this note that \eqref{eq:pereq} implies 
  \begin{align*}
  n_1^{(1)}+n_1^{(2)}
  &= \ell- \Vert x_1 \Vert_1 = \ell-\Vert x_2 \Vert_1 = n_2^{(1)}+n_2^{(2)}\\
  n_1^{(1)}+n_1^{(3)}
  &= \ell - \Vert y_1 \Vert_1 = \ell-\Vert y_2 \Vert_1 =
  n_2^{(1)}+n_2^{(3)}\\
  n_1^{(1)}+n_1^{(4)}
  &= \ell - \Vert x_1-y_1 \Vert_1 = \ell - \Vert x_2 - y_2 \Vert_1 = n_2^{(1)}+n_2^{(4)}.
  \end{align*}
  Having those three equations and using 
  $\sum_{i=1}^4 n_1^{(i)} = \sum_{i=1}^4n_2^{(i)} = \ell$
  yields to the desired fact that $n_1^{(j)}=n_2^{(j)}$ for $j=1,\dots,4$.
  
  As a consequence with the invariance property for $k\in\mathbb{N}_0$ we get
  \begin{align*}
  \mathbb{P}(X_{k+1}=y_1\mid X_k=x_1) \underset{\eqref{eq:ex_P}}{=} 
  \mathbb{P}(X_{k+1}= Py_2\mid X_k= Px_2) 
  =  \mathbb{P}(X_{k+1}=y_2\mid X_k=x_2),
  \end{align*}
  which finishes the proof.
\end{proof}
Now we are able to prove the claimed statement of the exact number of free parameters.

\newtheorem*{prop27}{Proposition \ref{prop:perm-inv-param}}

\begin{prop27}
  The transition matrix $M$ of a permutation invariant Markov chain $\Xmp$ is determined by $\ell(\ell+1)(\ell+5)/6$ parameters.
\end{prop27}
\begin{proof}
  For $j,k\in [\ell]$ define
  \[
  n_{j,k} := \left\vert \left\{ \Vert x-y \Vert_1 \colon x\in \mathcal{Z}_j, y\in \mathcal{Z}_k   \right\} \right \vert,
  \]
  which denotes the number of different possible values of the sum of ``1''s within the difference of vectors with $j$ and $k$ non-zero entries.
  For $x,y\in \{0,1\}^\ell$ with $x\in \mathcal{Z}_j$ and $y\in \mathcal{Z}_k$ it is clear that switching all ``$0$'' and ``$1$'' entries in $x$ and $y$ does not change $\Vert x-y \Vert_1$, such that $n_{j,k}=n_{\ell-j,\ell-k}$. Similarly one can conclude that
  \[
  n_{j,k} = n_{\ell-j,k} = n_{j,\ell-k}. 
  \]
  %
  Furthermore, $n_{j,k} = n_{k,j}$, because $\Vert x-y \Vert_1$ is symmetric under interchange of $x$ and $y$. 
  For $j \le k \le \frac{\ell-1}{2}$ we obtain
  \begin{align*}
  n_{j,k} = \left| \{ 0, 1, \dots , j\} \right| = j+1,
  \end{align*}
  and taking the symmetries explained above into account we get for any $j,k\in [\ell]$ that
  \begin{align*}
  n_{j,k} = \min\{j+1, k+1, \ell+1-j, \ell+1-k\} \, .
  \end{align*}
  Thus, by Lemma~\ref{lem:carPI} the number $N_{\textnormal{PI}}$ of (possibly) different transition matrix entries for a permutation invariant Markov chain is
  \begin{align*}
  N_{\textnormal{PI}}(\ell) &= \sum_{j=0}^{\ell} \limits \sum_{k=0}^{\ell} \limits \min\{j+1, k+1, \ell+1-j, \ell+1-k\}\\
  &= (\ell+1)(\ell+2)(\ell+3)/6.
  \end{align*}
  The latter equality is shown by induction using the middle-split
  \begin{align*}
  N_{\textnormal{PI}}(\ell + 1) &= \sum_{\substack{j=0 \\ j\neq \left\lfloor\frac{\ell}{2}+1\right\rfloor}}^{\ell+1} \limits \sum_{\substack{k=0 \\ k\neq \left\lfloor\frac{\ell}{2}+1\right\rfloor}}^{\ell+1} \limits \min\{j+1, k+1, \ell+2-j, \ell+2-k\} + \frac{(\ell+2)(\ell+3)}{2} \\
  &= \sum_{j=0}^{\ell} \limits \sum_{k=0}^{\ell} \limits \min\{j+1, k+1, \ell+1-j, \ell+1-k\} + \frac{(\ell+2)(\ell+3)}{2}\\
  & = N_{\textnormal{PI}}(\ell) + \frac{(\ell+2)(\ell+3)}{2} \, ,
  \end{align*}
  which yields the result. The number of independent entries/parameters in the transition matrix is further reduced, since all rows of the matrix sum up to one. The first entry in every row can therefore be considered dependent on the other entries in the row. As there are $\ell+1$ independent parameters in the first column, the total number of independent parameters is
  \[
  (\ell+1)(\ell+2)(\ell+3)/6 - 6 (\ell+1) / 6 = \ell(\ell+1)(\ell+5)/6 \, .
  \qedhere
  \]
\end{proof}

\subsection{Proofs of Section~\ref{subsec: VND}}
\label{suppl_sec:proof_app_lem}


We require the following auxiliary results, where we remind the reader that the commonly underlying probability space is $(\Omega,\mathcal{F},\mathbb{P})$.
\begin{lem}\label{lem:condp} Let $A,B_1,\ldots,B_k\in\mathcal{F}$ with pairwise disjoint 
  $B_1,\ldots,B_k$.
  Assume that $\mathbb{P}(B_i)= \mathbb{P}(B_j)$  and $\mathbb{P}(A \mid B_i)=\mathbb{P}(A\mid B_j)$ for all $i,j\in\{1,\dots,k\}$. Then 
  \[\mathbb{P}\left(A \mid \bigsqcup_{j=1}^kB_j \right) = \mathbb{P}(A|B_i)\]
  for all $i\in\{1,\dots,k\}$. (Here $\bigsqcup$ denotes a union of pairwise disjoint sets.)
\end{lem}
\begin{proof}
  It is sufficient to show the statement for $k=2$, since then the rest follows inductively. For $k=2$ we have	
  \begin{align*}
  \mathbb{P}(A\mid B_1\sqcup B_2) &= \frac{\mathbb{P}(A\cap(B_1\sqcup B_2))}{\mathbb{P}(B_1\sqcup B_2)}
  = \frac{\mathbb{P}((A\cap B_1)\sqcup (A\cap B_2))}{\mathbb{P}(B_1\sqcup B_2)}\\
  &= \frac{\mathbb{P}(A\cap B_1)+ \mathbb{P}(A\cap B_2)}{2\mathbb{P}(B_1)}
  = \frac{\mathbb{P}(A\cap B_1)}{2\mathbb{P}(B_1)}+ \frac{\mathbb{P}(A\cap B_2)}{2\mathbb{P}(B_2)}\\
  & = (1/2)\mathbb{P}(A\mid B_1)+(1/2)\mathbb{P}(A\mid B_2)\\
  &= \mathbb{P}(A\mid B_1) = \mathbb{P}(A\mid B_2).
  \end{align*}
\end{proof}

\begin{lem} \label{lem: init_implies_Xk}
  Let $\Xmp$ be a permutation invariant vector Markov chain  with
  \begin{equation*}
  \mathbb{P}(X_0 = y) 
  =\mathbb{P}(X_0=Py),  \quad y\in\{0,1\}^{\ell},
  \end{equation*}
  for a permutation matrix $P\in\{0,1\}^{\ell \times \ell}$. Then, for any $k\in \mathbb{N}_0$ we have
  \begin{equation}
  \label{eq:prob_equal}
  \mathbb{P}(X_k=y) = \mathbb{P}(X_k= P y), \quad y\in\{0,1\}^{\ell}.
  \end{equation}
\end{lem}
\begin{proof}
  We prove the statement by induction over $k$. Note that, by assumption, \eqref{eq:prob_equal} holds for $k=0$. If it is true for $k$, then
  \begin{align*}
  \mathbb{P}(X_{k+1}=y)& = \sum_{x\in\{0,1\}^\ell}\mathbb{P}(X_{k+1}=y \mid X_{k}=x)\mathbb{P}(X_{k}=x)\\
  &=\sum_{x\in\{0,1\}^\ell}\mathbb{P}(X_{k+1}=P y\mid X_{k}=P x)\, \mathbb{P}(X_{k}=P x)\\
  &=\mathbb{P}(X_{k+1}=P y),
  \end{align*}
  which verifies \eqref{eq:prob_equal} for $k+1$ and finishes the proof.
\end{proof}
\begin{lem}  \label{lem:cond_prob_repre_vnd}
  Let $\Xmp$ be a permutation invariant vector Markov chain with permutation invariant initial distribution, that is, 
  \begin{equation}
  \mathbb{P}(X_0 = y) 
  =\mathbb{P}(X_0=Py),
  \end{equation}
  for any permutation matrix $P\in\{0,1\}^{\ell \times \ell}$ and any $ y\in\{0,1\}^{\ell}$. Then, for any ${i\in\{1,\dots,\ell\}}$, $x\in \{0,1\}^\ell$ and $k\in\mathbb{N}_0$ we have
  \begin{equation}
  \label{eq:to_show}
  \mathbb{P}(X_{k+1}^{(i)}=y^{(i)} \mid X_k=x) 
  = \mathbb{P}(X_{k+1}^{(i)}=y^{(i)} \mid X^{(i)}_k=x^{(i)}, \Vert X_k \Vert_1=\Vert x\Vert_1 )
  \end{equation}
  with $x=(x^{(1)},\dots,x^{(\ell)})^T$ and $y=(y^{(1)},\dots,y^{(\ell)})^T$.
\end{lem}

\begin{proof}
  Define the set
  \[
  I_{x,i} := \{ z\in\mathcal{Z}_{\Vert x \Vert_1}\mid x^{(i)} = z^{(i)}  \}.
  \]
  We aim to apply Lemma~\ref{lem:condp} with $A=\{X_{k+1}^{(i)} = y^{(i)}\}$ and $B_z = \{X_k=z\}$ for $z\in I_{x,i}$. For any $z\in I_{x,i}$ we have a permutation matrix $P$ such that $Px = z$
  and $(P z')^{(i)} = {z'}^{(i)}$ for any $z'\in \{0,1\}^\ell$. (The last condition says that the permutation does not change the $i$th coordinate entry.) 
  Then, with Lemma~\ref{lem: init_implies_Xk} we have
  \[
  \mathbb{P}(B_x) = \mathbb{P}(X_k=x) = \mathbb{P}(X_k = Px) = \mathbb{P}(B_z).
  \]
  In addition to that, by the permutation invariance we obtain
  \begin{align*}
  \mathbb{P}(A \mid B_x) 
  & = \mathbb{P}(X_{k+1}^{(i)} = y^{(i)} \mid X_k=x)
  = \mathbb{P}(X_{k+1}^{(i)} = (Py)^{(i)} \mid X_k= Px)\\
  & = \mathbb{P}(X_{k+1}^{(i)} = y^{(i)} \mid X_k=z)
  = \mathbb{P}(A \mid B_z).
  \end{align*}
  Furthermore
  \[
  \bigsqcup_{z\in I_{x,i}} B_z 
  = \bigsqcup_{z\in I_{x,i}} \{X_k=z\}
  = \{X_k^{(i)}=x^{(i)}, \Vert X_k \Vert_1=\Vert x\Vert_1\}.
  \]
  Thus, by the application of Lemma~\ref{lem:condp} we have \eqref{eq:to_show}.
\end{proof}
An immediate consequence is the following.
\begin{rem} \label{rem:const_in_k}
  Under the assumptions of the previous lemma we have 
  by the fact that $\Xmp$ is a homogeneous Markov chain that the expression
  \[
  \mathbb{P}(X_{k+1}^{(i)}=y^{(i)} \mid X^{(i)}_k=x^{(i)}, \Vert X_k \Vert_1=\Vert x\Vert_1 )
  \]
  is independent of $k\in\mathbb{N}_0$.
\end{rem}

Now we have all the tools for proving the VND characterization theorem.

\newtheorem*{thm218}{Theorem \ref{thm:vnd-char}}

\begin{thm218}[Characterization of VND Markov chain]
  For a vector Mar\-kov chain $\Xmp$ assume that the initial distribution is permutation invariant\footnote{The distribution of $X_0$ is permutation invariant if $\mathbb{P}(X_0=y) = \mathbb{P}(X_0=Py)$ for any $y\in \{0,1\}^\ell$ and any permutation matrix $P\in \{0,1\}^{\ell\times\ell}$.}.
  Then, the following statements are equivalent:
  \begin{enumerate}
    \item The Markov chain $\Xmp$ is vector norm dependent;
    \item The Markov chain $\Xmp$ is permutation invariant and conditionally independent.
  \end{enumerate}
\end{thm218}

\begin{proof}
  By Proposition~\ref{prop:vnd_perm_inv} a vector norm dependent Markov chain is permutation invariant. Furthermore, by the definition of vector norm dependence and Lemma~\ref{lem:cond_prob_repre_vnd} the conditional independence property is satisfied.
  
  We turn to the other direction. 
  Let $\Xmp$ be permutation invariant and conditionally independent. Then, for all $x,y\in\set$ we have
  \[
  \mathbb{P}(X_{k+1}=y\mid X_k=x)=\prod_{i=1}^\ell\mathbb{P}(X_{k+1}^{(i)}=y^{(i)}\mid X_k=x),
  \]
  such that by Lemma~\ref{lem:cond_prob_repre_vnd} the equality of (\ref{eq:m_vnd_def}) of Definition~\ref{def:vnd} follows.
  By Remark~\ref{rem:const_in_k}, the transition probabilities (\ref{eq:vnd_const_map}) of Definition~\ref{def:vnd} are constant in $k$, so that we only need to argue that they are also constant w.r.t. $i\in\{1,\dots,\ell\}$. Let $y_i,\overline{y}_i \in \{0,1\}^\ell$ with
  \begin{align*}
  y_i & = (y^{(1)},\dots,y^{(i-1)},1,y^{(i+1)},\dots,y^{(\ell)})^T,\\
  \overline{y}_i 
  & =  (y^{(1)},\dots,y^{(i-1)},0,y^{(i+1)},\dots,y^{(\ell)})^T,
  \end{align*}
  that is, $y_i$ and $\overline{y}_i$ differ only in the $i$th entry of the vector. Furthermore, let $P_{i,j}\in\{0,1\}^{\ell\times \ell}$ be the permutation which only permutes the $i$-th and $j$-th entry of a vector. Then for any $x\in \{0,1\}^{\ell}$ with $x=(x^{(1)},\dots,x^{(\ell)})^T$ we have
  \begin{align*}
  & \frac{\mathbb{P}(X_{k+1}^{(i)}=1\mid X_k^{(i)}=x^{(i)}, \Vert X_k \Vert_1 = \Vert x \Vert_1)}{\mathbb{P}(X_{k+1}^{(i)}=0\mid X_k^{(i)}=x^{(i)}, \Vert X_k \Vert_1 = \Vert x \Vert_1)}
  \overset{\eqref{eq:to_show}}{=}
  \frac{\mathbb{P}(X_{k+1}^{(i)}=y_i^{(i)}\mid X_k =x)}{\mathbb{P}(X_{k+1}^{(i)}=\overline{y}_i^{(i)}\mid X_k =x)} 	\\
  =	& \frac{\prod_{j=1}^{\ell} 			\mathbb{P}(X_{k+1}^{(j)}=y_i^{(j)}\mid X_k=x)}{\prod_{j=1}^{\ell} \mathbb{P}(X_{k+1}^{(j)}=
    \overline{y}_i^{(j)}\mid X_k=x)}  	
  = \frac{\mathbb{P}(X_{k+1}=y_i \mid X_k=x)}{\mathbb{P}(X_{k+1}=\overline{y}_i \mid X_k=x)} \\
  = 	& \frac{\mathbb{P}(X_{k+1}=P_{i,j} y_i \mid X_k=P_{i,j}x)}{\mathbb{P}(X_{k+1}=P_{i,j}\overline{y}_i \mid X_k=P_{i,j}x)}	
  = \frac{\prod_{j=1}^{\ell} 			\mathbb{P}(X_{k+1}^{(j)}=P_{i,j} y_i^{(j)}\mid X_k=P_{i,j} x)}{\prod_{j=1}^{\ell} \mathbb{P}(X_{k+1}^{(j)}=
    P_{i,j} \overline{y}_i^{(j)}\mid X_k=P_{i,j} x)} \\
  =   & 
  \frac{\mathbb{P}(X_{k+1}^{(j)}=(P_{i,j} y_i)^{(j)}\mid X_k =P_{i,j} x)}{\mathbb{P}(X_{k+1}^{(j)}=(P_{i,j}\overline{y}_i)^{(j)}\mid X_k = P_{i,j} x)}
  \overset{\eqref{eq:to_show}}{=} \frac{\mathbb{P}(X_{k+1}^{(j)}=1\mid X_k^{(j)}=x^{(i)}, \Vert X_k \Vert_1 = \Vert x \Vert_1)}{\mathbb{P}(X_{k+1}^{(j)}=0\mid X_k^{(j)}=x^{(i)}, \Vert X_k \Vert_1 = \Vert x \Vert_1)},
  \end{align*}
  where we used the conditional independence and permutation invariance. Taking into account that for any $j\in\{1,\dots,\ell\}$ holds
  \begin{align*}
  &	\mathbb{P}(X_{k+1}^{(j)}=1\mid X_k^{(j)}=x^{(i)}, \Vert X_k \Vert_1 = \Vert x \Vert_1) \\
  &\qquad + \mathbb{P}(X_{k+1}^{(j)}=0\mid X_k^{(j)}=x^{(i)}, \Vert X_k \Vert_1 = \Vert x \Vert_1) =1,
  \end{align*}
  leads to
  \begin{align*}
  &~ \frac{1}{\mathbb{P}(X_{k+1}^{(i)}=0\mid X_k^{(i)}=x^{(i)}, \Vert X_k \Vert_1 = \Vert x \Vert_1)} -1\\
  =&~ \frac{1}{\mathbb{P}(X_{k+1}^{(j)}=0\mid X_k^{(j)}=x^{(i)}, \Vert X_k \Vert_1 = \Vert x \Vert_1)} -1
  \end{align*}
  for any $i,j\in \{1,\dots,\ell\}$, any value $x^{(i)} \in \{0,1\}$ and $\Vert x \Vert_1 \in [\ell\,]$. In consequence, simplifying the former expression, we get for any $i,j\in \{1,\dots,\ell\}$, $r\in [\ell\,]$ and $b\in \{0,1\}$ that
  \begin{align*}
  \mathbb{P}(X_{k+1}^{(i)}=b \mid X^{(i)}_k=b, \Vert X_k \Vert_1=r ) =\mathbb{P}(X_{k+1}^{(j)}=b \mid X^{(j)}_k=b, \Vert X_k \Vert_1=r )\,,
  \end{align*}
  which finishes the proof.
\end{proof}
We add a representation of the transition matrix of the ``sum'' Markov chain $(S_k)_{k\in\mathbb{N}}$ based on a vector norm dependent Markov chain $\Xmp$.

\newtheorem*{prop221}{Proposition \ref{prop:repres_Q_vnd}}

\begin{prop221}
  Given a vector norm dependent Markov chain $\Xmp$ with transition matrix $M^{(\rm VND)}$ determined by the parameters $\lambda_0,\dots,\lambda_{\ell-1},\eta_1,\dots,\eta_{\ell}\in [0,1]$, see Remark~\ref{rem:param_vnd} in the main text. Then,
  the transition matrix $Q=(q_{i,j})_{i,j\in[\ell\,]}$ of the corresponding ``sum'' Markov chain $(S_k)_{k\in\mathbb{N}}$ is given by
  \begin{align*}
  q_{i,j} := \sum_{r=\max\{0,i-j\}}^{\min\{i,\ell-j\}} {i\choose r} {\ell-i\choose j-i+r} \eta_{i}^{i-r}(1-\eta_{i})^{r} \lambda_i^{\ell-j-r}(1-\lambda_i)^{j-i+r},
  \end{align*}
  for any $i,j\in[\ell\,]$, where for completeness we set $\eta_0:=1$ and $\lambda_\ell:=1$.
\end{prop221}
%
\begin{proof}
  For arbitrary $i\in[\ell\,]$ set
  \[
  x=(\underbrace{1,\dots,1}_{i \text{ times}},\underbrace{0,\dots,0}_{\ell-i \text{ times}})
  \]
  and note that $x\in\mathcal{Z}_i$. Thus, for $j\in[\ell\,]$ by the lumping property we have
  \[
  q_{i,j} = \mathbb{P}(S_{k+1}=j\mid X_k=x).
  \]
  The idea is to decompose the event $X_k=x$ into different sets in such a way that we can use counting problem arguments. For this define the random variable
  \[
  R:=\text{``The number of 1s from $x$ that become 0s at time $k+1$''}.
  \]
  Observe that the random variable $R$ takes only values in $\{\max\{0, i-j\},\dots,\min\{i,\ell-j\}\}$, since trivially $0\leq R\leq i$, $R \ge i-j$ because the number of ones at time $k+1$ is $j$ and $R\leq \ell-j$ follows by the fact that the number of zeros at time $k+1$ is $\ell-j$. Furthermore, we have
  \[
  q_{i,j}= \sum_{r=\max\{0, i-j\}}^{\min\{i,\ell-j\}} \mathbb{P}(S_{k+1}=j,R=r\mid X_k=x).
  \]
  
  For $r\in\{ \max\{0, i-j\},\dots,\min\{i,\ell-j\}  \}$ we have
  \begin{align} \label{eq:repr_ell}
  \MoveEqLeft[4] \mathbb{P}(S_{k+1}=j,R=r\mid X_k=x) =\nonumber\\
  & {i\choose r} {\ell-i\choose j-i+r} \eta_{i}^{i-r}(1-\eta_{i})^{r} \lambda_i^{\ell-j-r}(1-\lambda_i)^{j-i+r},
  \end{align}
  which is justified as follows. Recall that $R=r$ means that there were $r$ ones that became zeros and observe that the number of cases when that happens is ${i\choose r} {\ell-i\choose j-i+r}$. (This is true since there are ${i \choose r}$ possibilities for $r$ ones to become zeros and eventually, to have $j$ ones at time $k+1$, there are ${\ell-i \choose j-i+r}$ possibilities of zeros which become ones.) Finally, by taking into account that the probabilities of $i-r$ ones to remain ones is $\eta_{i+1}^{i-r}$, of $r$ ones to become zeros is $(1-\eta_{i+1})^r$, of $j-i+r$ zeros to become ones is $(1-\lambda_i)^{j-i+r}$ and of $\ell-j-r$ zeros to remain zeros is $\lambda_i^{\ell-j-r}$, the representation of \eqref{eq:repr_ell} is verified. 
\end{proof}

\subsection{Proof of Section~\ref{subsec: inv_lump_ident}} \label{sec:proof-params}

\newtheorem*{thm223}{Theorem \ref{thm:params-unique}}

\begin{thm223}	[Inverse lumping identifiability for VND Markov chains]
  Let $(S_k)_{k\in\mathbb{N}}$ be a ``sum'' Markov chain with transition matrix $Q^{\rm(VND)}$ on $[\ell\,]$ based on a vector norm dependent Markov chain $\Xmp$ with transition matrix $M^{\rm(VND)}$. If $\ell$ is odd, then the parameters $\lambda_j$ and $\eta_{j+1}$ with $j=0,\dots,\ell-1$ defining $M^{(\rm VND)}$ are uniquely determined by the entries of $Q^{\rm(VND)}$. If $\ell$ is even, the same holds true provided that $\lambda_{\ell/2} \geq 1-\eta_{\ell/2}$.
\end{thm223}

\begin{proof}
  In Proposition~\ref{prop:repres_Q_vnd} we provided a functional representation of the entries of $Q=(q_{s,r})_{s,r\in [\ell\,]}$ in terms of the parameters $\lambda_0,\dots,\lambda_{\ell-1},\eta_{1},\dots,\eta_{\ell}\in [0,1]$. The idea is to exploit this structure. 
  
  First, observe that $q_{0,0} = \lambda_0^\ell$ and $q_{\ell,\ell} = \eta_\ell^\ell$ such that $\lambda_0$ and $\eta_\ell$ are uniquely determined. If $\ell = 1$, this concludes the proof, so that in all of the following we can assume $\ell \ge 2$. For $i=1,2$ let $\lambda_{0,(i)},\dots,\lambda_{\ell-1,(i)},\eta_{1,(i)},\dots,\eta_{\ell,(i)}\in[0,1]$ be solutions of the inverse lumping problem (of course with $\lambda_{0,(i)}=q_{0,0}^{1/\ell}$ and $\eta_{\ell,(i)}=q_{\ell,\ell}^{1/\ell}$). For $s\in\{1,\dots,\ell-1\}$ let us use the notation $\widetilde\eta_{s,(i)}:= 1-\eta_{s,(i)}$ and note that from Proposition~\ref{prop:repres_Q_vnd} it follows that
  \begin{align}
  \label{eq:s,1}
  \widetilde{\eta}_{s,(1)}^s \lambda_{s,(1)}^{\ell-s} = q_{s,0} = \widetilde{\eta}_{s,(2)}^s \lambda_{s,(2)}^{\ell-s}. 
  \end{align}
  We immediately see that $\widetilde{\eta}_{s,(1)} = \widetilde{\eta}_{s,(2)}$ iff $\lambda_{s,(1)} = \lambda_{s,(2)}$. Now, for $s\in\{1,\dots,\ell-1\}$ assume that the pair $(\widetilde \eta_{s,(1)},\lambda_{s,(1)})$ is different from $(\widetilde \eta_{s,(2)},\lambda_{s,(2)})$, which leads to the fact that without loss of generality we have $\lambda_{s,(1)}>0$.  Therefore $x:= \lambda_{s,(2)}/\lambda_{1,(2)} \in [0,\infty)$. Additionally by $s\in\{1,\dots,\ell-1\}$ we have $\alpha:= \ell/s \in (1,\ell]$. By exploiting \eqref{eq:s,1} we obtain $\widetilde{\eta}_{s,(2)} = \widetilde{\eta}_{s,(1)} x^{1-\alpha}$ and trivially $\lambda_{s,(2)} = \lambda_{s,(1)} x$, which gives
  \begin{align}
  \label{eq:aux1}
  x - \widetilde{\eta}_{s,(2)} x &= (x^{\alpha} - \widetilde{\eta}_{s,(1)} x) x^{1-\alpha},\\
  \label{eq:aux2}
  x -\lambda_{s,(2)} x &= (1 - \lambda_{s,(1)} x) x \, .
  \end{align}
  Taking the form of $q_{s,1}$ from Proposition~\ref{prop:repres_Q_vnd} into account yields
  \begin{align*}
  & s \widetilde{\eta}_{s,(1)}^{s-1} (x-\widetilde{\eta}_{s,(1)} x) \lambda_{s,(1)}^{\ell-s} + (\ell - s) \widetilde{\eta}_{s,(1)}^s \lambda_{s,(1)}^{\ell-s-1} (x - \lambda_{s,(1)} x) = x q_{s,1} \\ 
  & =  s \widetilde{\eta}_{s,(2)}^{s-1} (x-\widetilde{\eta}_{s,(2)} x) \lambda_{s,(2)}^{\ell-s} + (\ell - s) \widetilde{\eta}_{s,(2)}^s \lambda_{s,(2)}^{\ell-s-1} (x - \lambda_{s,(2)} x).
  \end{align*}
  Using \eqref{eq:aux1} and \eqref{eq:aux2} on the left-hand side of the previous equality gives
  \begin{align*}
  & s \widetilde{\eta}_{s,(1)}^{s-1} (x-\widetilde{\eta}_{s,(1)} x) \lambda_{s,(1)}^{\ell-s} + (\ell - s) \widetilde{\eta}_{s,(1)}^s \lambda_{s,(1)}^{\ell-s-1} (x - \lambda_{s,(1)} x) \\ 
  & = s \widetilde{\eta}_{s,(1)}^{s-1} (x^\alpha - \widetilde{\eta}_{s,(1)} x) \lambda_{s,(1)}^{\ell-s} + (\ell - s) \widetilde{\eta}_{s,(1)}^s \lambda_{s,(1)}^{\ell-s-1} (1 - \lambda_{s,(1)} x).
  \end{align*}
  Further transformations of this yield
  \begin{align}
  &x^\alpha = x + \frac{\ell - s}{s} \frac{\widetilde{\eta}_{s,(1)}}{\lambda_{s,(1)}} (x - 1) \, . \label{eq:xalpha}
  \end{align}
  It is clear that $x=1$ is a solution to equation \eqref{eq:xalpha} and that there is at most one other solution. To find a simple expression for the other solution for general $\alpha$, we consider the representation of $q_{s,2}$ (again) from Proposition~\ref{prop:repres_Q_vnd}. We have
  \begin{align*}
  &  s(s-1) \widetilde{\eta}_{s,(1)}^{s-2} (x-\widetilde{\eta}_{s,(1)} x)^2 \lambda_{s,(1)}^{\ell-s}\\
  &+ 2s(\ell - s) \widetilde{\eta}_{s,(1)}^{s-1} \lambda_{s,(1)}^{\ell-s-1} (x-\widetilde{\eta}_{s,(1)} x)(x - \lambda_{s,(1)} x)\\
  &+ (\ell - s)(\ell-s-1) \widetilde{\eta}_{s,(1)}^s \lambda_{s,(1)}^{\ell-s-2} (x - \lambda_{s,(1)} x)^2 =  2 x^2 q_{s,2}\\  
  = & s(s-1) \widetilde{\eta}_{s,(2)}^{s-2} (x-\widetilde{\eta}_{s,(2)} x)^2 \lambda_{s,(2)}^{\ell-s}\\
  &+ 2s(\ell - s) \widetilde{\eta}_{s,(2)}^{s-1} \lambda_{s,(2)}^{\ell-s-1} (x-\widetilde{\eta}_{s,(2)} x)(x - \lambda_{s,(2)} x)\\
  &+ (\ell - s)(\ell-s-1) \widetilde{\eta}_{s,(2)}^s \lambda_{s,(2)}^{\ell-s-2} (x - \lambda_{s,(2)} x)^2.
  \end{align*}
  Note that in the special case $s=1$ the first term on either side vanishes and the remaining terms are exactly those given in Proposition~\ref{prop:repres_Q_vnd}. By \eqref{eq:aux1} and \eqref{eq:aux2} the right-hand side can be further modified such that
  \begin{align*}
  & s(s-1) \widetilde{\eta}_{s,(1)}^{s-2} (x-\widetilde{\eta}_{s,(1)} x)^2 \lambda_{s,(1)}^{\ell-s}\\
  &+ 2s(\ell - s) \widetilde{\eta}_{s,(1)}^{s-1} \lambda_{s,(1)}^{\ell-s-1} (x-\widetilde{\eta}_{s,(1)} x)(x - \lambda_{s,(1)} x)\\
  &+ (\ell - s)(\ell-s-1) \widetilde{\eta}_{s,(1)}^s \lambda_{s,(1)}^{\ell-s-2} (x - \lambda_{s,(1)} x)^2\\
  =& s(s-1) \widetilde{\eta}_{s,(1)}^{s-2} (x^\alpha-\widetilde{\eta}_{s,(1)} x)^2 \lambda_{s,(1)}^{\ell-s}\\
  &+ 2s(\ell - s) \widetilde{\eta}_{s,(1)}^{s-1} \lambda_{s,(1)}^{\ell-s-1} (x^\alpha-\widetilde{\eta}_{s,(1)} x)(1 - \lambda_{s,(1)} x)\\
  &+ (\ell - s)(\ell-s-1) \widetilde{\eta}_{s,(1)}^s \lambda_{s,(1)}^{\ell-s-2} (1 - \lambda_{s,(1)} x)^2.
  \end{align*}
  By plugging \eqref{eq:xalpha} in, the previous expression reduces to a quadratic equation in $x$. 
  One solution is again $x=1$ and the other solution is 
  \[x = \frac{\ell \widetilde{\eta}_{s,(1)}}{\ell \widetilde{\eta}_{s,(1)} -2s \widetilde{\eta}_{s,(1)}+2s \lambda_{s,(1)}} = \frac{\alpha \widetilde{\eta}_{s,(1)}}{(\alpha -2) \widetilde{\eta}_{s,(1)} + 2 \lambda_{s,(1)}}.\] 
  Defining $y:= \frac{\lambda_{s,(1)}}{\widetilde{\eta}_{s,(1)}}$, we obtain 
  $x = \frac{\alpha}{2y  + \alpha -2}$ and from $x \in [0, \infty)$ we get $y \in (1 - \frac{\alpha}{2}, \infty]$. Substituting the second solution back into equation \eqref{eq:xalpha} yields
  \begin{align}
  \left(\frac{\alpha}{2y + \alpha -2}\right)^\alpha = \frac{1}{y}\frac{(2\alpha -2) - (\alpha - 2)y}{2 y + \alpha -2} \label{eq:doesthishold?}
  \end{align}
  which can be rewritten as
  \begin{align*}
  f_\alpha(y) := \big(2y + \alpha -2 \big)^{\alpha-1} \big( (\alpha - 2)y - 2 (\alpha-1) \big) + \alpha^\alpha y = 0 \, .
  \end{align*}
  Note that $f_\alpha$ is continuous on the whole domain for every $\alpha > 1$. Straightforward calculations give
  \begin{align*}
  f'_\alpha(y) =& \alpha \big(2y + \alpha -2 \big)^{\alpha-2} \big( 2(\alpha - 2)y - 3 \alpha + 4 \big) + \alpha^\alpha,\\
  f''_\alpha(y) =& 4\alpha(\alpha-1)(\alpha-2)\big(2y + \alpha -2 \big)^{\alpha-3} \big( y-1 \big),\\
  f'''_\alpha(y) =& 4\alpha(\alpha-1)(\alpha-2)\big(2y + \alpha -2 \big)^{\alpha-4} \big( 2(\alpha - 2)y - \alpha + 4 \big),\\
  f'''_\alpha(1) =& 4\alpha^{\alpha-2}(\alpha-1)(\alpha-2),\\
  g_\alpha(\epsilon) :=& f'_\alpha\left(1-\frac{\alpha}{2}+\epsilon\right) = -\alpha^2 (\alpha -1) \epsilon^{\alpha-2} - 2\alpha^2 (2 - \alpha) \epsilon^{\alpha-1} + \alpha^\alpha.
  \end{align*}
  We treat three cases:
  \vspace*{0.5\baselineskip}
  
  \noindent\textbf{Case 1: $\alpha > 2$}\\
  Here, $1- \alpha/2 < 0$ and thus $y\geq 0$.
  We have $f''_\alpha(y) = 0$ if and only if $y=1$ and $f'''_\alpha(1) > 0$, so $f'_\alpha(y)$ has a unique global minimum at $y=1$. Since $f'_\alpha(1) = 0$ this means that $f'_\alpha(y) > 0$ for all $0 < y < 1$ and $y>1$. Since $f_\alpha(1) = 0$ and $\lim_{y \to \infty} f_\alpha(y) = \infty$, we conclude that $y=1$ is the only zero.\vspace*{0.5\baselineskip}
  
  \noindent\textbf{Case 2: $\alpha = 2$}\\
  Note that this case only occurs if $\ell$ is even. Here, $f_2(y) = 0$ for any $y>0$ and we have 
  \[
  x = \frac{1}{y} \, \Longleftrightarrow \, \frac{\lambda_{s,(2)}}{\lambda_{s,(1)}} = \frac{\widetilde{\eta}_{s,(1)}}{\lambda_{s,(1)}}  \, \Longleftrightarrow \, \lambda_{s,(2)} = \widetilde{\eta}_{s,(1)} \, \Longleftrightarrow \, \lambda_{s,(1)} = \widetilde{\eta}_{s,(2)}. 
  \] 
  If $\lambda_{s,(1)} \ge \widetilde{\eta}_{s,(1)}$,
  we get $\lambda_{s,(2)} = \widetilde{\eta}_{s,(1)} \le \lambda_{s,(1)} = \widetilde{\eta}_{s,(2)}$ which means that the second solution is invalid unless both solutions are the same, which is excluded by assumption. Thus only one of the solutions is valid.\vspace*{0.5\baselineskip}
  
  \noindent\textbf{Case 3: $1 < \alpha < 2$}\\
  Here, $1- \alpha/2 > 0$ and thus 
  $y >1- \alpha/2$. 
  We have $f''_\alpha(y) = 0$ if and only if $y=1$ and $f'''_\alpha(1) < 0$, so $f'_\alpha(y)$ has a local maximum at $y=1$. For completeness, we notice 
  \begin{align*}
  \lim_{y \searrow 1-\frac{\alpha}{2}} \limits f'_\alpha(y) = \lim_{\epsilon \searrow 0} \limits g_\alpha(\epsilon) = \lim_{\epsilon \searrow 0} \limits \left(-\alpha^2 (\alpha -1) \epsilon^{\alpha-2} - 2\alpha^2 (2 - \alpha) \epsilon^{\alpha-1} + \alpha^\alpha \right) = - \infty \, .
  \end{align*}
  This means that $f'_\alpha(y) < 0$ for 
  $1-\alpha/2<y<1$ and $y>1$.
  Since $f_\alpha(1) = 0$, we conclude that $f_\alpha(y) < 0$ for $y>1$ and $f_\alpha(y) > 0$ for $1 - \frac{\alpha}{2} < y < 1$. Furthermore, note that
  \[
  \lim_{y \to \infty} f_\alpha(y) = \lim_{y \to \infty} \Big( (\alpha - 2) 2^{\alpha-1} y^{\alpha} + \alpha^\alpha y \Big) = (\alpha - 2) 2^{\alpha-1} \lim_{y \to \infty} y^{\alpha} = - \infty
  \]
  thus we conclude that $y=1$ is the only zero.
  \vspace*{0.5\baselineskip}
  
  This proves the claim of uniqueness since in all three cases we obtain $\lambda_{s,(1)}=\lambda_{s,(2)}$ and therefore $\eta_{s,(1)}=\eta_{s,(1)}$ which contradicts the assumption that the tuples $(\lambda_{s,(i)},\eta_{s,(i)})$ with $i=1,2$ are different to each other.
\end{proof}

\section{Generalization of cooperative and competitive VND Markov chain models} \label{suppl_sec:comp-coop-def}
\begin{defi} \label{def:coop-comp-gen}
  For $i \in \{2, \dots, \ell\}$ we call a VND Markov chain $\Xmp$ with transition matrix $M^{(\text{VND})}$, determined by $\lambda_j,\eta_{j+1}$ with $j\in\{0,\dots,\ell-1\}$ (see \eqref{eq:trans_p_r} and \eqref{eq:trans_q_r}), $i$-cooperative if for all $a \in \{1, \dots, i-1\}$ holds
  \begin{align*}
  \frac{1-\lambda_{a}}{1- \lambda_0} &> 1 & \frac{1-\eta_a}{1- \eta_i} &> 1 \, .
  \end{align*}
  In contrast to that for $i \in \{1, \dots, \ell-1\}$ we call such a VND Markov chain $i$-competitive if for all $a \in \{1, \dots, i\}$ and $b \in \{i, \dots, \ell-1\}$  holds\begin{align*}
  \frac{1-\lambda_{a-1}}{1- \lambda_b} &> 1 & \frac{1-\eta_{b+1}}{1- \eta_a} &> 1 \, .
  \end{align*}
\end{defi}

As indicated by the terms, $i$-cooperativeness implies that the state with zero or $i$ emitters signaling are preferred, which can be interpreted as a group of $i$ emitters having a tendency to occupy the same state. On the other hand, $i$-competitiveness indicates that emitters compete over being on with a preferred number of $i$ open channels. Given these generalized definitions, we can easily show the following incompatibility result.
\begin{thm}
  For $i\in\{2,\dots,\ell\}$ and $j\in\{1,\dots,\ell-1\}$
  with $i \neq j$ a VND Markov chain cannot be $i$-cooperative and $j$-competitive.
\end{thm}

\begin{proof}
  For $j < i$ a $j$-competitive VND Markov chain satisfies $\frac{1-\lambda_{0}}{1- \lambda_j} > 0$ and if it is also $i$-competitive we have $\frac{1-\lambda_{j}}{1- \lambda_0} > 0$, since $j \le i-1$, which is a contradiction.
  
  For $i < j$ a $j$-competitive VND Markov chain satisfies $\frac{1-\eta_a}{1- \eta_j} > 0$ for all $a \in \{1, \dots, i\}$ and if it is also  $i$-competitive we have $\frac{1-\eta_{j}}{1- \eta_a} > 0$ for all $a \in \{1, \dots, i\}$, since $i \le j-1$, which yields a contradiction.
\end{proof}

In order to fix terminology, we call a VND Markov chain which is both $j$-competitive and $j$-cooperative \emph{purely $j$-cooperative}, since in that case only the states $0$ and $j$ are highly populated, which indicates that the dynamics are fully dominated by a cluster of $j$ emitters.

\newpage

\section{Illustration to Section~3}
\label{suppl_sec:max}
We consider the case of $\ell=2$ and use the polynomials
\begin{align}
P_0(\theta_H)&=\sum_{i=0}^{2}n_{0,i}\log q^{(\text{VND})}_{0,i}(\lambda_0),\label{eq:pol0}\\
P_1(\theta_H)&=\sum_{i=0}^{2}n_{1,i}\log q^{(\text{VND})}_{1,i}(\lambda_1, \eta_1),\label{eq:pol1}\\ P_2(\theta_H)&=\sum_{i=0}^{2}n_{2,i}\log q^{(\text{VND})}_{2,i}(\eta_2),\label{eq:pol2}
\end{align}
where the number of state changes from $i$ to $j$ in $S_{1:K}$ is denoted as $n_{i,j}$. With that we have $\sum_{k=0}^2P_k(\theta_H) = \sum_{k=1}^{K-1}\log q^{(\text{VND})}_{s_k,s_{k+1}}(\theta_H)$, which provides a representation of the ``middle'' term of the log-likelihood function 
\begin{align}
\label{al: log_lik_disc}
\ell(\theta,s_{1:K}, y_{1:K}) &= \log\pi^{(s_1)} + \sum_{k=1}^{K-1}\log q^{(\text{VND})}_{s_k,s_{k+1}}(\theta_H) + \sum_{k=1}^K\log g_{\theta_E}(y_k,s_k).
\end{align}
This log-likelihood function is fundamental for the Baum-Welch algorithm, in particular it leads to $\theta\mapsto f(\theta,\theta_t, y_{1:K})$ which is the function that is used in the maximization step, described in Section \ref{sec:BW_alg} of the main text. From \eqref{al: log_lik_disc} we obtain that for unimodality of $\theta\mapsto f(\theta,\theta_t, y_{1:K})$ it is sufficient to show that $\sum_{k=0}^2P_k(\theta_H)$ is strictly unimodal, since the two other terms are strictly concave in the Gaussian scenario we are interested in. Therefore, in Figure~\ref{fig:concavity} we plot $P_0$, $P_1$ and $P_2$ to provide an indication of their behavior. In general $\sum_{k=0}^2P_k(\theta_H)$ is not strictly unimodal, but a sufficient condition may be to restrict to $\lambda_1 \geq 1-\eta_1$.

\begin{figure}[h!]
  \centering
  \includegraphics[width=\textwidth]{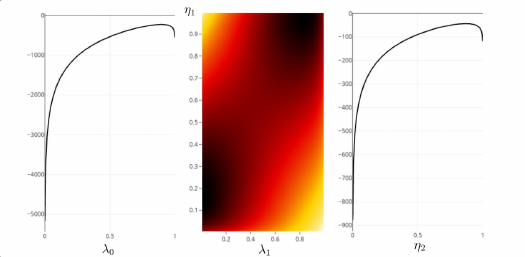}
  \caption{\label{fig:concavity} 
    From a simulation of 1000 points from the VND model with parameters (0.9,0.8,0.98,0.89) plots of the polynomials $P_0$ (\ref{eq:pol0}), $P_1$ (\ref{eq:pol1}), and $P_2$ (\ref{eq:pol2}) respectively. The plot of $P_1$ is depicted here as a heat map. Although, this heat map shows two maxima (in black), if one restricts to $\lambda_1\geq 1- \eta_1$ the resulting polynomial is strictly unimodal.}
\end{figure}

\FloatBarrier
\pagebreak

\section{Additional Comments on Dwell Time Estimation}\label{suppl_sec:dwell_time}

In order to visually determine time constants of dwell time distributions depicted in Figure~\ref{fig:dwell_times} one can examine histograms of dwell times with logarithmic $x$-axis. To see this, consider $z := \log(x)$ and the exponential distribution
\begin{align*}
\phi(x) dx =&~ \frac{1}{\tau} \exp(-x/\tau) dx = \frac{1}{\tau} \exp(-\exp(z)/\tau) \exp(z) dz\\
=&~ \frac{1}{\tau} \exp(z-\exp(z)/\tau) dz =: \phi_{\log}(z) dz \, .
\end{align*}
Then $\phi_{\log}'(z) = 0$, i.e. the maximum of the transformed density, is at $\exp(z) = x = \tau$.

However, as the histograms in Figure~\ref{fig:dwell_times_alt} show, this approach does not work for state 0 in our case since the time constants of the other states are too close to the measurement time intervals. This leads to binning artifacts and an uncertainty of more than a factor of $2$ of the time constants.

\begin{figure}[h!]
  \centering
  \subcaptionbox*{data set $2$, state $0$, binning 1}[0.41\textwidth]{\includegraphics[width=0.4\textwidth]{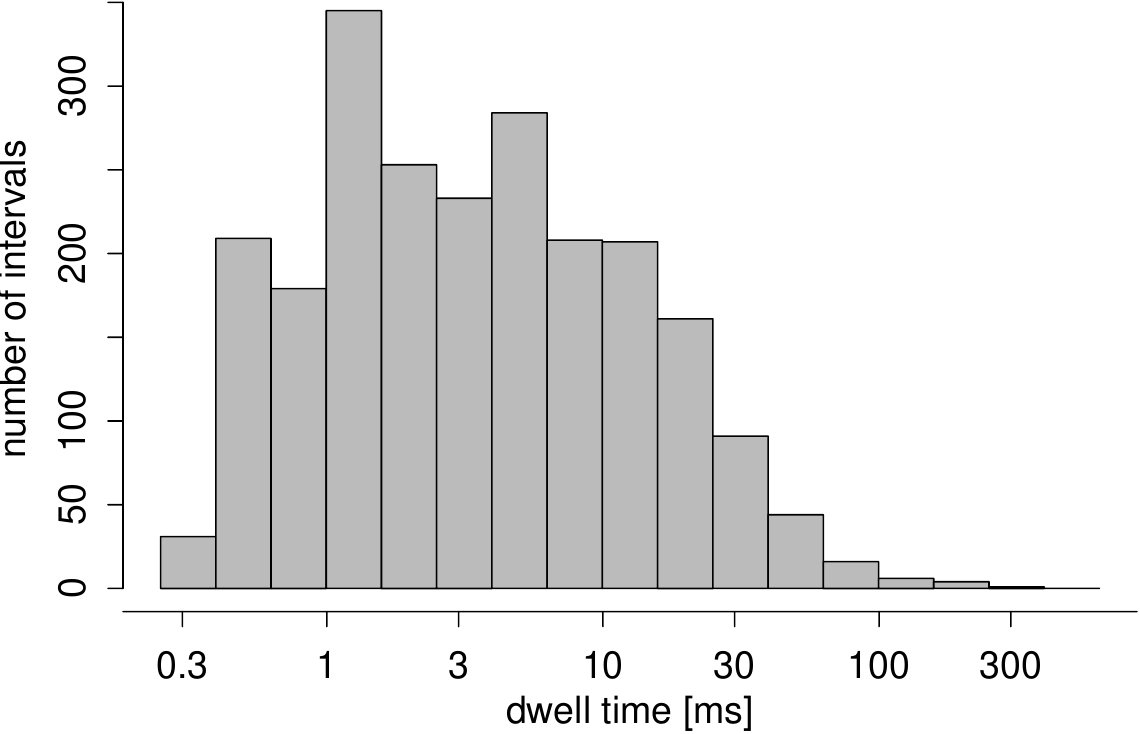}}
  \subcaptionbox*{data set $2$, state $0$, binning 2}[0.41\textwidth]{\includegraphics[width=0.4\textwidth]{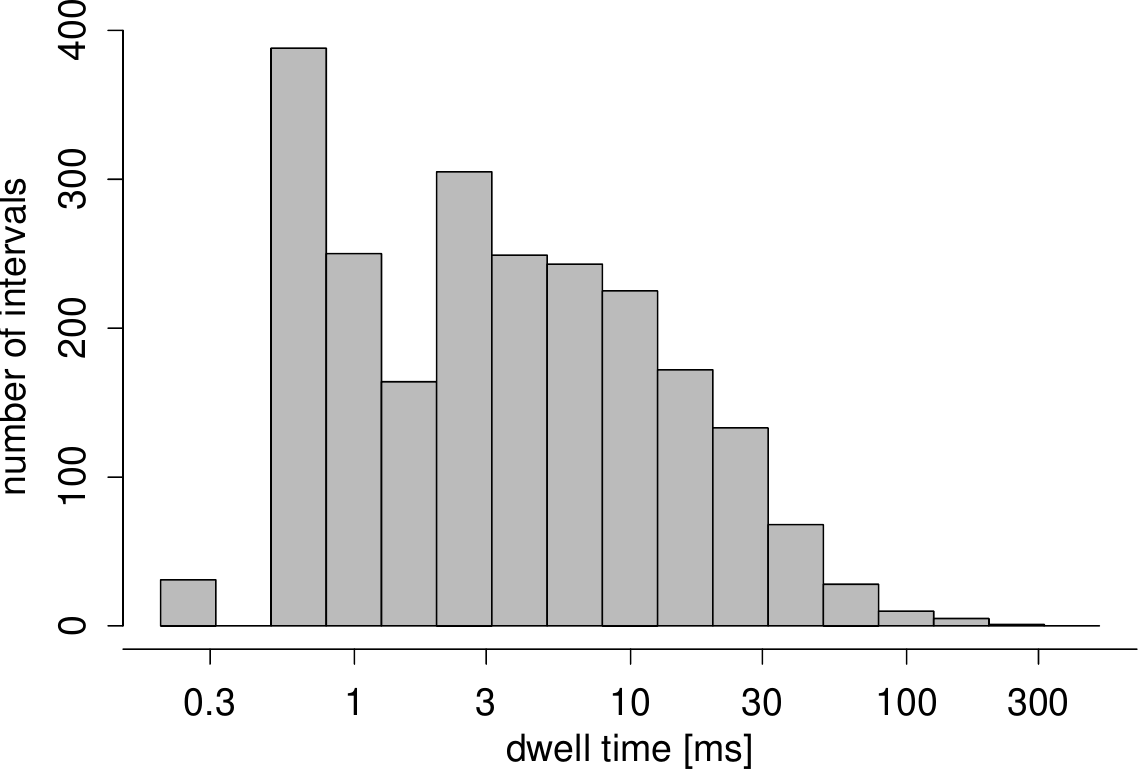}}
  \caption{Dwell times of state 0. The histograms display dwell times extracted from the Viterbi path. Note that the horizontal axis is logarithmic. \label{fig:dwell_times_alt}}
\end{figure}

In order to get more meaningful results, we consider a Gaussian kernel density estimate of the log-transformed dwell times. The kernel bandwidth is determined via Silverman's rule of thumb as cross validation suffers from the same binning artifacts as histograms. The kernel density estimates are displayed in in Figure~\ref{fig:dwell_times_kde}. One can see that some of the curves for states $0$ and $2$ are flat near the maxima regions so one may consider the dwell times read off from them less reliably. Also, almost all curves have pronounced ``shoulders'' which indicate deviations from an exponential distribution, especially for short times. These can be partly explained by binning artifacts but they can also by interpreted as indicating deviations from the Markov assumption at small time scales.

\begin{figure}[h!]
  \centering
  \subcaptionbox*{data set $1$, state $0$}[0.31\textwidth]{\includegraphics[width=0.3\textwidth]{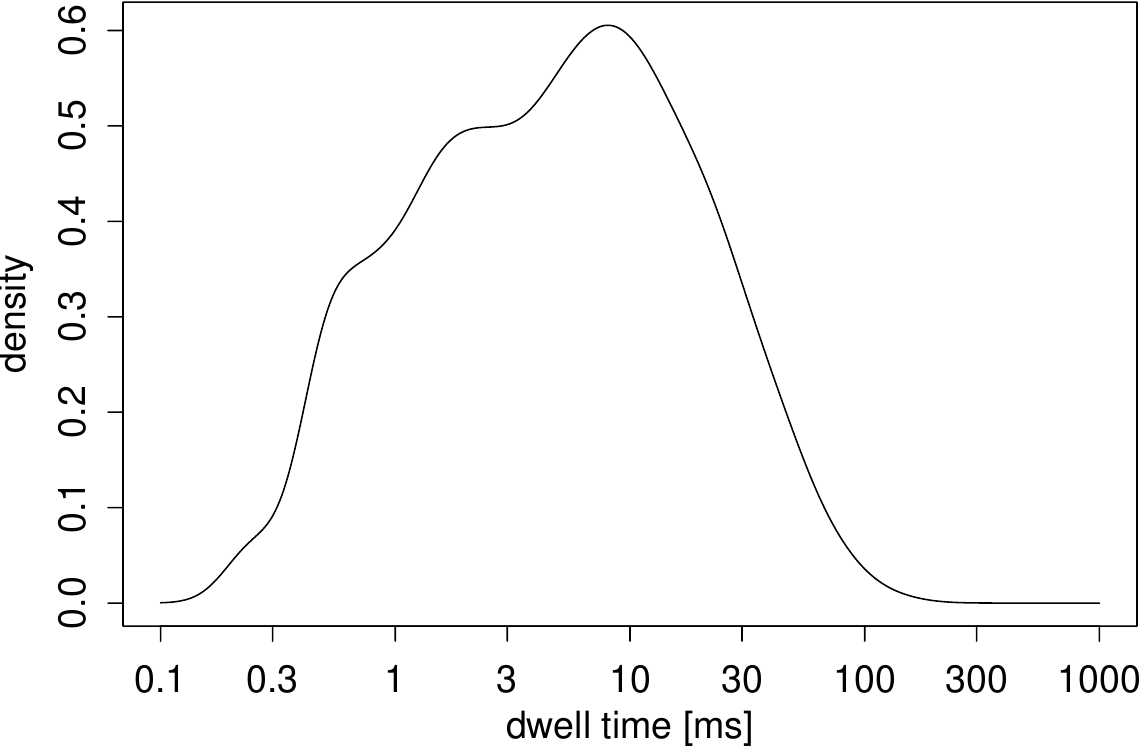}}
  \subcaptionbox*{data set $1$, state $1$}[0.31\textwidth]{\includegraphics[width=0.3\textwidth]{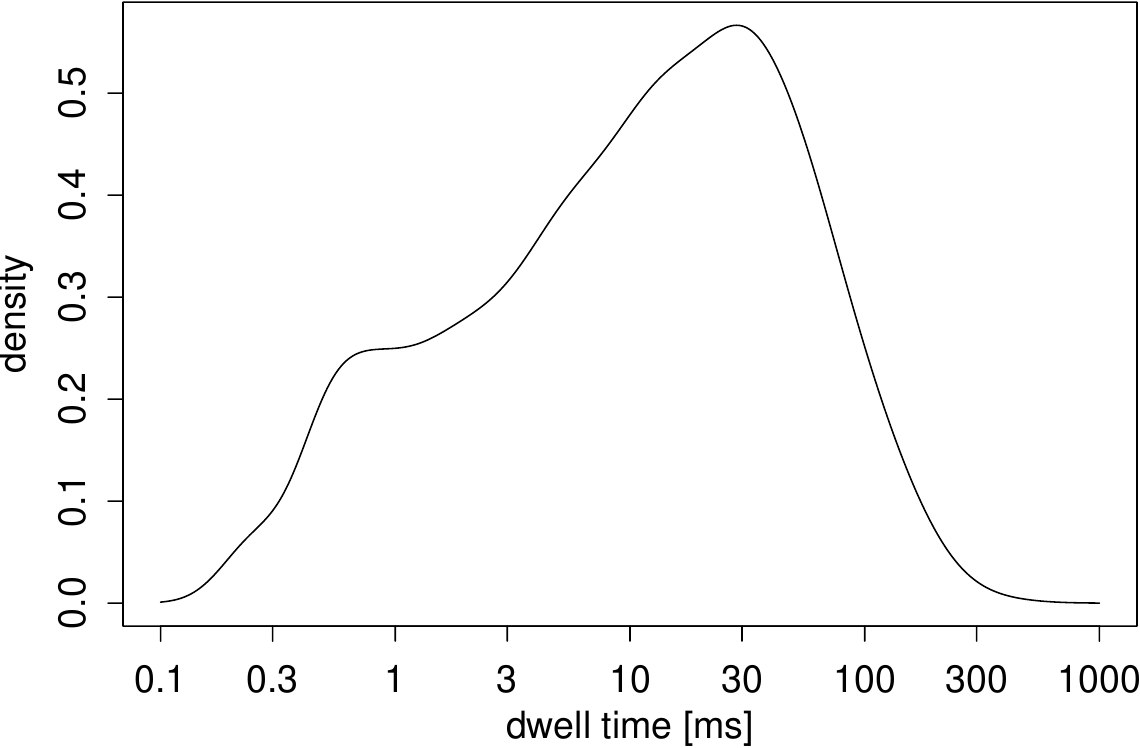}}
  \subcaptionbox*{data set $1$, state $2$}[0.31\textwidth]{\includegraphics[width=0.3\textwidth]{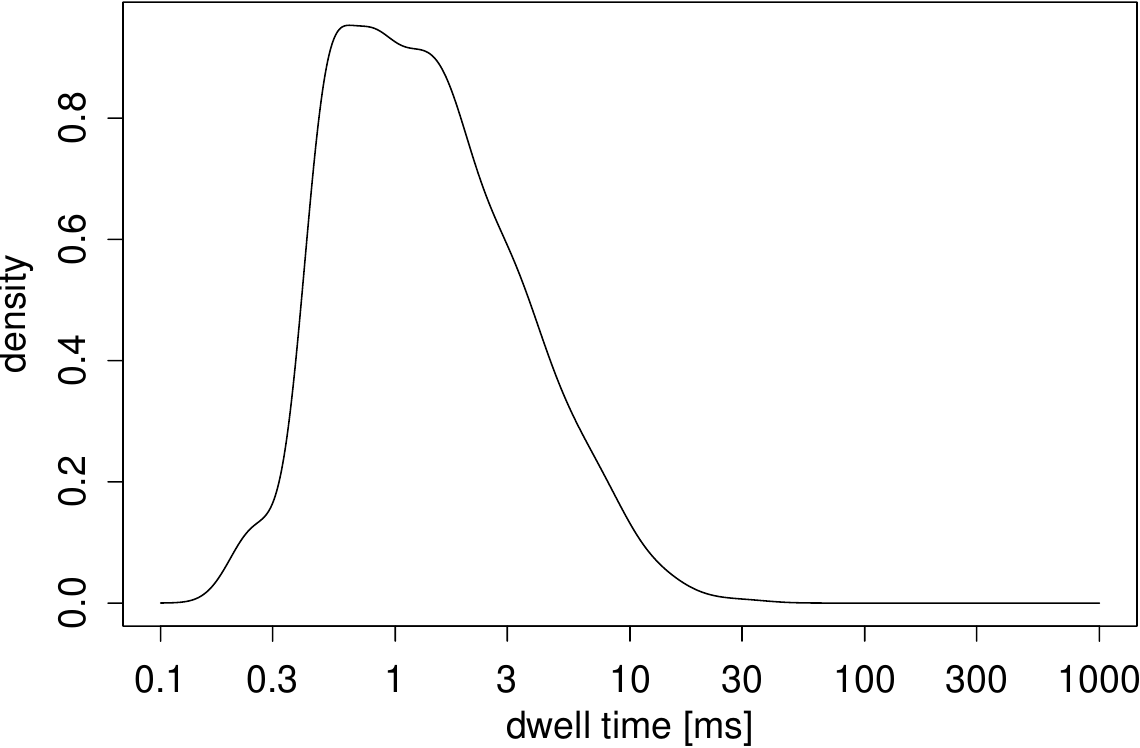}}\\
  \subcaptionbox*{data set $2$, state $0$}[0.31\textwidth]{\includegraphics[width=0.3\textwidth]{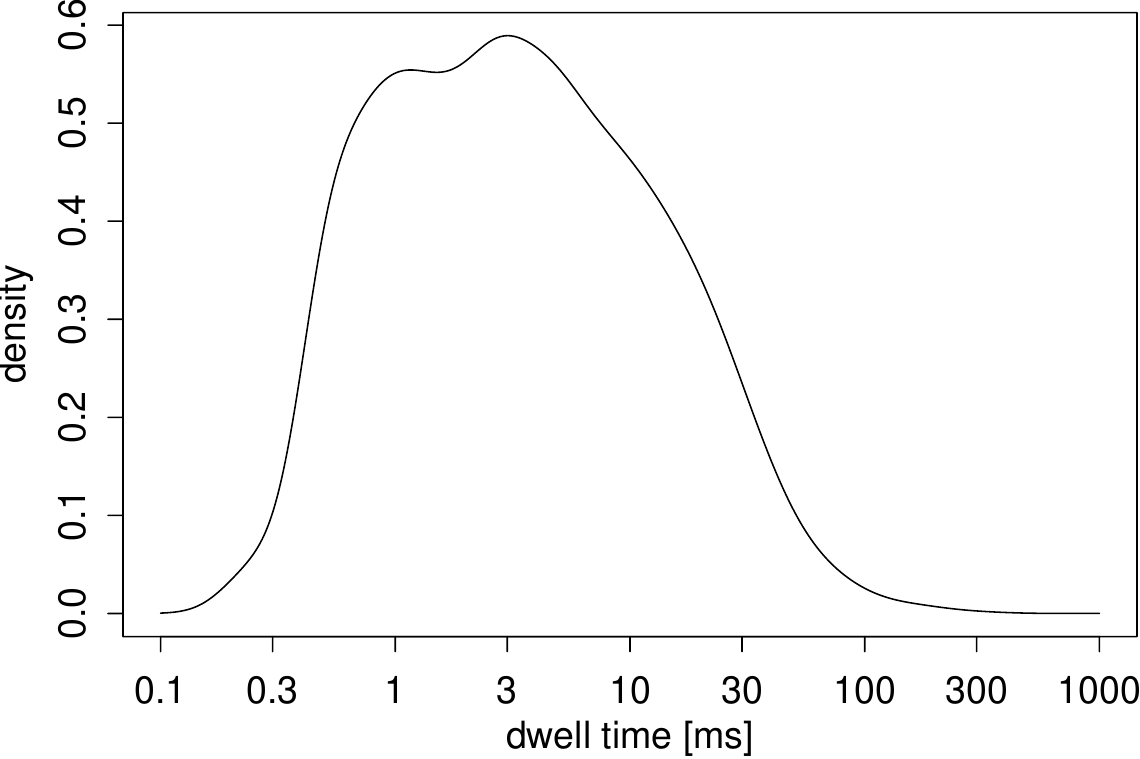}}
  \subcaptionbox*{data set $2$, state $1$}[0.31\textwidth]{\includegraphics[width=0.3\textwidth]{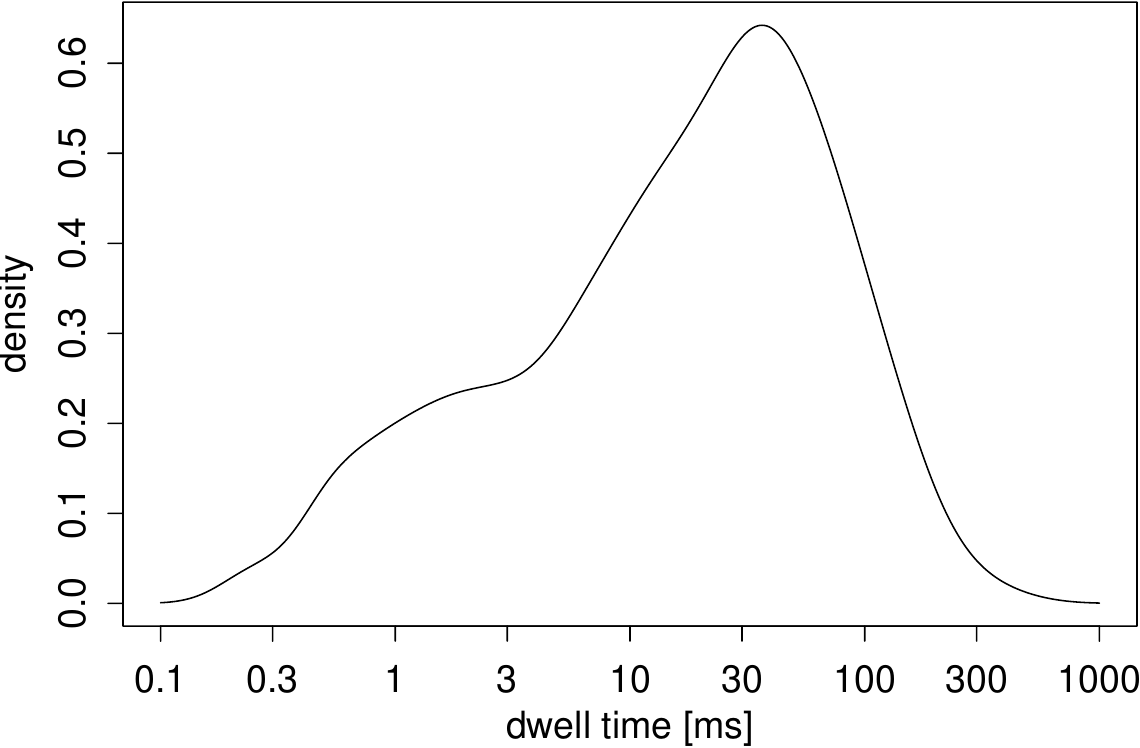}}
  \subcaptionbox*{data set $2$, state $2$}[0.31\textwidth]{\includegraphics[width=0.3\textwidth]{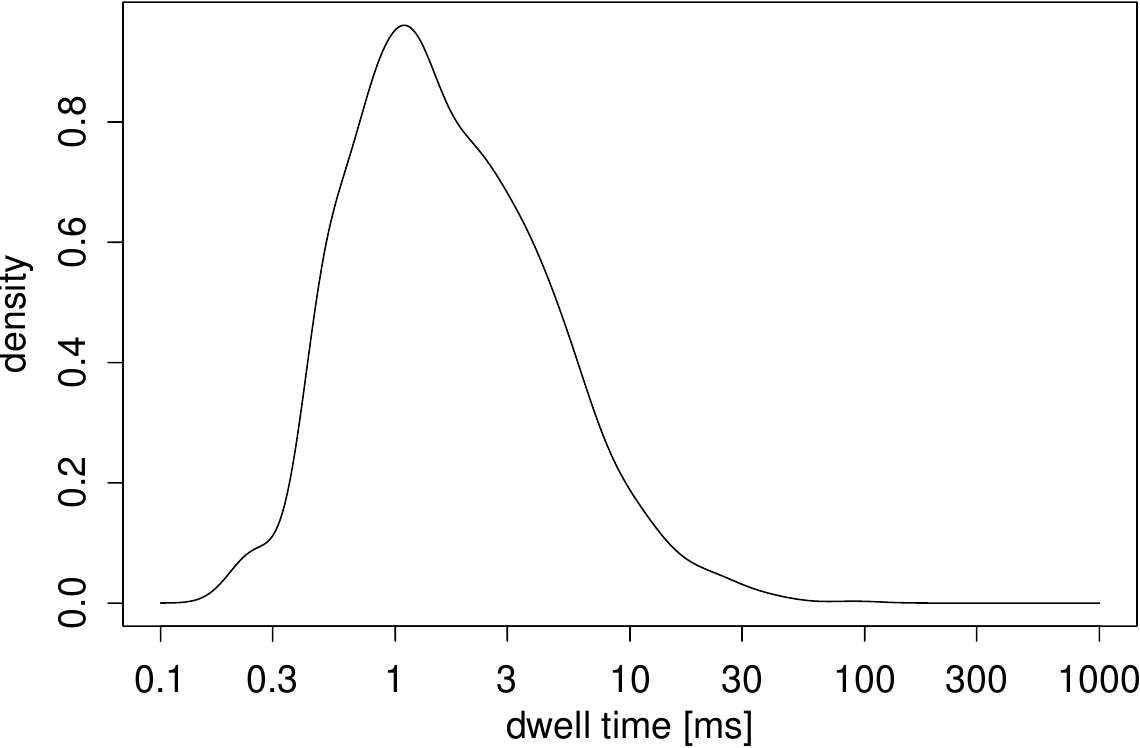}}
  \caption{Kernel density estimates of dwell times of the lower three states for both data sets. State 3 has to few events to yield a meaningful result. The densities display dwell times extracted from the Viterbi path. Note that the horizontal axis is logarithmic. \label{fig:dwell_times_kde}}
\end{figure}

\FloatBarrier
\newpage

\section{Illustration of potential individual channel traces} \label{suppl_sec: Intro}

Figure~\ref{fig:plausible_single_traces} below contains a decomposition of measured data into possible marginal traces of two individual channels, displayed in red and blue.

\begin{figure}[h!!]
  \centering
  \subcaptionbox{data set 1}[0.95\textwidth]{\includegraphics[width=0.8\textwidth]{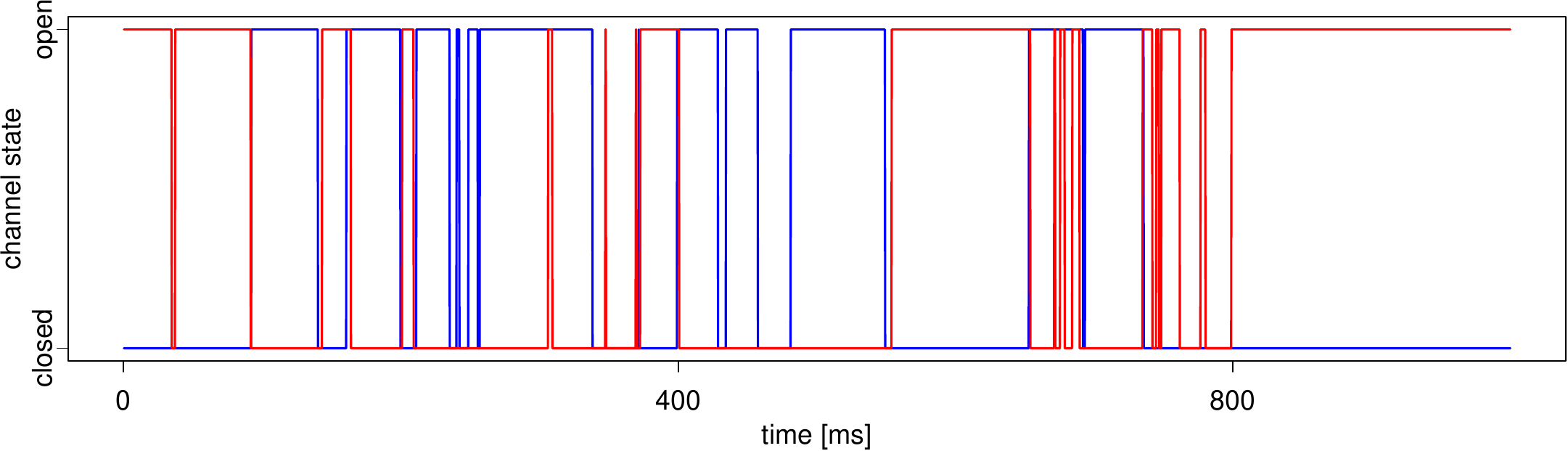}}\\
  \subcaptionbox{CK model with $\widehat{\kappa}=0$}[0.95\textwidth]{\includegraphics[width=0.8\textwidth]{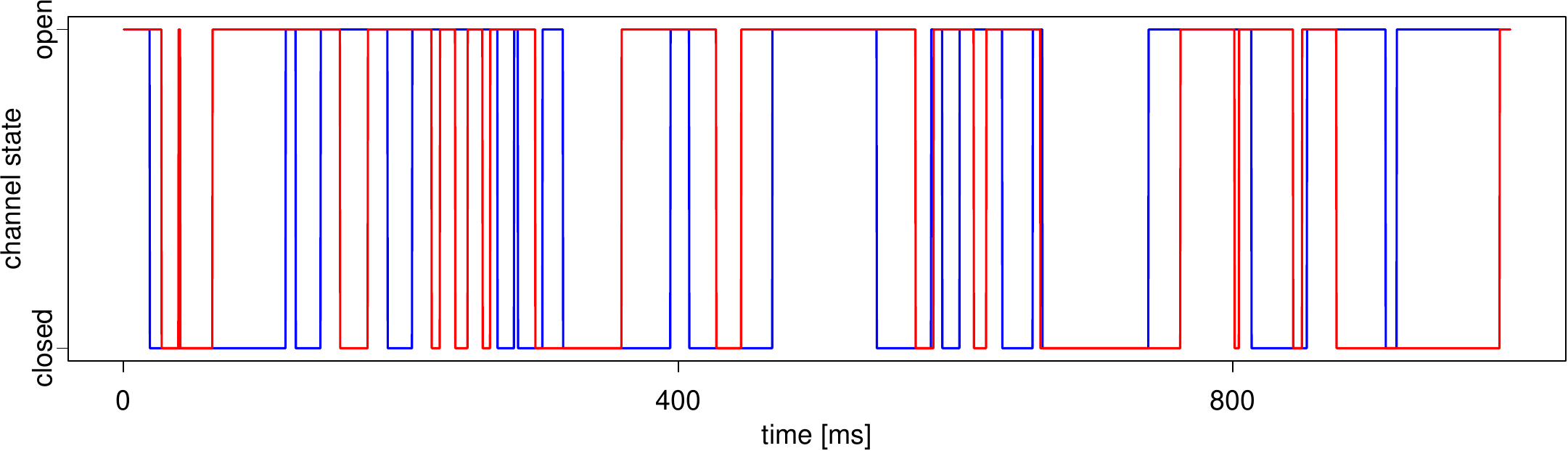}}\\
  \subcaptionbox{VND model}[0.95\textwidth]{\includegraphics[width=0.8\textwidth]{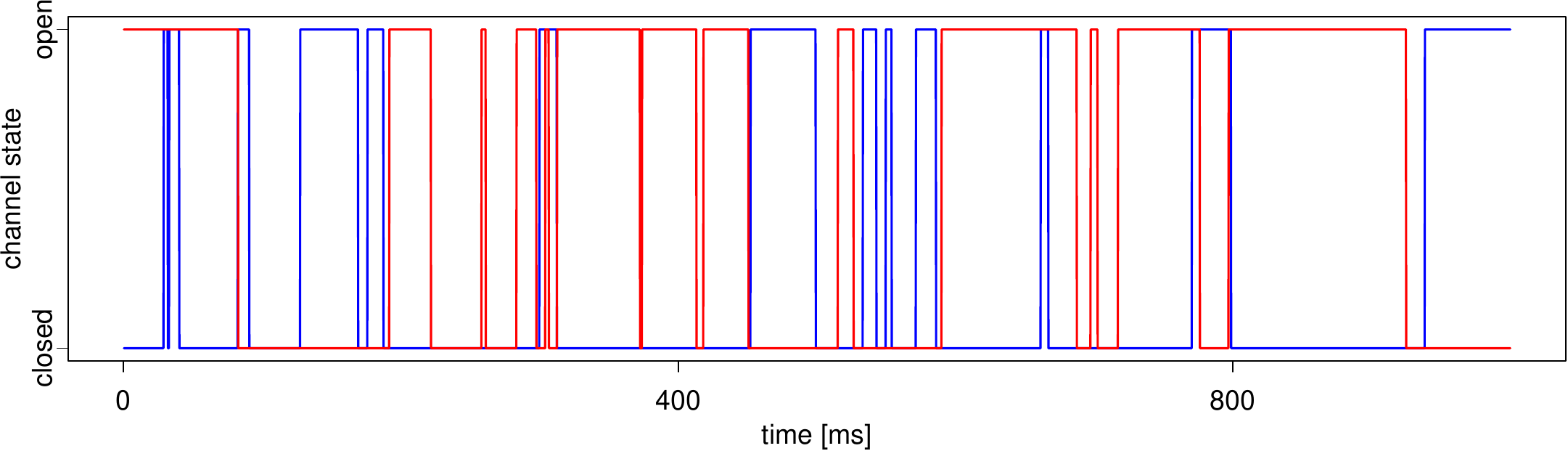}}
  \caption{
    Decomposition of
    measured data into possible marginal traces of two individual channels, displayed in red and blue. 
    \label{fig:plausible_single_traces}}
\end{figure}

For simplicity, events are excluded in which one channel opens and the other closes at exactly the same time. Instead, for every time interval, during which the measured sum process remains in state 1, indicating one open channel, one of the two channels is chosen to be open at random. The resulting possible traces are compared to marginal traces simulated from CK and VND models using estimated model parameters. While the CK traces are uncorrelated, the VND traces as well as the possible traces for real data show very strong negative correlation. Therefore, the possible data traces resemble much more closely the VND traces than the CK traces.

\FloatBarrier
\newpage

\section{Additional material to the ion channels application}

\subsection{Planar lipid bilayer channel measurements}
\label{suppl_sec:exp_meth}

ER vesicles from HEK293 cells expressing RyR2-WT were prepared as described previously \citep{meli2011}. 
Planar lipid bilayers were formed by painting a mixture of phosphatidylethanolamine and phosphatidylcholine (3:1 ratio; Avanti Polar Lipids) across 200-$\upmu$m aperture in polysulfonate cup (Warner Instruments) separating 2 chambers. The trans chamber (1.0 ml), representing the intra-SR (luminal) compartment, was connected to the head stage input of a bilayer voltage clamp amplifier. The cis chamber (1.0 ml), representing the cytoplasmic compartment, was held at virtual ground. Used basic solutions were as follows: 1 mM EGTA, 250/125 mM Hepes/Tris, 50 mM KCl, 0.64 mM CaCl2, pH 7.35 as cis solution (150 nM free [Ca${}^{2+}$]) and 53 mM Ba(OH)2, 50 mM KCl, 250 mM Hepes, pH 7.35 as trans solution. 
RyR2-WT channels were reconstituted by spontaneously fusing ER vesicles into the planar lipid bilayer. Currents through incorporated RyR2-WT channels were recorded at 0 mV using an amplifier (BC-525D, Warner Instruments), filtered at 1 kHz (LPF-8, Warner Instruments), and digitized at 4 kHz. Data acquisition was performed using Digidata 1440A and Axoscope 10 software (Axon Instruments). 
Activity / gating of RyR2-WT channels was first recorded at 150 nM cytosolic [Ca${}^{2+}$] and either 5 or 10 mM luminal [Ca${}^{2+}$]. After that, the cytosolic [Ca${}^{2+}$] was increased up to 5 $\upmu$M together with addition of 1mM Na-ATP to fully activate and to assess the number of channels in bilayer. At the end of the experiment, 8-16 $\upmu$M ryanodine was applied to confirm RyR2 channels identity.

{color{blue}
  \subsection{Simulations for Cooperative Gating}
  \label{suppl_sec:coop_gating}
  
  Here, we present simulations showing the robustness of finding cooperative gating even if channel number is misidentified. The setup is similar to Subsection~\ref{subsec:robust}, simulating data from a cooperative gating model instead of a competitive gating model. We simulate $1\,000\,000$ data points from systems with different numbers of channels $\ell$ in the VND model with different sets of parameters which were chosen such that the highest number of channels open at the same time was $3$. In the language of Definition~\ref{def:coop-comp-gen} all models used for simulation are $3$-cooperative. Then we fit a VND model with $3$ channels to the data and inspect the estimated parameters for signs of competitive or cooperative gating (c.f. Definition~\ref{def:coop-comp}). In order to acquire variance estimates, $100$ repetitions were done for each set of parameters. We investigate the ratios $\frac{1-\widehat\eta_1}{1-\widehat\eta_3}$, $\frac{1-\widehat\eta_2}{1-\widehat\eta_3}$, $\frac{1-\widehat\lambda_1}{1-\widehat\lambda_0}$, and $\frac{1-\widehat\lambda_1}{1-\widehat\lambda_0}$, where we expect all of these ratios to be 
  larger than $1$ in case of cooperative gating since this would indicate that transitions into the state with none or three channels open are preferred relative to transitions out of these states. Again, we use the shorthand notation $\overline{\lambda} := 1- \lambda$ and $\overline{\eta} := 1- \eta$.
  
  \begin{figure}[h!]
    \centering
    \subcaptionbox{}[0.45\textwidth]{\includegraphics[width=0.45\textwidth]{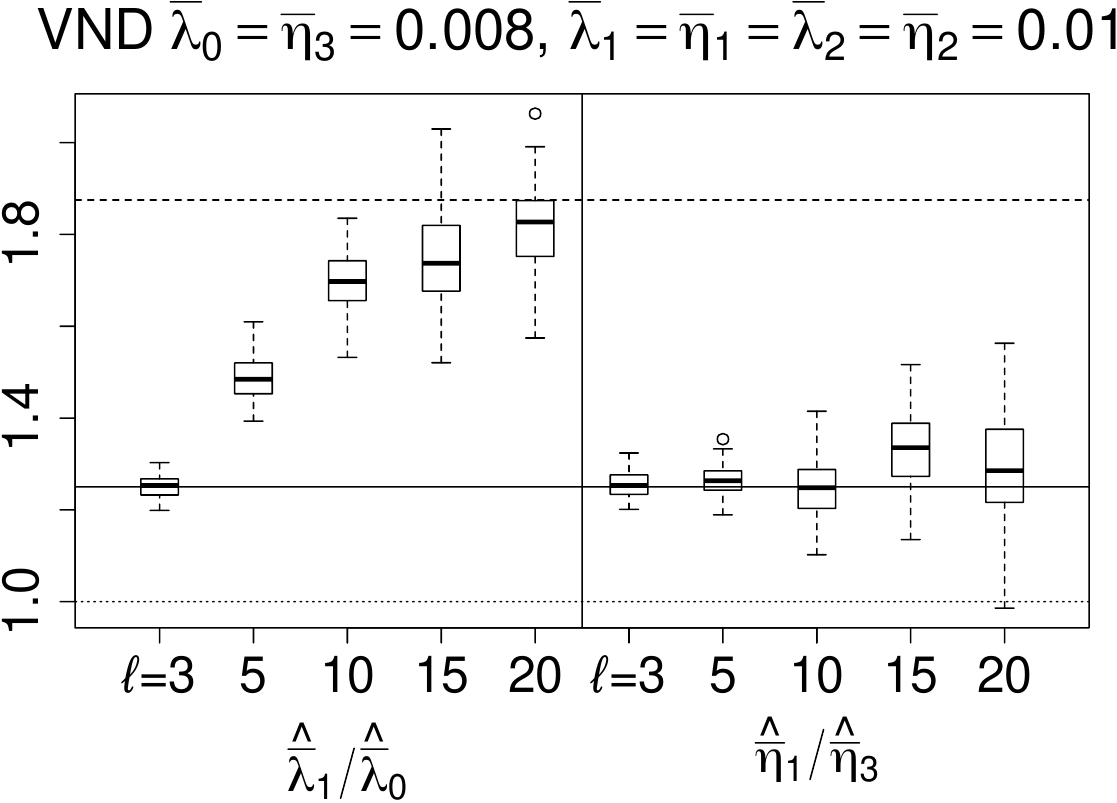}}
    \hspace*{0.02\textwidth}
    \subcaptionbox{}[0.45\textwidth]{\includegraphics[width=0.45\textwidth]{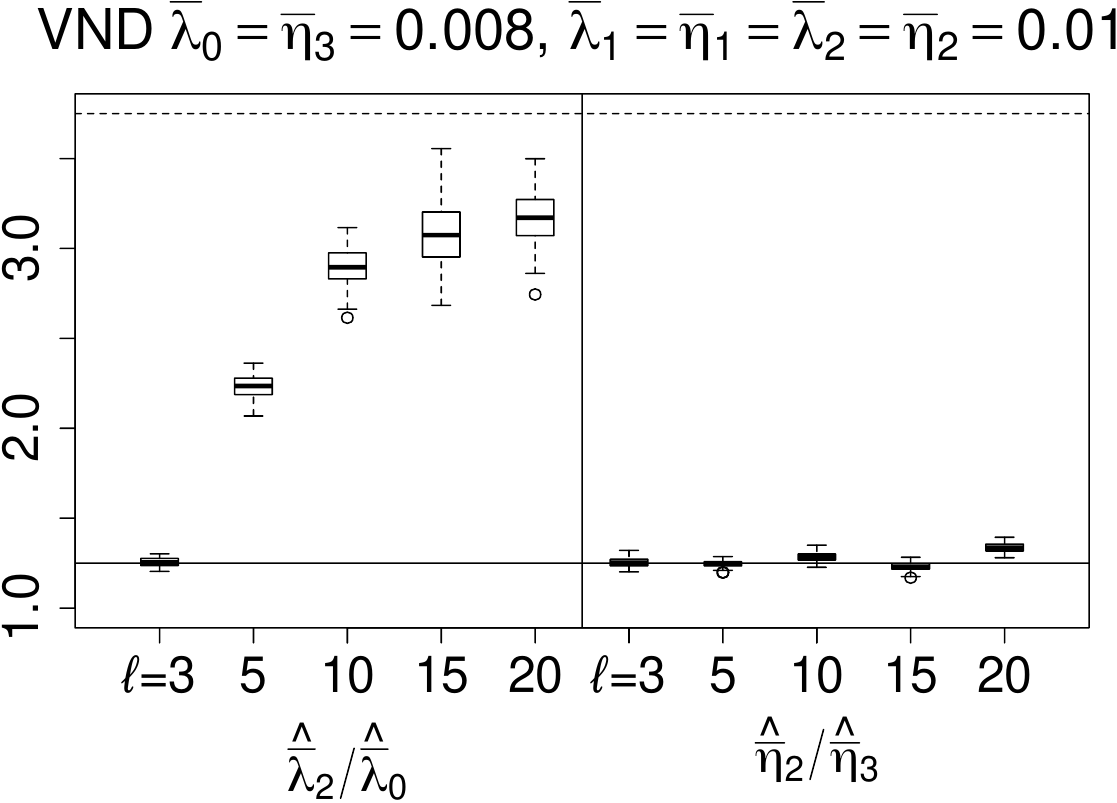}}\\
    \subcaptionbox{}[0.45\textwidth]{\includegraphics[width=0.45\textwidth]{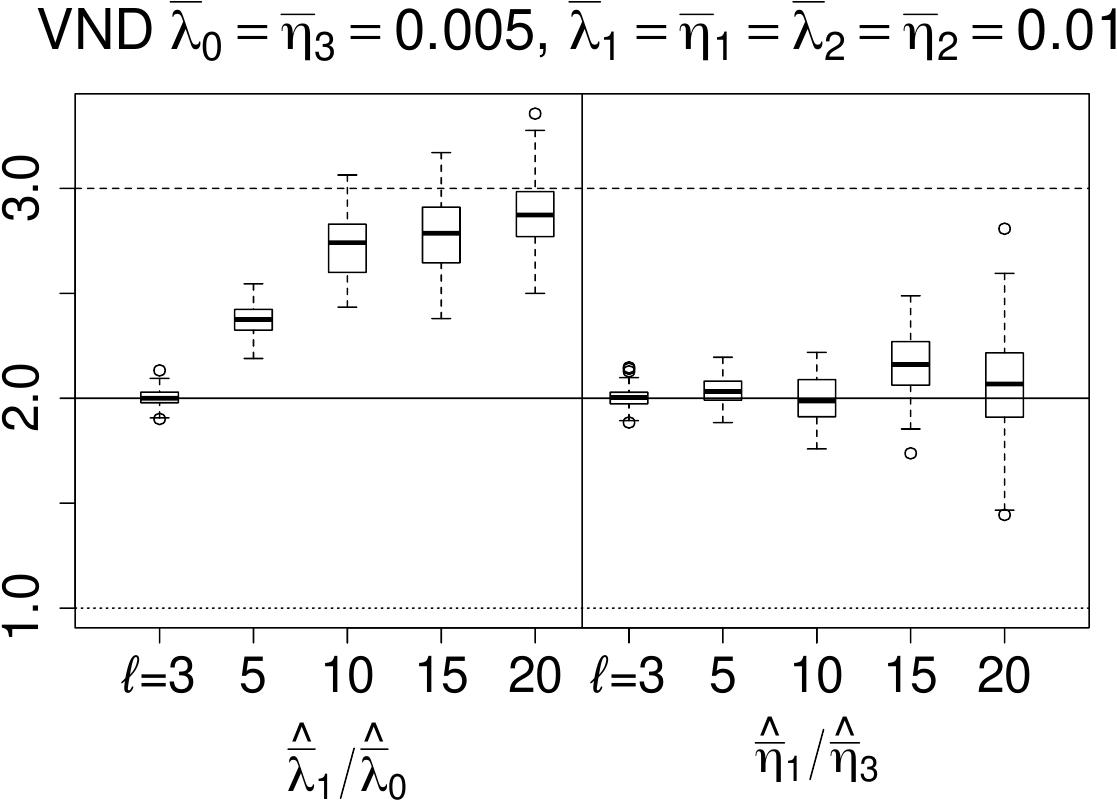}}
    \hspace*{0.02\textwidth}
    \subcaptionbox{}[0.45\textwidth]{\includegraphics[width=0.45\textwidth]{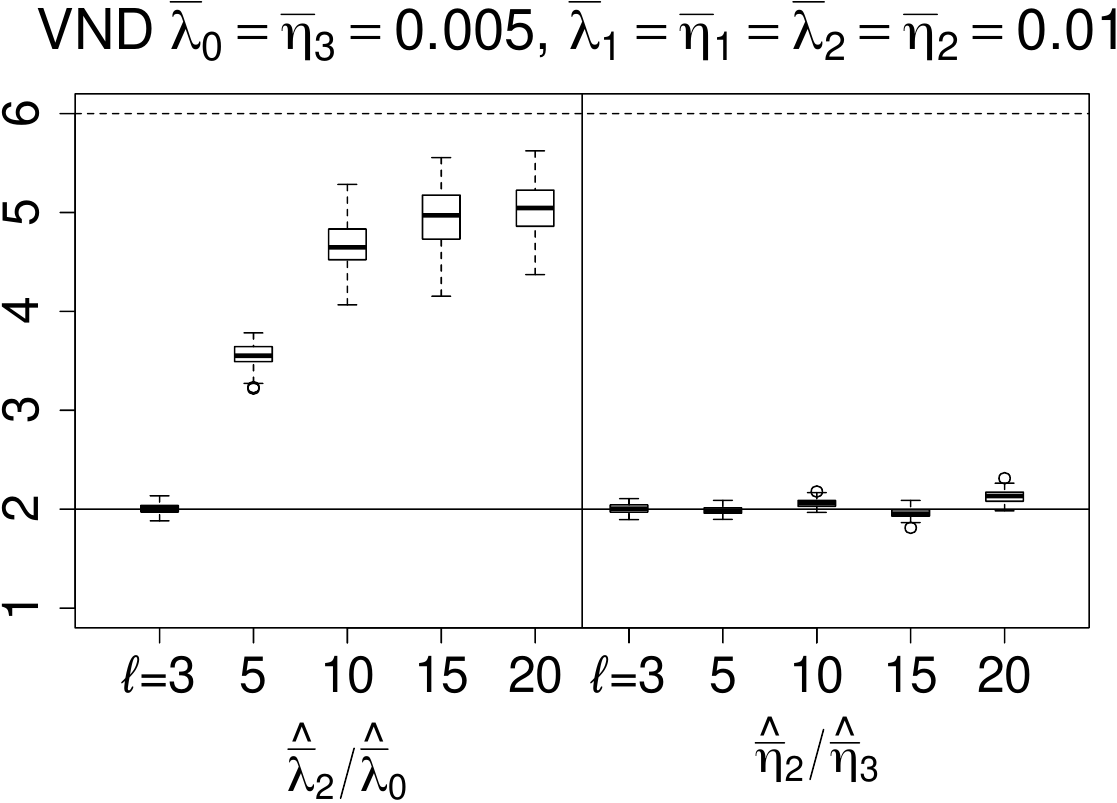}}
    \caption{Boxplots displaying ratios of parameters estimated by a 3 channel VND model for $100$ repetitions of simulations of $1\,000\,000$ data points from a system with $\ell = 3, 5, 10, 15, 20$ channels with the given parameters and $\lambda_k = 1$ for $k > 1$ and $\eta_k = 0.8$ for $k > 2$. For reference, the true values of the ratios are indicated by a solid line and $\overline{\lambda}$-ratios adjusted for channel number based on the considerations explained in Subsection~\ref{suppl_sec:robustness} are indicated by a dashed line. \label{fig:boxplots-coop}}
  \end{figure}
  
  As is clear from the estimated parameters displayed in Figure \ref{fig:boxplots-coop}, the qualitative coupling behavior is recovered. Since all parameter ratios are clearly positive, none of the systems would be wrongly interpreted as competitively coupled.
}
\FloatBarrier

\subsection{Estimating the number of channels}
\label{suppl_sec: est_num_chan}

Here, we investigate the model selection criteria we applied in Section 4.2 of the article with simulated data. We simulated data sets of size $n=600\,000$ from the VND model with $\ell = 5,10,15,20$ channels. The parameters were chosen in such a way that the resultant traces closely resemble data set 1. In particular, no more than three channels are open at the same time in any data set. The parameter value were chosen as the results of the fit of the $\ell$ channel VND models to data set 1. This means that the transition matrices are deliberately similar to each other and the traces all look similar to data set 1. The numerical values of the parameters are listed in Table \ref{tab:simulation_parameters}.

\begin{table}
  \caption{Parameters used for the VND models in simulations. The other parameters were fixed as $\lambda_3 = \dots = \lambda_{\ell-1} = 1$ and $\eta_4 = \dots = \eta_{\ell} = 0$. The parameters were chosen as the results of parameters fits of the VND models with the same number of channels to data set 1. As a results, the transition matrices corresponding to these parameters are very similar.}
  \begin{center}
    \begin{tabular}{l|c|c|c|c|c|c}
      & $\lambda_0$ & $\lambda_1$ & $\lambda_2$ & $\eta_1$ & $\eta_2$ & $\eta_3$ \\
      \hline
      $\ell = 5$ & 0.99507 & 0.99947 & 0.99974 & 0.99221 & 0.93665 & 0.94441 \\
      $\ell = 10$ & 0.99753 & 0.99977 & 0.99990 & 0.99221 & 0.93665 & 0.94388 \\
      $\ell = 15$ & 0.99836 & 0.99985 & 0.99994 & 0.99221 & 0.93665 & 0.94390 \\
      $\ell = 20$ & 0.99877 & 0.99989 & 0.99996 & 0.99221 & 0.93666 & 0.94353 
    \end{tabular}
  \end{center}
  \label{tab:simulation_parameters}
\end{table}

\begin{figure}[h!]
  \centering
  \subcaptionbox{$\ell=5$ simulation}[0.45\textwidth]{\includegraphics[width=0.45\textwidth]{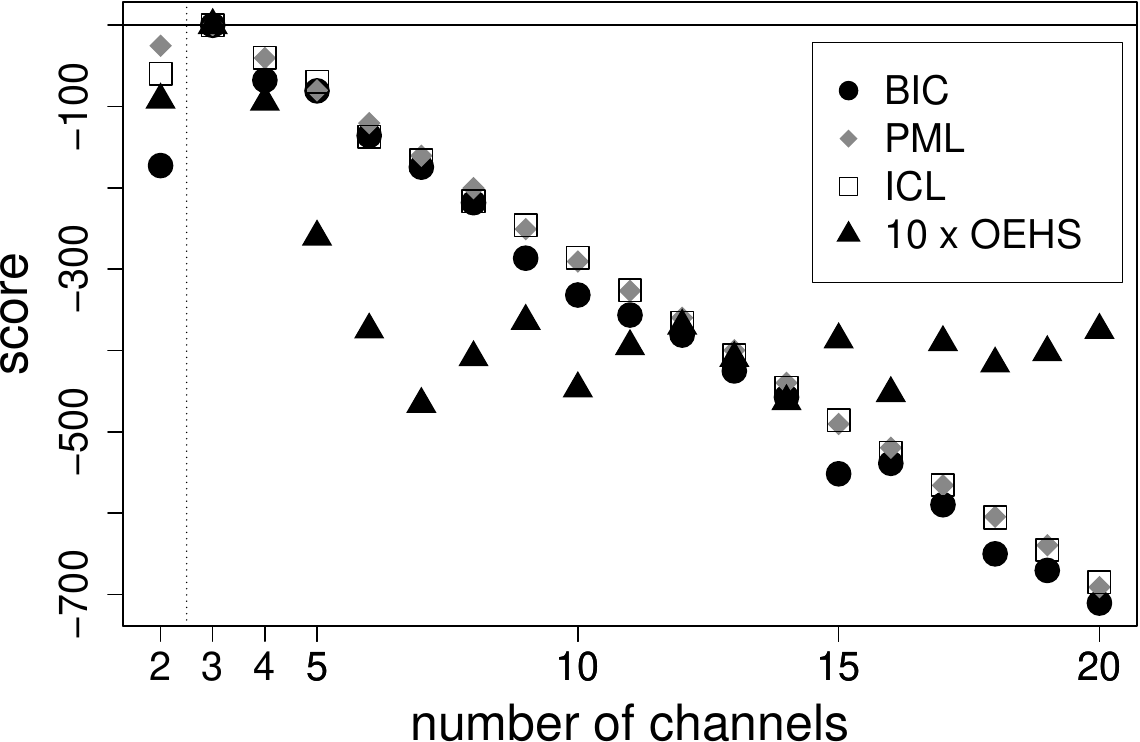}}
  \hspace*{0.02\textwidth}
  \subcaptionbox{$\ell=10$ simulation}[0.45\textwidth]{\includegraphics[width=0.45\textwidth]{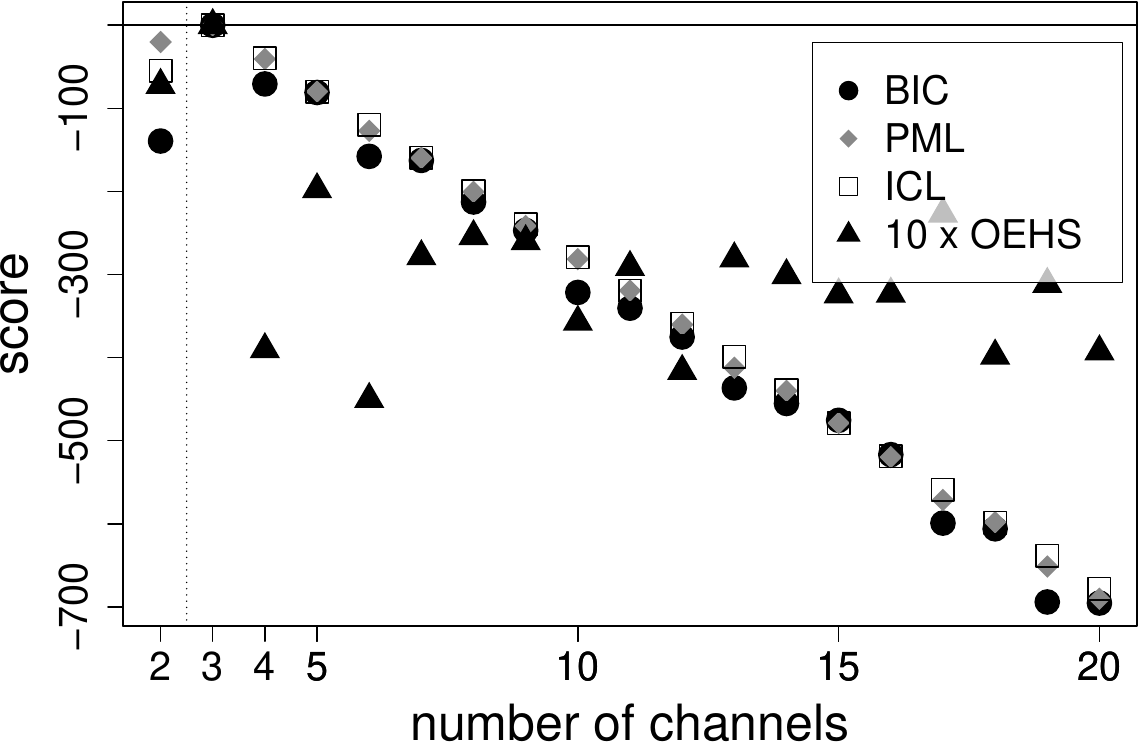}}\vspace*{\baselineskip}\\
  \subcaptionbox{$\ell=15$ simulation}[0.45\textwidth]{\includegraphics[width=0.45\textwidth]{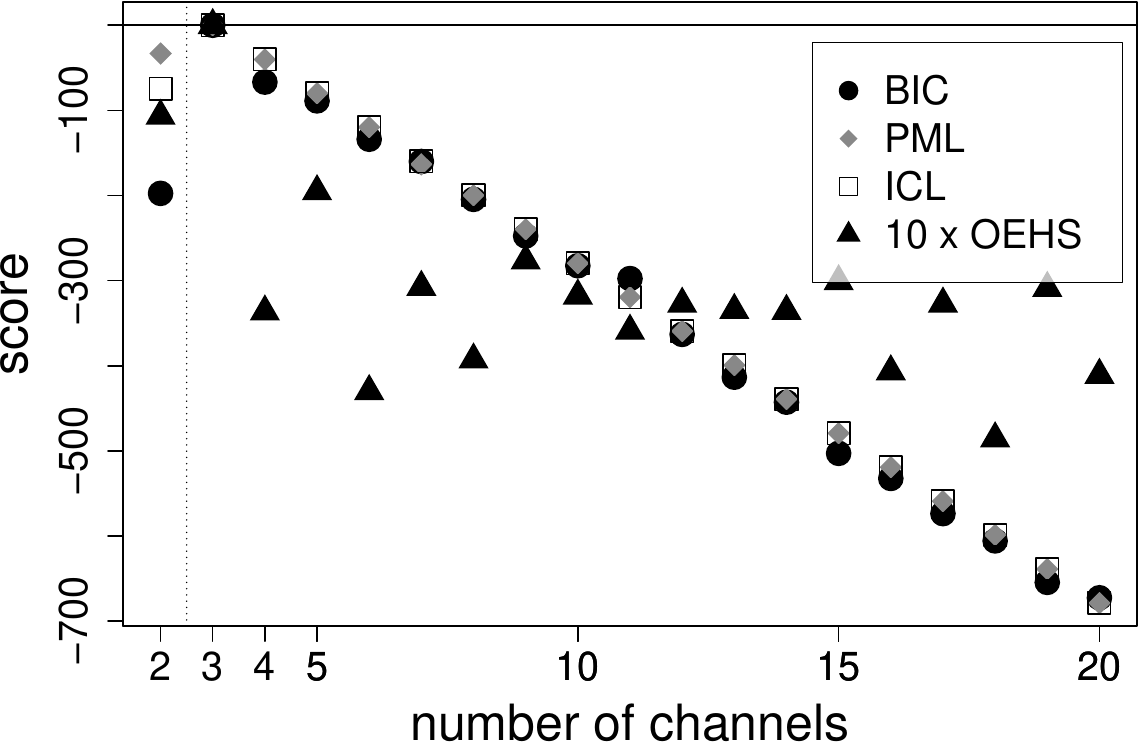}}
  \hspace*{0.02\textwidth}
  \subcaptionbox{$\ell=20$ simulation}[0.45\textwidth]{\includegraphics[width=0.45\textwidth]{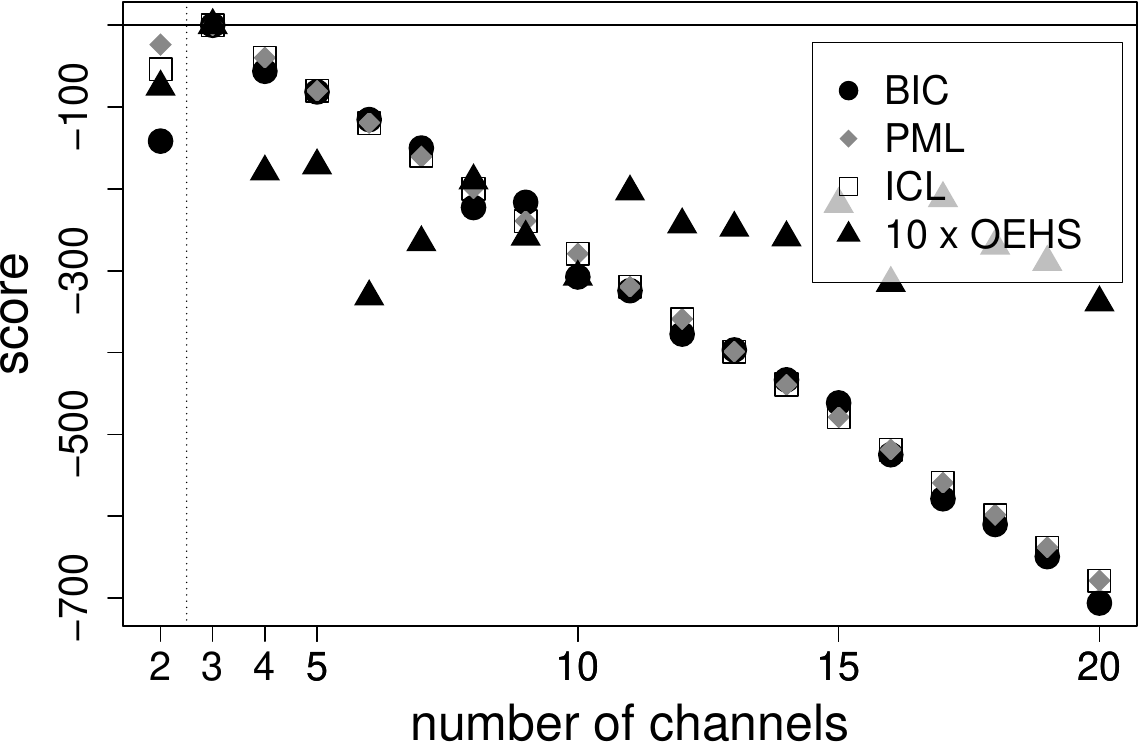}}
  \caption{Several model selection criteria for the VND model with $\ell \in \{2, \dots, 20\}$ applied to the simulated data sets. The values for $\ell = 2$ were divided by $10$ so they do not dominate the vertical axis scaling. Values for the cross validation for $\ell \ge 3$ have been inflated by a factor of $10$ since they are much more similar than the other criteria. For all data sets, one can conclude that the model with $\ell=3$ provides the best fit. \label{fig:n_channels}}
\end{figure}

The results of the model selection criteria are displayed in Figure \ref{fig:n_channels}. As one can clearly see, all model selection criteria prefer the model with $\ell = 3$, contrary to the true number of $\ell$ chosen in the simulations. This can be explained by the strongly competitive gating suggested by the estimated parameters. In fact, the left upper corner part $U^{\textnormal{\rm(VND,$\ell$)}}$ of the transition matrix $Q^{(\text{VND})}$ of the hidden Markov chain within the VND model for $\ell$ channels, which can be estimated from the data at hand, takes the form
\begin{align*}
U^{\textnormal{\rm(VND,$\ell$)}} :=& \begin{pmatrix}
\lambda_0^\ell & \ell \lambda_0^{\ell-1}(1-\lambda_0) & \frac{\ell (\ell-1)}{2}\lambda_0^{\ell-2}(1-\lambda_0)^2 \\
\lambda_1^{\ell-1}(1-\eta_1) & q_{1,1} & q_{1,2} \\
\lambda_1^{\ell-2}(1-\eta_2)^2 & q_{2,1} & q_{2,2}
\end{pmatrix}
\end{align*}
with
\begin{align*}
q_{1,1} :=& \lambda_1^{\ell-1} \eta_1 + (\ell-1)\lambda_1^{\ell-2}(1-\lambda_1)(1-\eta_1)\\
q_{1,2} :=& (\ell-1)\lambda_1^{\ell-2}(1-\lambda_1) \eta_1 + \frac{(\ell-1)(\ell-2)}{2}\lambda_1^{\ell-3}(1-\lambda_1)(1-\eta_1)^2\\
q_{2,1} :=& 2 \lambda_2^{\ell-2} \eta_2 (1-\eta_2) + (\ell-2)\lambda_2^{\ell-3}(1-\lambda_2)(1-\eta_2)^2\\
q_{2,2} :=& \lambda_2^{\ell-2}\eta_2^2 + 2(\ell-2)\lambda_2^{\ell-3}(1-\lambda_2) \eta_2 (1-\eta_2)\\
&+ \frac{(\ell-2)(\ell-3)}{2}\lambda_2^{\ell-4}(1-\lambda_2)^2(1-\eta_2)^2.
\end{align*}
This shows that determining $\ell$ from the data is very difficult since the entries of $Q^{\textnormal{\rm(VND,$\ell$)}}$ depend only very weakly on $\ell$ (in analogy of estimating the number of trials within a binomial distribution with small success probability). In the case that only part of the state space is realized in a data set, the model selection hinges on the matrix entries $q_{0,2}$ and $q_{2,0}$, which differ most clearly between estimated matrices for different $\ell$, both increasing with $\ell$. Both entries are very small and therefore very sensitive to noise. Therefore, the results remain inconclusive apart from the finding that the parameter $\ell$ is very difficult to identify from the present data sets and even simulated data with similar features.

\subsection{Robustness to channel number}
\label{suppl_sec:robustness}

To understand why the ratio $\frac{1-\widehat\lambda_0}{1-\widehat\lambda_1}$ decreases with increasing $\ell$, we note the for small values of $\overline{\lambda}_j$ and $\overline{\eta}_j$, the upper left square of the transition matrix can be approximated by the linearization
\begin{align*}
U^{\textnormal{\rm(VND,$\ell$)}} =& \begin{pmatrix}
1 - \ell\overline{\lambda}_0 & \ell\overline{\lambda}_0 & 0 \\
\overline{\eta}_1 & 1 - (\ell-1)\overline{\lambda}_1 - \overline{\eta}_1 & (\ell-1)\overline{\lambda}_1 \\
0 & 2\overline{\eta}_2 & 1 - 2\overline{\eta}_2
\end{pmatrix} + \mathcal{O}(\overline{\lambda}_i^2, \overline{\eta}_j^2, \overline{\lambda}_i \overline{\eta}_j) \, .
\end{align*}
Now, one can set $\widehat{U}^{\textnormal{\rm(VND,$3$)}} = U^{\textnormal{\rm(VND,$\ell$)}} + \mathcal{O}(\overline{\lambda}_i^2, \overline{\eta}_j^2, \overline{\lambda}_i \overline{\eta}_j)$, which leads to the relations
\begin{align*}
\frac{1}{3 \widehat{\overline{\lambda}}_0} - 1 &= \frac{1}{\ell \overline{\lambda}_0} - 1 + \mathcal{O}(\overline{\lambda}_i^2, \overline{\eta}_j^2, \overline{\lambda}_i \overline{\eta}_j) & \Rightarrow \quad \widehat{\overline{\lambda}}_0 &= \frac{\ell}{3} \overline{\lambda}_0 + \mathcal{O}(\overline{\lambda}_i^2, \overline{\eta}_j^2, \overline{\lambda}_i \overline{\eta}_j)\\
\frac{1}{2 \widehat{\overline{\lambda}}_1} - 1 &= \frac{1}{(\ell-1) \overline{\lambda}_1} - 1 + \mathcal{O}(\overline{\lambda}_i^2, \overline{\eta}_j^2, \overline{\lambda}_i \overline{\eta}_j) & \Rightarrow \quad \widehat{\overline{\lambda}}_1 &= \frac{\ell-1}{2} \overline{\lambda}_1 + \mathcal{O}(\overline{\lambda}_i^2, \overline{\eta}_j^2, \overline{\lambda}_i \overline{\eta}_j)\\
\frac{1-\widehat\lambda_0}{1-\widehat\lambda_1} &= \frac{2\ell}{3(\ell-1)} \frac{1-\lambda_0}{1-\lambda_1} + \mathcal{O}(\overline{\lambda}_i^2, \overline{\eta}_j^2, \overline{\lambda}_i \overline{\eta}_j) \, ,
\end{align*}
which indicates that, in the limit of all $\overline{\eta}_j$, $\overline{\lambda}_i$ going to $0$ at the same rate, the ratio $\frac{1-\lambda_0}{1-\lambda_1}$ will be underestimated by a factor $\frac{2\ell}{3(\ell-1)}$, which approaches $\frac{2}{3}$ for large $\ell$.

In conclusion, underestimating the number of channels may lead one to wrongly miss but not to wrongly detect competitive gating. This underscores the usefulness of the VND model for the identification of coupled gating, even if the number of channels is not exactly known and only a subset of the channels are open at the same time during the observation time span. More specifically, it reinforces the qualitative finding of competitive gating for the ion channels analyzed above as this does not arise as an artifact of underestimating the number of channels.

\bibliographystyle{Chicago} 
\bibliography{LDCpaper}       

\begin{thebibliography}{}

\bibitem[\protect\citeauthoryear{Ball, Milne, Tame, and Yeo}{Ball
  et~al.}{1997}]{ball1997}
Ball, F., R.~K. Milne, I.~D. Tame, and G.~F. Yeo (1997).
\newblock Superposition of interacting aggregated continuous-time {{Markov}}
  chains.
\newblock {\em Adv. Appl. Probab.\/}~{\em 29\/}(1), 56--91.

\bibitem[\protect\citeauthoryear{Ball and Rice}{Ball and Rice}{1992}]{ball1992}
Ball, F.~G. and J.~A. Rice (1992).
\newblock Stochastic models for ion channels: {{Introduction}} and
  bibliography.
\newblock {\em Math. Biosci.\/}~{\em 112\/}(2), 189--206.

\bibitem[\protect\citeauthoryear{Bartsch, Llabr{\'{e}}s, Pein, Kattner,
  Sch\"{o}n, Diehn, Tanabe, Munk, Zachariae, and Steinem}{Bartsch
  et~al.}{2019}]{Bartsch2019}
Bartsch, A., S.~Llabr{\'{e}}s, F.~Pein, C.~Kattner, M.~Sch\"{o}n, M.~Diehn,
  M.~Tanabe, A.~Munk, U.~Zachariae, and C.~Steinem (2019).
\newblock High-resolution experimental and computational electrophysiology
  reveals weak $\upbeta$-lactam binding events in the porin {PorB}.
\newblock {\em Scientific Reports\/}~{\em 9\/}(1).

\bibitem[\protect\citeauthoryear{Baum and Petrie}{Baum and
  Petrie}{1966}]{baum1966}
Baum, L.~E. and T.~Petrie (1966).
\newblock Statistical inference for probabilistic functions of finite state
  {{Markov}} chains.
\newblock {\em Ann. Math. Stat.\/}~{\em 37\/}(6), 1554--1563.

\bibitem[\protect\citeauthoryear{Baum, Petrie, Soules, and Weiss}{Baum
  et~al.}{1970}]{baum1970}
Baum, L.~E., T.~Petrie, G.~Soules, and N.~Weiss (1970).
\newblock A maximization technique occurring in the statistical analysis of
  probabilistic functions of {{Markov}} chains.
\newblock {\em Ann. Math. Stat\/}~{\em 41\/}(1), 164--171.

\bibitem[\protect\citeauthoryear{Becker, Honerkamp, Hirsch, Fr{\"o}be,
  Schlatter, and Greger}{Becker et~al.}{1994}]{becker1994}
Becker, J.~D., J.~Honerkamp, J.~Hirsch, U.~Fr{\"o}be, E.~Schlatter, and
  R.~Greger (1994).
\newblock Analysing ion channels with hidden {{Markov}} models.
\newblock {\em Pfl\"ugers Arch.\/}~{\em 426\/}(3), 328--332.

\bibitem[\protect\citeauthoryear{Behr, Holmes, and Munk}{Behr
  et~al.}{2018}]{behr2018multiscale}
Behr, M., C.~Holmes, and A.~Munk (2018).
\newblock Multiscale blind source separation.
\newblock {\em Ann. Stat.\/}~{\em 46\/}(2), 711--744.

\bibitem[\protect\citeauthoryear{Bickel, Ritov, and Ryd\'en}{Bickel
  et~al.}{1998}]{BRR1998}
Bickel, P.~J., Y.~Ritov, and T.~Ryd\'en (1998).
\newblock {Asymptotic normality of the maximum-likelihood estimator for general
  hidden Markov models}.
\newblock {\em Ann. Stat.\/}~{\em 26\/}(4), 1614 -- 1635.

\bibitem[\protect\citeauthoryear{Bielecki, Jakubowski, and
  Nieweglowski}{Bielecki et~al.}{2013}]{bielecki2013}
Bielecki, T., J.~Jakubowski, and M.~Nieweglowski (2013).
\newblock {Intricacies of dependence between components of multivariate Markov
  chains: weak Markov consistency and weak Markov copulae}.
\newblock {\em Electron. J. Probab.\/}~{\em 18}, 21 pp.

\bibitem[\protect\citeauthoryear{Brand, Oliver, and Pentland}{Brand
  et~al.}{1997}]{brand1997}
Brand, M., N.~Oliver, and A.~Pentland (1997).
\newblock Coupled hidden {{Markov}} models for complex action recognition.
\newblock In {\em Proc. {{IEEE Comput}}. {{Soc}}. {{Conf}}. {{Comput}}.
  {{Vis}}. {{Pattern Recognit}}.}, pp.\  994--999.

\bibitem[\protect\citeauthoryear{Capp{\'e}, Moulines, and Ryden}{Capp{\'e}
  et~al.}{2005}]{cappe2005}
Capp{\'e}, O., E.~Moulines, and T.~Ryden (2005).
\newblock {\em Inference in Hidden {{Markov}} Models}.
\newblock Springer {{Series}} in {{Statistics}}. {New York}: {Springer-Verlag}.

\bibitem[\protect\citeauthoryear{{Celeux} and Durand}{{Celeux} and
  Durand}{2008}]{celeux2008}
{Celeux}, G. and J.-B. Durand (2008).
\newblock Selecting hidden {Markov} model state number with cross-validated
  likelihood.
\newblock {\em Computational Statistics\/}~{\em 23\/}(4), 541--546.

\bibitem[\protect\citeauthoryear{Chen, Liang, Zhao, Hu, and Tian}{Chen
  et~al.}{2009}]{chen2009a}
Chen, C., J.~Liang, H.~Zhao, H.~Hu, and J.~Tian (2009).
\newblock Factorial {{HMM}} and parallel {{HMM}} for gait recognition.
\newblock {\em IEEE T. Syst. Man. Cy. C\/}~{\em 39\/}(1), 114--123.

\bibitem[\protect\citeauthoryear{Chen, Wasserstrom, and Shiferaw}{Chen
  et~al.}{2009}]{chen2009}
Chen, W., J.~A. Wasserstrom, and Y.~Shiferaw (2009).
\newblock Role of coupled gating between cardiac ryanodine receptors in the
  genesis of triggered arrhythmias.
\newblock {\em Am. J. Physiol. Heart Circ. Physiol.\/}~{\em 297\/}(1),
  H171--180.

\bibitem[\protect\citeauthoryear{Chen, Shen, Shan, and Kou}{Chen
  et~al.}{2016}]{chen2016}
Chen, Y., K.~Shen, S.-O. Shan, and S.~C. Kou (2016).
\newblock Analyzing single-molecule protein transportation experiments via
  hierarchical hidden {{Markov}} models.
\newblock {\em J. Am. Stat. Assoc.\/}~{\em 111\/}(515), 951--966.

\bibitem[\protect\citeauthoryear{Chung, Anderson, and Krishnamurthy}{Chung
  et~al.}{2007}]{chung2007}
Chung, S.-H., O.~S. Anderson, and V.~V. Krishnamurthy (Eds.) (2007).
\newblock {\em Biological Membrane Ion Channels: Dynamics, Structure, and
  Applications}.
\newblock Biological and {{Medical Physics}}, {{Biomedical Engineering}}. {New
  York}: {Springer-Verlag}.

\bibitem[\protect\citeauthoryear{Chung and Kennedy}{Chung and
  Kennedy}{1996}]{chung1996}
Chung, S.-H. and R.~A. Kennedy (1996).
\newblock Coupled {{Markov}} chain model: {{Characterization}} of membrane
  channel currents with multiple conductance sublevels as partially coupled
  elementary pores.
\newblock {\em Math. Biosci.\/}~{\em 133\/}(2), 111--137.

\bibitem[\protect\citeauthoryear{Csisz{\'a}r and Shields}{Csisz{\'a}r and
  Shields}{2000}]{csiszar2000}
Csisz{\'a}r, I. and P.~C. Shields (2000).
\newblock {The consistency of the BIC Markov order estimator}.
\newblock {\em Ann. Stat.\/}~{\em 28\/}(6), 1601 -- 1619.

\bibitem[\protect\citeauthoryear{Dabrowski and McDonald}{Dabrowski and
  McDonald}{1992}]{dabrowski1992}
Dabrowski, A.~R. and D.~McDonald (1992).
\newblock Statistical analysis of multiple ion channel data.
\newblock {\em Ann. Stat.\/}~{\em 20\/}(3), 1180--1202.

\bibitem[\protect\citeauthoryear{de~Gunst, K{\"u}nsch, and Schouten}{de~Gunst
  et~al.}{2001}]{gunst2001}
de~Gunst, M. C.~M., H.~R. K{\"u}nsch, and J.~G. Schouten (2001).
\newblock Statistical analysis of ion channel data using hidden {Markov} models
  with correlated state-dependent noise and filtering.
\newblock {\em J. Am. Stat. Assoc.\/}~{\em 96\/}(455), 805--815.

\bibitem[\protect\citeauthoryear{Diehn, Munk, and Rudolf}{Diehn
  et~al.}{2019}]{diehn2019}
Diehn, M., A.~Munk, and D.~Rudolf (2019).
\newblock Maximum likelihood estimation in hidden {{Markov}} models with
  inhomogeneous noise.
\newblock {\em ESAIM Probab. Stat.\/}~{\em 23}, 492--523.

\bibitem[\protect\citeauthoryear{Fine, Singer, and Tishby}{Fine
  et~al.}{1998}]{fine1998}
Fine, S., Y.~Singer, and N.~Tishby (1998).
\newblock The hierarchical hidden {{Markov}} model: {{Analysis}} and
  applications.
\newblock {\em Mach. Learn.\/}~{\em 32\/}(1), 41--62.

\bibitem[\protect\citeauthoryear{Fredkin and Rice}{Fredkin and
  Rice}{1991}]{fredkin1991}
Fredkin, D.~R. and J.~A. Rice (1991).
\newblock On the superposition of currents from ion channels.
\newblock {\em Philos. Trans. R. Soc. Lond., B, Biol. Sci.\/}~{\em
  334\/}(1271), 347--356.

\bibitem[\protect\citeauthoryear{Gales and Young}{Gales and
  Young}{2008}]{gales2008}
Gales, M. and S.~Young (2008).
\newblock {\em The Application of Hidden {{Markov}} Models in Speech
  Recognition}.
\newblock {Now Publishers Inc}.

\bibitem[\protect\citeauthoryear{Gassiat and Boucheron}{Gassiat and
  Boucheron}{2003}]{gassiat2003}
Gassiat, E. and S.~Boucheron (2003).
\newblock Optimal error exponents in hidden {Markov} models order estimation.
\newblock {\em IEEE Trans. Inf. Theory\/}~{\em 49\/}(4), 964--980.

\bibitem[\protect\citeauthoryear{Ghahramani and Jordan}{Ghahramani and
  Jordan}{1997}]{ghahramani1997}
Ghahramani, Z. and M.~I. Jordan (1997).
\newblock Factorial hidden {{Markov}} models.
\newblock {\em Mach. Learn.\/}~{\em 29\/}(2), 245--273.

\bibitem[\protect\citeauthoryear{Gnanasambandam, Nielsen, Nicolai, Sachs,
  Hofgaard, and Dreyer}{Gnanasambandam et~al.}{2017}]{gnanasambandam2017}
Gnanasambandam, R., M.~S. Nielsen, C.~Nicolai, F.~Sachs, J.~P. Hofgaard, and
  J.~K. Dreyer (2017).
\newblock Unsupervised idealization of ion channel recordings by minimum
  description length: Application to human {{PIEZO1}}-channels.
\newblock {\em Front. Neuroinform.\/}~{\em 11}.

\bibitem[\protect\citeauthoryear{Gottschau}{Gottschau}{1992}]{gottschau1992}
Gottschau, A. (1992).
\newblock Exchangeability in {{multivariate Markov chain models}}.
\newblock {\em Biometrics\/}~{\em 48\/}(3), 751--763.

\bibitem[\protect\citeauthoryear{Guan, Raich, and Wong}{Guan
  et~al.}{2016}]{guan2016}
Guan, X., R.~Raich, and W.-K. Wong (2016).
\newblock Efficient multi-instance learning for activity recognition from time
  series data using an auto-regressive hidden {{Markov}} model.
\newblock In {\em Proceedings of the 33rd {{International Conference}} on
  {{Machine Learning}} - {{Volume}} 48}, {{ICML}}'16, {New York, NY, USA}, pp.\
   2330--2339. {JMLR.org}.

\bibitem[\protect\citeauthoryear{Jula~Vanegas, Behr, and Munk}{Jula~Vanegas
  et~al.}{2022}]{jula2021multiscale}
Jula~Vanegas, L., M.~Behr, and A.~Munk (2022).
\newblock Multiscale quantile segmentation.
\newblock {\em J. Am. Stat. Assoc.\/}~{\em 117\/}(539), 1384--1397.

\bibitem[\protect\citeauthoryear{Keleshian, Edeson, Liu, and Madsen}{Keleshian
  et~al.}{2000}]{keleshian2000a}
Keleshian, A.~M., R.~O. Edeson, G.-J. Liu, and B.~W. Madsen (2000).
\newblock Evidence for cooperativity between nicotinic acetylcholine receptors
  in patch clamp records.
\newblock {\em Biophys. J.\/}~{\em 78\/}(1), 1--12.

\bibitem[\protect\citeauthoryear{Kemeny and Snell}{Kemeny and
  Snell}{1976}]{kemeny1976}
Kemeny, J.~G. and J.~L. Snell (1976).
\newblock {\em Finite {{Markov chains}}: {{With}} a {{new appendix}}
  "{{generalization}} of a {{fundamental matrix}}"}.
\newblock Undergraduate {{Texts}} in {{Mathematics}}. {New York}:
  {Springer-Verlag}.

\bibitem[\protect\citeauthoryear{Khan, Martinac, Madsen, Milne, Yeo, and
  Edeson}{Khan et~al.}{2005}]{khan2005}
Khan, R.~N., B.~Martinac, B.~W. Madsen, R.~K. Milne, G.~F. Yeo, and R.~O.
  Edeson (2005).
\newblock Hidden {{Markov}} analysis of mechanosensitive ion channel gating.
\newblock {\em Math. Biosci.\/}~{\em 193\/}(2), 139--158.

\bibitem[\protect\citeauthoryear{Klein, Timmer, and Honerkamp}{Klein
  et~al.}{1997}]{klein1997}
Klein, S., J.~Timmer, and J.~Honerkamp (1997).
\newblock Analysis of multichannel patch clamp recordings by hidden {{Markov}}
  models.
\newblock {\em Biometrics\/}~{\em 53\/}(3), 870--884.

\bibitem[\protect\citeauthoryear{Krogh, Larsson, {von Heijne}, and
  Sonnhammer}{Krogh et~al.}{2001}]{krogh2001}
Krogh, A., B.~Larsson, G.~{von Heijne}, and E.~L.~L. Sonnhammer (2001).
\newblock Predicting transmembrane protein topology with a hidden {Markov}
  model: Application to complete {{genomes}}.
\newblock {\em J. Mol. Biol.\/}~{\em 305\/}(3), 567--580.

\bibitem[\protect\citeauthoryear{Laver, O'Neill, and Lamb}{Laver
  et~al.}{2004}]{laver2004}
Laver, D.~R., E.~R. O'Neill, and G.~D. Lamb (2004).
\newblock Luminal {{Ca2}}+\textendash regulated {{Mg2}}+ {{inhibition}} of
  {{skeletal RyRs reconstituted}} as {{isolated channels}} or {{coupled
  clusters}}.
\newblock {\em J. Gen. Physiol.\/}~{\em 124\/}(6), 741--758.

\bibitem[\protect\citeauthoryear{Leh{\'e}ricy}{Leh{\'e}ricy}{2019}]{lehericy2019}
Leh{\'e}ricy, L. (2019).
\newblock Consistent order estimation for nonparametric hidden {{Markov}}
  models.
\newblock {\em Bernoulli\/}~{\em 25\/}(1), 464--498.

\bibitem[\protect\citeauthoryear{Manogaran, Vijayakumar, Varatharajan,
  Malarvizhi~Kumar, Sundarasekar, and Hsu}{Manogaran
  et~al.}{2018}]{manogaran2018}
Manogaran, G., V.~Vijayakumar, R.~Varatharajan, P.~Malarvizhi~Kumar,
  R.~Sundarasekar, and C.-H. Hsu (2018).
\newblock Machine learning based big data processing framework for cancer
  diagnosis using hidden {{Markov}} model and {{GM}} clustering.
\newblock {\em Wireless. Pers. Commun.\/}~{\em 102\/}(3), 2099--2116.

\bibitem[\protect\citeauthoryear{Mari, Haton, and Kriouile}{Mari
  et~al.}{1997}]{mari1997}
Mari, J.-F., J.-P. Haton, and A.~Kriouile (1997).
\newblock Automatic word recognition based on second-order hidden {{Markov}}
  models.
\newblock {\em IEEE T. Audio Speech\/}~{\em 5\/}(1), 22--25.

\bibitem[\protect\citeauthoryear{Marx, Gaburj{\'a}kov{\'a}, Gaburjakova,
  Henrikson, Ondrias, and Marks}{Marx et~al.}{2001}]{marx2001}
Marx, S.~O., J.~Gaburj{\'a}kov{\'a}, M.~Gaburjakova, C.~A. Henrikson,
  K.~Ondrias, and A.~R. Marks (2001).
\newblock Coupled gating between cardiac calcium release channels (ryanodine
  receptors).
\newblock {\em Circ. Res.\/}~{\em 88}, 1151--1158.

\bibitem[\protect\citeauthoryear{Meli, Refaat, Dura, Reiken, Wronska, Wojciak,
  Carroll, Scheinman, and Marks}{Meli et~al.}{2011}]{meli2011}
Meli, A., M.~M. Refaat, M.~Dura, S.~Reiken, A.~Wronska, J.~Wojciak, J.~Carroll,
  M.~M. Scheinman, and A.~R. Marks (2011).
\newblock A novel ryanodine receptor mutation linked to sudden death increases
  sensitivity to cytosolic calcium.
\newblock {\em Circ. Res.\/}~{\em 109\/}(3), 281--290.

\bibitem[\protect\citeauthoryear{Mirams, Cui, Sher, Fink, Cooper, Heath,
  McMahon, Gavaghan, and Noble}{Mirams et~al.}{2011}]{mirams2011}
Mirams, G.~R., Y.~Cui, A.~Sher, M.~Fink, J.~Cooper, B.~M. Heath, N.~C. McMahon,
  D.~J. Gavaghan, and D.~Noble (2011).
\newblock Simulation of multiple ion channel block provides improved early
  prediction of compounds' clinical torsadogenic risk.
\newblock {\em Cardiovasc. Res\/}~{\em 91\/}(1), 53--61.

\bibitem[\protect\citeauthoryear{Neukirch, Rudolf, Garcia, and
  Galiana}{Neukirch et~al.}{2019}]{neukirch2019}
Neukirch, M., D.~Rudolf, X.~Garcia, and S.~Galiana (2019).
\newblock Amplitude-phase decomposition of the magnetotelluric impedance
  tensor.
\newblock {\em Geophysics\/}~{\em 84\/}(5), E301--E310.

\bibitem[\protect\citeauthoryear{Pein, Eltzner, and Munk}{Pein
  et~al.}{2021}]{PEM2021}
Pein, F., B.~Eltzner, and A.~Munk (2021).
\newblock Analysis of patchclamp recordings: model-free multiscale methods and
  software.
\newblock {\em European Biophysics Journal\/}~{\em 50}, 187--209.

\bibitem[\protect\citeauthoryear{Pein, {Tecuapetla-Gomez}, Sch\"utte, Steinem,
  and Munk}{Pein et~al.}{2018}]{pein2018a}
Pein, F., I.~{Tecuapetla-Gomez}, O.~M. Sch\"utte, C.~Steinem, and A.~Munk
  (2018).
\newblock Fully automatic multiresolution idealization for filtered ion channel
  recordings: Flickering event detection.
\newblock {\em IEEE Trans. Nanobioscience\/}~{\em 17}, 300--320.

\bibitem[\protect\citeauthoryear{Perkel}{Perkel}{2010}]{perkel2010}
Perkel, J.~M. (2010).
\newblock High-throughput ion channel screening: {{A}} ``patch''-work solution.
\newblock {\em Biotechniques\/}~{\em 48\/}(1), 25--29.

\bibitem[\protect\citeauthoryear{Porta, Diaz-Sylvester, Neumann, Escobar,
  Fleischer, and Copello}{Porta et~al.}{2012}]{PDNEFC2012}
Porta, M., P.~L. Diaz-Sylvester, J.~T. Neumann, A.~L. Escobar, S.~Fleischer,
  and J.~A. Copello (2012).
\newblock Coupled gating of skeletal muscle ryanodine receptors is modulated by
  ca2+, mg2+, and atp.
\newblock {\em American Journal of Physiology-Cell Physiology\/}~{\em
  303\/}(6), C682--C697.

\bibitem[\protect\citeauthoryear{Sakmann and Neher}{Sakmann and
  Neher}{1995}]{sakmann1995b}
Sakmann, B. and E.~Neher (Eds.) (1995).
\newblock {\em Single-{{channel recording}}\/} (Second ed.).
\newblock {Springer US}.

\bibitem[\protect\citeauthoryear{Salvage, Gallant, Beard, Ahmad, Valli, Fraser,
  Huang, and Dulhunty}{Salvage et~al.}{2019}]{salvage2019}
Salvage, S.~C., E.~M. Gallant, N.~A. Beard, S.~Ahmad, H.~Valli, J.~A. Fraser,
  C.~L.-H. Huang, and A.~F. Dulhunty (2019).
\newblock Ion channel gating in cardiac ryanodine receptors from the arrhythmic
  {{RyR2}}-{{P2328S}} mouse.
\newblock {\em J. Cell Sci.\/}~{\em 132\/}(10).

\bibitem[\protect\citeauthoryear{Schmidt-Hieber, Schneider, Staudt, Krajina,
  Aspelmeier, and Munk}{Schmidt-Hieber et~al.}{2021}]{Schmidt2021}
Schmidt-Hieber, J., L.~Schneider, T.~Staudt, A.~Krajina, T.~Aspelmeier, and
  A.~Munk (2021).
\newblock Posterior analysis of n in the binomial (n,p) problem with both
  parameters unknown -- with applications to quantitative nanoscopy.
\newblock {\em Ann. Stat.\/}~{\em 49\/}(6), 3534--3558.

\bibitem[\protect\citeauthoryear{Sherlock, Xifara, Telfer, and Begon}{Sherlock
  et~al.}{2013}]{sherlock2013}
Sherlock, C., T.~Xifara, S.~Telfer, and M.~Begon (2013).
\newblock A coupled hidden {Markov} model for disease interactions.
\newblock {\em J. R. Stat. Soc. C-Appl.\/}~{\em 62\/}(4), 609--627.

\bibitem[\protect\citeauthoryear{Siekmann, Fackrell, Crampin, and
  Taylor}{Siekmann et~al.}{2016}]{siekmann2016}
Siekmann, I., M.~Fackrell, E.~J. Crampin, and P.~Taylor (2016).
\newblock Modelling modal gating of ion channels with hierarchical {Markov}
  models.
\newblock {\em Proc. R. Soc. A.\/}~{\em 472}.

\bibitem[\protect\citeauthoryear{Sin and Kim}{Sin and Kim}{1995}]{sin1995}
Sin, B. and J.~H. Kim (1995).
\newblock Nonstationary hidden {{Markov}} model.
\newblock {\em Signal Process.\/}~{\em 46\/}(1), 31--46.

\bibitem[\protect\citeauthoryear{Staudt, Aspelmeier, Laitenberger, Geisler,
  Egner, and Munk}{Staudt et~al.}{2020}]{staudt2020}
Staudt, T., T.~Aspelmeier, O.~Laitenberger, C.~Geisler, A.~Egner, and A.~Munk
  (2020).
\newblock Statistical {{molecule mounting}} in {{super}}-{{resolution
  fluorescence microscopy}}: {{Towards quantitative nanoscopy}}.
\newblock {\em Stat. Sci.\/}~{\em 35\/}(1), 92--111.

\bibitem[\protect\citeauthoryear{Taur and Frishman}{Taur and
  Frishman}{2005}]{taur2005}
Taur, Y. and W.~Frishman (2005).
\newblock The cardiac ryanodine receptor ({{RyR2}}) and its role in heart
  disease.
\newblock {\em Cardiol. Rev.\/}~{\em 13\/}(3), 142--146.

\bibitem[\protect\citeauthoryear{Touloupou, Finkenstädt, and
  Spencer}{Touloupou et~al.}{2020}]{touloupou2020}
Touloupou, P., B.~Finkenstädt, and S.~E.~F. Spencer (2020).
\newblock Scalable bayesian inference for coupled hidden markov and semi-markov
  models.
\newblock {\em J. Comput. Graph. Stat.\/}~{\em 29\/}(2), 238--249.

\bibitem[\protect\citeauthoryear{{van der Kamp} and {Osgood}}{{van der Kamp}
  and {Osgood}}{2017}]{vanderkamp2017}
{van der Kamp}, W.~S. and N.~D. {Osgood} (2017).
\newblock Multivariate hidden markov models for personal smartphone sensor
  data: time series analysis.
\newblock In {\em 2017 IEEE Int. Conf. Healthc. Inform.}, pp.\  179--188.

\bibitem[\protect\citeauthoryear{Venkataramanan and Sigworth}{Venkataramanan
  and Sigworth}{2002}]{venkataramanan2002}
Venkataramanan, L. and F.~J. Sigworth (2002).
\newblock Applying hidden {{Markov}} models to the analysis of single ion
  channel activity.
\newblock {\em Biophys. J.\/}~{\em 82\/}(4), 1930--1942.

\bibitem[\protect\citeauthoryear{Walker, Kohl, Lehnart, Greenstein, Lederer,
  and Winslow}{Walker et~al.}{2015}]{walker2015}
Walker, M.~A., T.~Kohl, S.~E. Lehnart, J.~L. Greenstein, W.~J. Lederer, and
  R.~L. Winslow (2015).
\newblock On the adjacency matrix of ryr2 cluster structures.
\newblock {\em PLOS Computational Biology\/}~{\em 11\/}(11), 1--21.

\bibitem[\protect\citeauthoryear{Walker, Williams, Kohl, Lehnart, Jafri,
  Greenstein, Lederer, and Winslow}{Walker et~al.}{2014}]{walker2014}
Walker, M.~A., G.~S.~B. Williams, T.~Kohl, S.~E. Lehnart, M.~S. Jafri, J.~L.
  Greenstein, W.~J. Lederer, and R.~L. Winslow (2014).
\newblock Superresolution modeling of calcium release in the heart.
\newblock {\em Biophys. J.\/}~{\em 107\/}(12), 3018--3029.

\bibitem[\protect\citeauthoryear{Westhead and Vijayabaskar}{Westhead and
  Vijayabaskar}{2017}]{westhead2017hidden}
Westhead, D.~R. and M.~Vijayabaskar (2017).
\newblock {\em Hidden {Markov} models: Methods and protocols}.
\newblock Springer.

\bibitem[\protect\citeauthoryear{Williams, Thomas, and George}{Williams
  et~al.}{2018}]{williams2018}
Williams, A.~J., N.~L. Thomas, and C.~H. George (2018).
\newblock The ryanodine receptor: Advances in structure and organization.
\newblock {\em Current Opinion in Physiology\/}~{\em 1}, 1--6.

\bibitem[\protect\citeauthoryear{Yeo, Edeson, Milne, and Madsen}{Yeo
  et~al.}{1989}]{yeo1989}
Yeo, G.~F., R.~O. Edeson, R.~K. Milne, and B.~W. Madsen (1989).
\newblock Superposition properties of independent ion channels.
\newblock {\em Proc. R. Soc. Lond., B, Biol. Sci.\/}~{\em 238\/}(1291),
  155--170.

\bibitem[\protect\citeauthoryear{Yonekura, Beskos, and Singh}{Yonekura
  et~al.}{2021}]{yonekura2021}
Yonekura, S., A.~Beskos, and S.~Singh (2021).
\newblock Asymptotic analysis of model selection criteria for general hidden
  markov models.
\newblock {\em Stoch. Process. their Appl.\/}~{\em 132}, 164--191.

\bibitem[\protect\citeauthoryear{Zhang and Kassam}{Zhang and
  Kassam}{2001}]{zhang2001blind}
Zhang, Y. and S.~A. Kassam (2001).
\newblock Blind separation and equalization using fractional sampling of
  digital communications signals.
\newblock {\em Signal Processing\/}~{\em 81\/}(12), 2591--2608.

\bibitem[\protect\citeauthoryear{Zucchini, MacDonald, and Langrock}{Zucchini
  et~al.}{2017}]{zucchini2017hidden}
Zucchini, W., I.~L. MacDonald, and R.~Langrock (2017).
\newblock {\em Hidden Markov models for time series: An introduction using R}.
\newblock CRC press.

\end{thebibliography}


\end{document}